\pdfoutput=1
\documentclass[nonacm,sigconf]{acmart}
\usepackage[LGR,T1]{fontenc}
\usepackage[utf8]{inputenc}
\usepackage{multirow}

\AtEndPreamble{
    \theoremstyle{acmplain}
    \newtheorem{assumption}[theorem]{Assumption}
    \theoremstyle{acmdefinition}
    \newtheorem{remark}[theorem]{Remark}
}

\usepackage[binary-units]{siunitx}
\DeclareSIUnit{\txn}{txn}
\newcommand{\Replicas}{\mathfrak{R}}
\newcommand{\Clients}{\mathfrak{C}}
\newcommand{\Faulty}{\mathcal{F}}
\newcommand{\n}{\mathbf{n}}
\newcommand{\f}{\mathbf{f}}
\newcommand{\nf}{\mathbf{nf}}
\newcommand{\ID}[1]{\mathop{\textsf{id}}(#1)}
\newcommand{\Replica}[1][r]{\textsc{#1}}
\newcommand{\Primary}[1][p]{\textsc{#1}}
\newcommand{\Client}{\gamma}

\newcommand{\Name}[1]{\textnormal{\textsc{#1}}}
\newcommand{\PoE}{\Name{PoE}}

\newcommand{\Bitcoin}{\Name{Bitcoin}}
\newcommand{\Ethereum}{\Name{Ethereum}}

\newcommand{\PBFT}{\Name{Pbft}}
\newcommand{\FaB}{\Name{FaB}}
\newcommand{\ZZ}{\Name{Zyzzyva}}
\newcommand{\LinBFT}{\Name{LinBFT}}
\newcommand{\SBFT}{\Name{Sbft}}
\newcommand{\HS}{\Name{HotStuff}}
\newcommand{\MinBFT}{\Name{MinBFT}}

\newcommand{\MAC}{\Name{MAC}}

\newcommand{\T}{\tau}
\newcommand{\MName}[1]{\texttt{#1}}
\newcommand{\Message}[2]{\ensuremath{\MName{#1}(#2)}}
\newcommand{\Cert}[2]{\langle #1 \rangle_{#2}}
\newcommand{\PCert}[2]{\langle\!| #1 |\!\rangle_{#2}}
\newcommand{\RCert}[2]{\langle\!| #1 |\!\rangle_{#2}'}
\newcommand{\View}[1][v]{#1}
\newcommand{\rn}{\rho}
\newcommand{\PC}[1]{\mathbb{P}(#1)}
\newcommand{\CC}[1]{\mathbb{C}(#1)}

\newcommand{\RC}[1]{\mathbb{R}(#1)}
\newcommand{\VE}{\mathsf{E}}
\newcommand{\VI}{\mathsf{V}}
\newcommand{\NVP}[2]{\mathsf{P}(#1, #2)}
\newcommand{\NVrn}[1]{{\mathsf{LP}}(#1)}
\newcommand{\NVrncc}[1]{{\mathsf{LC}}(#1)}
\newcommand{\NVm}[2]{\mathsf{M}(#1, #2)}
\newcommand{\NVcc}[1]{\mathsf{C}(#1)}
\newcommand{\Ledger}{\mathcal{L}}
\newcommand{\Digest}[1]{\texttt{digest}(#1)}

\newcommand{\CSize}{\mathbf{C}}
\newcommand{\MSize}{\mathbf{M}}
\newcommand{\WSize}{\mathbf{W}}
\newcommand{\BW}{\mathbf{B}}
\newcommand{\BigO}[1]{\mathcal{O}\mathord{\left(#1\right)}}

\newcommand{\abs}[1]{\lvert #1 \rvert}
\newcommand{\Size}[1]{\lVert #1 \rVert}
\newcommand{\union}{\cup}
\newcommand{\intersect}{\cap}
\newcommand{\difference}{\setminus}
\newcommand{\lfref}[2]{Line~\ref{#1:#2} of Figure~\ref{#1}}
\newcommand{\lsfref}[3]{Lines~\ref{#1:#2}--\ref{#1:#3} of Figure~\ref{#1}}
\newcommand{\subref}[2]{\ref{#1}(\ref{#1:#2})}
\newcommand{\inref}[1]{\textbf{(\ref{#1})}}

\usepackage{algorithmic}
\newcommand{\GETS}{:=}
\newenvironment{myprotocol}{
    \hrule
    \small
    \smallskip
    \algsetup{linenosize=\scriptsize}
    \begin{algorithmic}[1]
        \newenvironment{algenumerate}{\begin{enumerate}}{\end{enumerate}}
        \newcommand{\SPACE}{\item[]}
        \newcommand{\TITLE}[2]{\item[] \textbf{\underline{##1}} (##2) \textbf{:}\\[2pt]}
        \makeatletter
            \newcommand{\EVENT}[1]{\STATE \textbf{event} ##1 \textbf{do}\begin{ALC@g}}
            \newcommand{\ENDEVENT}{\end{ALC@g} \STATE \textbf{end event}}
        \makeatother
}{
    \end{algorithmic}
    \smallskip
    \hrule
}

\usepackage{tikz,pgfplots,pgfplotstable}
\usetikzlibrary{arrows.meta,calc,decorations.pathreplacing,patterns}
\tikzset{
    dot/.append style={circle,scale=0.35,draw=black,fill=black},
    sdot/.append style={scale=0.45,draw=black,fill=black},
    label/.append style={align=center,font=\strut\footnotesize},
    >=Stealth,
    every edge/.append style={semithick},
    thread/.append style={align=center,draw,thick,rectangle,text width=1cm,text height=2ex,text depth=.25ex,minimum height=0.75cm,font=\strut},
}

\definecolor{colA}{RGB}{230,159,0}
\definecolor{colB}{RGB}{86,180,233}
\definecolor{colC}{RGB}{0,158,115}
\definecolor{colD}{RGB}{240,228,66}
\definecolor{colE}{RGB}{0,114,178}
\definecolor{colF}{RGB}{213,94,0}
\definecolor{colG}{RGB}{204,121,167}

\title{On the Correctness of Speculative Consensus}
\thanks{An extended abstract of this work appeared at the 24\textsuperscript{th} International Conference on Extending Database Technology (EDBT 2021)~\cite{edbt_2021}.}

\author{Jelle Hellings}
\affiliation{\department{Department of Computing and Software}
             \institution{McMaster University}
             \streetaddress{1280 Main Street West}
             \city{Hamilton} \state{ON} \country{Canada}}

\author{Suyash Gupta}
\affiliation{\department{RISELab}
             \department{Department of Electrical Engineering and\\ Computer Science}
             \institution{University of California, Berkeley}
             \city{Berkeley} \state{CA} \country{USA}}

\author{Sajjad Rahnama}
\affiliation{\department{Exploratory Systems Lab}
             \department{Department of Computer Science}
             \institution{University of California, Davis}
             \streetaddress{One Shields Avenue}
             \city{Davis} \state{CA} \country{USA}}

\author{Mohammad Sadoghi}
\affiliation{\department{Exploratory Systems Lab}
             \department{Department of Computer Science}
             \institution{University of California, Davis}
             \streetaddress{One Shields Avenue}
             \city{Davis} \state{CA} \country{USA}}

\begin{document}

\begin{abstract}
The introduction of \Bitcoin{} fueled the development of blockchain-based resilient data management systems that are resilient against failures, enable federated data management, and can support data provenance. The key factor determining the performance of such resilient data management systems is the consensus protocol used by the system to replicate client transactions among all participants. Unfortunately, existing high-throughput consensus protocols are costly and impose significant latencies on transaction processing, which rules out their usage in responsive high-performance data management systems.

In this work, we improve on this situation by introducing the \emph{Proof-of-Execution consensus protocol} (\PoE{}), a consensus protocol designed for high-performance low-latency resilient data management. \PoE{} introduces \emph{speculative execution}, which minimizes latencies by starting execution before consensus is reached, and \PoE{} introduces \emph{proof-of-executions} to guarantee successful execution to clients. Furthermore, \PoE{} introduces a single-round check-commit protocol to reduce the overall communication costs of consensus. Hence, we believe that \PoE{} is a promising step towards flexible general-purpose low-latency resilient data management systems.
\end{abstract}

\maketitle

\section{Introduction}
The introduction of the cryptocurrency \Bitcoin{}~\cite{bitcoin} marked the first wide-spread deployment of a \emph{permissionless blockchain}. The emergence of \Bitcoin{} and other blockchains has fueled the development of new resilient data management systems~\cite{hyperledger,blockchaindb,blockmeetdb,caper,blockplane,vldb_2020}. These new systems are attractive for the database community, as they can be used to provide data management systems that are resilient against failures, enable cooperative (federated) data management with many independent parties, and can support data provenance. Due to these qualities, interest in blockchains is widespread and includes applications in health care, IoT, finance, agriculture, and the governance of supply chains for fraud-prone commodities (e.g., such as hardwood and fish)~\cite{blockhealthover,blockhealthfac,blockchain_iot,surviott,blockfood,blockfoodtwo,bceco}.

At their core, blockchain systems are distributed systems in which each participating replica maintains a copy of a ledger that stores an append-only list of all transactions requested by users and executed by the system~\cite{mc_2021}. This ledger is constructed and stored in a tamper-proof manner: changes (e.g., appending new transactions) are made via a \emph{consensus protocol}, which will only allow changes that are supported by the majority of all participants, ruling out malicious changes (e.g., overwriting existing operations) by a minority of faulty participants~\cite{mc_2021,distalgo}. These \emph{consensus protocols} can be seen as generalizations of the well-known two-phase commit~\cite{2pc} and three-phase commit~\cite{3pc} protocols (as uses in traditional replicated databases) toward dealing gracefully with \emph{replica failures} and even \emph{malicious behavior}. As the ledger is replicated over and maintained by all participating replicas, it is highly \emph{resilient} and will survive even if individual participants fail.

The key factor determining the performance of blockchain-based systems is the choice of \emph{consensus protocol}~\cite{mc_2021,pbftj,icdcs}: the consensus protocol determines the \emph{throughput} of the system, as the consensus protocol determines the speed by which transactions are replicated and appended to the ledger of each replica, and the \emph{latency} clients observe on their requested transactions, as the operations necessary to reach consensus on these transactions determine the minimum time it takes for individual replicas to execute the requested transactions and inform clients of the result.

Unfortunately, existing consensus protocols typically focus on either providing high throughput or low latency, thereby failing to provide the combination of high throughputs and low latencies required for responsive high-performance data management systems. First, we observe that the Proof-of-Work style consensus protocols of permissionless blockchains such as \Bitcoin{} and \Ethereum{} suffer from high costs, very low throughputs, and very high latencies, making such permissionless designs impractical for high-performance data management~\cite{hypereal,badcoin,badbadcoin}. Permissioned blockchains, e.g., those based on \PBFT{}-style consensus, are more suitable for high-performance data management: fine-tuned permissioned systems can easily process up-to-hundreds-of-thousands transactions per second, this even in wide-area (Internet) deployments~\cite{icdcs,vldb_2020,mc_2021,pbftj}. Still, even the best permissioned consensus protocols cannot provide the low latencies we are looking for, as all reliable consensus protocols require three-or-more subsequent rounds of internal communication before requests can be executed and clients can be informed.

To further unlock the development of new resilient data management systems, we designed the \emph{Proof-of-Execution consensus protocol} (\PoE{}), a \emph{novel} consensus protocol that is able to provide high throughputs, while minimizing client latencies. At the core of \PoE{} are two innovative techniques:
\begin{enumerate}
\item \PoE{} introduces \emph{speculative execution}: \PoE{} executes transactions requested by clients and informs clients of the result \emph{before} consensus is reached on these transactions, while providing the clients a \emph{proof-of-execution} that guarantees that speculatively-executed transactions will eventually reach consensus; and
\item \PoE{} introduces the \emph{check-commit} protocol, a decentralized single-round protocol that, under normal conditions, can commit consensus decisions for which a \emph{proof-of-execution} exists and can replicate such decisions among all replicas without relying on specific replicas and without requiring several rounds of communication.
\end{enumerate}

By combining these innovative techniques, \PoE{} only imposes \emph{two} rounds of communication on a consensus step before execution can commence and the client can be informed, while only requiring three rounds of communication to complete a consensus step in the normal case. Furthermore, the design of \PoE{} is flexible and allows for optimizations that further balance communication costs, transaction latency, and recovery complexity. E.g., via the usage of digests to reduce communication costs (at the cost of more-complex recovery paths), via the usage of threshold signatures to further reduce communication costs (at the cost of higher latencies), via the usage of message authentication codes to reduce computation costs (at the cost of more-complex recovery paths), and via the usage of out-of-order processing to significantly improve throughput (at the cost of higher resource usage).

In this paper, we not only introduce the design of \PoE{}, but also provide rigorous proofs of the correctness of \emph{all parts} of the protocol. Furthermore, we provide an in-depth analytical and experimental evaluation of \PoE{} in comparison with other contemporary permissioned consensus protocols. In specific, we make the following contributions:

\begin{enumerate}
\item In Section~\ref{sec:cons_serv} we introduce the concept of \emph{speculative execution} and formalize its usage in a client-oriented system that processes transactions via consensus.
\item In Section~\ref{sec:poe}, we provide an in-depth description of all parts of \PoE{}. In specific:
\begin{enumerate}
\item Section~\ref{ss:poe_nc} describes the normal-case operations that are optimized for high-performance low-latency transaction processing utilizing speculative execution and proof-of-execution;
\item Section~\ref{ss:poe_failure} describes the situations in which the normal-case of \PoE{} can fail and the impact this has on the state of individual replicas;
\item Section~\ref{ss:poe_cc} introduces the novel single-round check-commit protocol that allows \PoE{} replicas to recover from minor failures without interrupting the normal-case operations;
\item Section~\ref{ss:poe_vc} introduces the view-change protocol that allows \PoE{} replicas to recover from major failures (including network failures) without invalidating any transactions that have received a proof-of-execution;
\item Section~\ref{ss:poe_proofs} proves the correctness of \emph{all parts} of \PoE{}, showing that \PoE{} provides consensus and maintains all speculatively-executed transactions that have received a proof-of-execution. Due to the level of detail (which includes all modes of operations), this proof of correctness is a major contribution in itself and can be used as a stepping stone in the analysis of other primary-backup consensus protocols; and
\item Section~\ref{ss:poe_client} proves that \PoE{} provides services to clients. In the normal case, \PoE{} does so by providing the client with a \emph{proof-of-execution}, which can be provided with low latency. In the case of failures, \PoE{} can always fall back to a \emph{proof-of-commit}, which takes an additional round of communication to establish. 
\end{enumerate}
\item In Section~\ref{sec:complexity}, we take an in-depth look at the complexity of \PoE{}, we introduce \PoE{} variants that use digests, threshold signatures, or message authentication codes to reduce complexity, and we show that \PoE{} and all its variants can utilize out-of-order processing to maximize performance.
\item In Section~\ref{ss:lin} we introduce \emph{Linear-\PoE{}}, a variant of \PoE{} that uses threshold signatures to make the normal-case and the check-commit protocol fully linear. Central in this variant is a novel \emph{Linear Check-Commit protocol} that utilizes aggregator rotation, aggregated (multi-round) check-commits, and recovery certificates to ensure a single-round decentralized design with only linear communication costs.
\item Finally, in Section~\ref{sec:anal}, we perform an in-depth analytical evaluation of \PoE{} in comparison with other frequently-used consensus protocols.
\end{enumerate}
Furthermore, Section~\ref{sec:prelim} introduces the notation used throughout this paper and Section~\ref{sec:conclusion} concludes on our findings.

\begin{figure*}
\centering
\begin{minipage}{\textwidth}
\centering
\makebox[0pt]{
\begin{tabular}{lccccccc}
\toprule
         &\multicolumn{2}{c}{Communication Rounds}&\multicolumn{2}{c}{Message Complexity\footnote{For readability, we used a simplified notation for the message complexity by omitting off-by-one terms. E.g., \PoE{} requires a total of $(\n-1)\CSize + 2\n(\n-1)$ messages.}}&\\
        \cmidrule(l{4pt}r{4pt}){2-3}\cmidrule(l{4pt}r{4pt}){4-5}
Protocol & Before Execution & Total &Total&Per Replica (max)& Environment\footnote{The \emph{environment} specifies in which environment the protocol can operate. This does not mean that the protocol will be able to make \emph{progress} (new consensus decisions are made), however. Indeed, all these protocols require sufficiently reliable communication to guarantee \emph{progress}. Only the protocols labeled with \emph{asynchronous (recovery)} can recover from any number of periods in which communication is not sufficiently reliable.} & Remarks\\
\midrule
\PoE{}        & 2 & 3   & $\n\CSize + 2\n^2$ & $\n\CSize +  \n$&\begin{tabular}{c}Asynchronous (recovery)\\Partial Synchrony (progress)\end{tabular}&Speculative execution.\\
Linear-\PoE{} & 3 & 5   & $\n\CSize + 4\n$ &$\n\CSize + \n$&\begin{tabular}{c}Asynchronous (recovery)\\Partial Synchrony (progress)\end{tabular}&Speculative execution.\\
\midrule
\PBFT{}       & 3 & 4   & $\n\CSize + 3\n^2$ &$\n\CSize + 2\n$&\begin{tabular}{c}Asynchronous (recovery)\\Partial Synchrony (progress)\end{tabular}&\\
\midrule
\ZZ{}\footnote{This entry only reflects the \emph{optimistic} fast path of \ZZ{}, which cannot deal with any replica failures.  The complexity of the slow path of \ZZ{} is akin that of \PBFT{}.}        & 1 & 1  & $\n\CSize$ & $\n\CSize$&\begin{tabular}{c}Asynchronous (recovery)\\Partial Synchrony (progress)\end{tabular}&\begin{tabular}{c}Requires reliable clients;\\Has vulnerabilities~\cite{zfail,zfailfix}\end{tabular}\\
\SBFT{}\footnote{This entry only reflects the \emph{optimistic} fast path of \SBFT{}, which cannot deal with any replica failures when replicas are either non-faulty or malicious. Furthermore, \SBFT{} utilizes a checkpoint protocol akin to the one used by \PBFT{}. As no explicit description of this checkpoint protocol is provided in the original \SBFT{} paper, we have used the cost of the \PBFT{} checkpoint protocol (complexity-terms related to the checkpoint protocol are subscripted with $\textsc{cp}$).}& 4 & $5 + 1_{\textsc{cp}}$ &$\n\CSize + 4\n + {\n^2}_{\textsc{cp}}$& $\n\CSize + 3$&\begin{tabular}{c}Asynchronous (recovery)\\Partial Synchrony (progress)\end{tabular}&\\
\HS{}\footnote{In the standard configuration of \HS{}, phases of up-to-four consensus decisions are overlapped. Even with this overlapping, each consensus decision has to go through four all-to-one-to-all communication phases. The costs in this table reflect the communication necessary to complete these four phases.}  & 7 & 8   & $\n\CSize + 3\n$ & $\n\CSize + 3$&\begin{tabular}{c}Partial synchrony\\(progress and recovery)\end{tabular}& No out-of-order processing.\\
\MinBFT{}&2&2&$\n\CSize + \n^2$ & $\n\CSize + \n$ &\begin{tabular}{c}Reliable communication\\(progress and recovery)\end{tabular}& \begin{tabular}{c}Requires trusted hardware;\\can tolerate more faulty replicas.\end{tabular}\\
\bottomrule
\end{tabular}}
\end{minipage}
\caption{A cost comparison of the normal-case operations of \PoE{} and other consensus protocols when reaching reaching consensus among $\n$ replicas on a client request with a size bounded by $\CSize$. We refer to Section~\ref{sec:anal} for an in-depth analysis and breakdown of the details in this table.
}\label{fig:anal}
\end{figure*}

 A summary of the comparison between \PoE{} and other high-performance consensus protocols can be found in Figure~\ref{fig:anal}. As one can see, \PoE{} outperforms other high-performance consensus protocols, as \PoE{} combines the lowest latencies with low communication costs. Furthermore, \PoE{} provides high resilience, as it can operate in fully asynchronous environments without any further assumptions on replicas or clients.

An extended abstract of this work appeared at the 24th International Conference on Extending Database Technology (EDBT 2021)~\cite{edbt_2021}. In comparison with that extended abstract, we have added a full presentation of the operations of \PoE{}, introduced the novel single-round check-commit protocol to further improve the performance of \PoE{}, included complete proofs of the correctness of \PoE{}, introduced a new single-round linear check-commit protocol for use in Linear-\PoE{}, and included an in-depth analytical evaluation of \PoE{} in comparison with contemporary consensus protocols.

\section{Preliminaries}\label{sec:prelim}

\paragraph{System Model}
We model a \emph{system} as a tuple $(\Replicas, \Clients)$, in which $\Replicas$ is a set of \emph{replicas} and $\Clients$ is a set of \emph{clients}. We assign each replica $\Replica \in \Replicas$ a unique identifier $\ID{\Replica}$ with $0 \leq \ID{\Replica} < \abs{\Replicas}$. We write $\Faulty \subseteq \Replicas$ to denote the set of \emph{Byzantine replicas} that can behave in arbitrary, possibly coordinated and malicious, ways.

We assume that non-faulty replicas behave in accordance to the protocols they are executing. We do not make any assumptions on clients: all clients can be malicious without affecting \PoE{}. We write $\n = \abs{\Replicas}$, $\f = \abs{\Faulty}$, and $\nf = \n - \f$ to denote the number of replicas, faulty replicas, and non-faulty replicas, respectively.

\paragraph{Communication}
We assume \emph{authenticated communication}: Byzantine replicas are able to impersonate each other, but replicas cannot impersonate non-faulty replicas. Authenticated communication is a minimal requirement to deal with Byzantine behavior. To enforce authenticated communication and simplify presentation, we assume that all messages are \emph{digitally signed} (e.g., via public-key cryptography)~\cite{cryptobook}: every replica and every client $z \in (\Replicas \union \Clients)$ can \emph{sign} arbitrary messages $m$, resulting in a certificate $\Cert{m}{z}$. These certificates are non-forgeable and can be constructed only if $z$ cooperates in constructing them. Based on only the certificate $\Cert{m}{z}$, anyone can verify that $m$ was originally supported by $z$. We refer to Section~\ref{ss:mac} for a discussion on how to eliminate digital signatures from all messages used between replicas in a system.

\paragraph{Consensus}
A \emph{consensus protocol}~\cite{mc_2021,distalgo} coordinates decision making among the replicas $\Replicas$ of a system by providing a reliable ordered replication of \emph{decisions}. To do so, consensus protocols provide the following guarantees:
\begin{description}
    \item[Termination] if non-faulty replica $\Replica \in \Replicas$ makes a $\rn$-th decision, then all non-faulty replicas $\Replica[q] \in \Replicas$ will make a $\rn$-th decision;
    \item[Non-Divergence] if non-faulty replicas $\Replica_1, \Replica_2 \in \Replicas$ make $\rn$-th decisions $D_1$ and $D_2$, respectively, then $D_1 = D_2$ (they make the same $\rn$-th decisions); and
    \item[Non-Triviality] whenever a non-faulty replica $\Replica \in \Replicas$ learns that a decision $D$ needs to be made, then replica $\Replica$ can force consensus on $D$.
\end{description}

In this work, we assume that each decision represents one or more client transactions. Hence, in practice, the \emph{non-triviality} guarantee simply specifies that replicas can force processing of new client requests whenever clients are requesting execution of transactions.

Consensus cannot be solved in environments in which communication is \emph{asynchronous} (e.g., when messages can get lost or be arbitrarily delayed)~\cite{flp}. Even though practical networks are reliable most of the time, they also have periods of failure during which they behave asynchronous. One  way to deal with this is by providing \emph{weak consensus}: weak consensus always guarantees non-divergence, while only guaranteeing termination and non-triviality in periods of \emph{reliable communication} (during which messages are delivered within some unknown bounded delay)~\cite{pbftj}. We assume $\n > 3\f$ ($\nf = \n - \f > 2\f$), a minimal requirement to provide consensus in an asynchronous environment~\cite{mc_2021,byzgenagain,bt}.

\section{From Consensus to Client Services}\label{sec:cons_serv}

In the previous section, we defined consensus. The definition of consensus does not specify how one builds an effective service that clients can use for the execution of their transactions, however. Next, we take a look at how consensus-base systems can provide such client services.

\paragraph{Traditional Execution}
First, we consider traditional consensus-based systems that provide client services. Consider a client $\Client$ requesting the execution of some transaction $\T$. In traditional systems, replicas will \emph{execute} $\T$ as the $\rn$-th transaction and inform $\Client$ of the outcome \emph{after they decided} upon $\T$ (using consensus) as the $\rn$-th decision and after executing all transactions decided upon by preceding decisions.

\begin{example}\label{ex:pbft}
Consider a deployment of the \PBFT{} consensus protocol~\cite{mc_2021,pbftj}. Under normal conditions, \PBFT{} operates via a \emph{primary-backup} design in which a designated replica (the primary) is responsible for proposing client transactions to all other replicas (the backups). The primary does so via a \MName{PrePrepare} message. Next, all replicas exchange their local state to determine whether the primary properly proposed a decision. To do so, all replicas participate in two phases of all-to-all communication.

In the \emph{first phase}, all (non-faulty) replicas that receive a proposal via \MName{PrePrepare} message $m$ broadcast a message \Message{Prepare}{m}. Then each replica $\Replica$ waits until it receives \MName{Prepare} messages identical to \Message{Prepare}{m} from at-least $\nf$ distinct replicas. After receiving these $\nf$ messages, the proposal $m$ is \emph{prepared}.
    
Notice that at-least $\nf - \f = \n-2\f$ \MName{Prepare} messages received by $\Replica$ are sent by non-faulty replicas. Hence, there are at-most $\f$ non-faulty replicas that did not participate in preparing $m$. As such, for any other \MName{PrePrepare} message $m$, replicas will only be able to collect up-to $2\f < \nf$ \MName{Prepare} messages, guaranteeing that only the message $m$ will be prepared at non-faulty replicas.

In the \emph{second phase}, all (non-faulty) replicas that prepared $m$ broadcast a message \Message{Commit}{m}. Then each replica $\Replica$ waits until it receives \MName{Commit} messages identical to \Message{Commit}{m} from at-least $\nf$ distinct replicas. After receiving these $\nf$ messages, the proposal $m$ is \emph{committed}, after which $\Replica$ decides $m$ (and, hence, executes the client transaction proposed by $m$).

In \PBFT{}, a replica $\Replica$ \emph{commits} $m$ when it has a guarantee that $m$ can always be recovered from the state of at-most $\nf$ replicas (e.g., all non-faulty replicas). To see this, consider the at-least $\nf$ \MName{Commit} messages received by $\Replica$ due to which $\Replica$ commits $m$. Of these messages, at-least $\nf - \f = \n-2\f$ are sent by non-faulty replicas. Now consider \emph{any replica} $\Replica[q]$ trying to recover based on the state of any set $C$ of at-least $\nf$ replicas. At-least $\nf - \f = \n-2\f$ of the replicas in $C$ are non-faulty replicas. Let $T = S - \Faulty$ be the non-faulty replicas in $S$ and let $D = C - \Faulty$ be the non-faulty replicas in $C$. We have $\abs{T} > \n-2\f$ and $\abs{D} > \n-2\f$. As we assumed $\n > 3\f$, we must have $(T \intersect D) \neq \emptyset$ (as otherwise, $\abs{T} + \abs{D} \geq 2(\n-2\f)$ and $\abs{T} + \abs{D} \leq \nf = \n-\f$ must hold, which would imply $\n \leq 3\f$). As such, $\Replica[q]$ will be able to recover $m$ from the state of any replica in $T \intersect D$. Hence, after a non-faulty replica commits $m$, there is the guarantee that $m$ can be recovered from the transactions prepared by any set of $\nf$ replicas.

We have sketched this working of \PBFT{} in Figure~\ref{fig:pbft}. Besides the normal-case operations of \PBFT{} outlined above, \PBFT{} also has two recovery mechanisms to recover from primary failures and network failures, namely a checkpoint protocol and a view-change protocol. Crucially, these recovery mechanisms assure that \emph{all} transactions that are ever decided (committed) by a non-faulty replica will eventually be recovered and committed by all non-faulty replicas whenever communication becomes sufficiently reliable.
\begin{figure}[t!]
    \centering
    \begin{tikzpicture}[yscale=0.5,xscale=0.75]
        \draw[thick,draw=black!75] (1.75,   0) edge ++(6.5, 0)
                                   (1.75,   1) edge ++(6.5, 0)
                                   (1.75,   2) edge ++(6.5, 0)
                                   (1.75,   3) edge[blue!50!black!90] ++(6.5, 0);

        \draw[thin,draw=black!75] (2, 0) edge ++(0, 3)
                                  (4,   0) edge ++(0, 3)
                                  (6, 0) edge ++(0, 3)
                                  (8,   0) edge ++(0, 3);

        \node[left] at (1.8, 0) {$\Replica_3$};
        \node[left] at (1.8, 1) {$\Replica_2$};
        \node[left] at (1.8, 2) {$\Replica_1$};
        \node[left] at (1.8, 3) {$\Replica[p]$};

        \path[->] (2, 3) edge (4, 2)
                         edge (4, 1)
                         edge (4, 0)
                           
                  (4, 2) edge (6, 0)
                         edge (6, 1)
                         edge (6, 3)

                  (4, 1) edge (6, 0)
                         edge (6, 2)
                         edge (6, 3)
                         
                  (6, 2) edge (8, 0)
                         edge (8, 1)
                         edge (8, 3)

                  (6, 1) edge (8, 0)
                         edge (8, 2)
                         edge (8, 3);
        
        \node[dot,colA] at (8, 0) {};
        \node[dot,colA] at (8, 1) {};
        \node[dot,colA] at (8, 2) {};
        \node[dot,colA] at (8, 3) {};
        
        \path (8, 3) edge[thick,colA] (8, -1.3);
        \node[label,below right,align=left] at (8, 0) {Decide $\T$\\Execute $\T$};

        \node[label,below,yshift=3pt] at (3, 0) {\vphantom{Decide $\T$}\MName{PrePrepare}};
        \node[label,below,yshift=3pt] at (5, 0) {\vphantom{Decide $\T$}\MName{Prepare}};
        \node[label,below,yshift=3pt] at (7, 0) {\vphantom{Decide $\T$}\MName{Commit}};
    \end{tikzpicture}
    \caption{A schematic representation of the normal-case of \PBFT{}: the primary $\Replica[p]$ proposes transaction $\T$ to all replicas via a \MName{PrePrepare} message. Next, replicas commit to $\T$ via a two-phase all-to-all message exchange. In this example, replica $\Replica_3$ is faulty and does not participate.}\label{fig:pbft}
\end{figure}
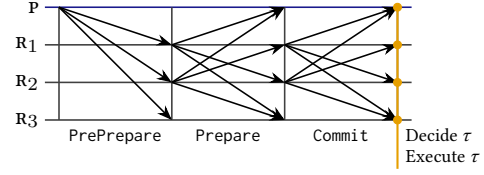
\end{example}

We assume that execution of $\T$ is deterministic: executing $\T$ at any non-faulty replica yields identical outputs when executed upon identical inputs. Using this deterministic nature, the consensus-coordinated replication and execution of transactions will implement a fully-replicated system in which all non-faulty replicas will manage a copy of the same data. Under these assumptions, it is straightforward to deliver \emph{client services}:
\begin{enumerate}
\item To assure that client transactions get executed, a client simply needs to send their transaction $\T$ to any non-faulty replica (whom can then use the non-triviality property of consensus to force a decision on $\T$).
\item To observe the result of execution of a transaction $\T$, the client simply awaits until it receives a single response by any non-faulty replica. As non-faulty replicas execute after they decided on $\T$, non-divergence and termination guarantee that \emph{all} non-faulty replicas will eventually do the same.
\end{enumerate}
When using traditional execution, a client can effectively \emph{detect} whether it received a single response by any non-faulty replica after it received identical responses from at-least $\f+1$ distinct replicas, as at-most $\f$ of those responses can originate from faulty replicas. Using the assumption that execution is deterministic and that all $\n - \f > \f$ non-faulty replicas will eventually execute $\T$, all non-faulty replicas will eventually inform the client with the same identical result. Hence, independent of the behavior of faulty replicas, the client will receive at-least $\f+1$ identical responsed and is able to reliably derive execution results.

\paragraph{Speculative Execution}
Traditional execution assures that it is easy to reason about the operations of a system, both for its replicas (which have strong guarantees during execution) and for the clients (whom can easily derive execution results). This ease-of-use comes at a significant cost: by only executing transactions until replicas are able to make consensus decisions, we significantly delay the latencies clients will observe on their requests, even when the system is operating entirely correctly.

\begin{example}
Consider again the operations of \PBFT{} with traditional execution, as outlined in Example~\ref{ex:pbft}. Let $\T$ be a transaction requested by some client and let $\delta$ be the message delay. At $t$, the primary receives $\T$ and is able to propose $\T$. Assuming that the bandwidth and processing time required to send, receive, and process a message is negligible, these \MName{PrePrepare} proposals will arrive after $t + \delta$ at all other replicas, whom then all broadcast \MName{Prepare} messages. All these \MName{Prepare} messages will arrive after $t + 2\delta$. Only then are all replicas able to broadcast \MName{Commit} messages, which will arrive after $t + 3\delta$. Hence, execution will only happen $3\delta$ after the initial proposal, and only after execution will the client be notified of any outcome.

Modern variants of \PBFT{} such as \SBFT{}~\cite{sbft} and \HS{}~\cite{hotstuff} use \emph{threshold signatures} to replace some phases of all-to-all communication with a quadratic message complexity (e.g., the prepare-phase and the commit-phase), to subphases of all-to-one and one-to-all communication with a linear message complexity each. Such implementations typically trade computational complexity and latency for lower communication costs, and will result in \PBFT{} variants with much higher client latencies. E.g., execution in \SBFT{} happens after $4\delta$ and execution in \HS{} happens after $7\delta$. 

The many phases before execution in these consensus protocols is especially noticeable in practical deployments of consensus: to maximize resilience against disruptions at any location, individual replicas need to be spread out over a wide-area network. Due to this spread-out nature, the message delay will be high and a message delay of $\SI{15}{\milli\second} \leq \delta \leq \SI{200}{\milli\second}$ is not uncommon~\cite{vldb_2020,ahl}.
\end{example}

As an alternative to traditional execution, we propose \emph{speculative execution}: replicas will \emph{execute} $\T$ as the $\rn$-th transaction and inform $\Client$ of the outcome \emph{before they decided} upon $\T$ as the $\rn$-th decision (but still after executing all preceding transactions). As replicas execute transactions before a final consensus decision is made, this introduces two challenges:
\begin{enumerate}
    \item A non-faulty replica $\Replica$ can execute $\T$ as the $\rn$-th transaction only to later make an $\rn$-th decision for another transaction $\T'$. In this case, $\Replica$ needs a way to \emph{rollback} the execution of $\T$ and replace it with an execution of $\T'$.
    \item As non-faulty replicas can rollback their execution, clients can no longer observe the result of execution of a transaction $\T$ via a single response of any non-faulty replica (or $\f+1$ identical responses).
\end{enumerate}
Even with these challenges, speculative execution is worthwhile: when the system operates correctly, speculative execution can greatly reduce the latency clients perceive upon their requests, especially in systems utilizing threshold signatures. In the next section, we introduce the \emph{Proof-of-Execution consensus protocol} (\PoE{}) that utilizes speculative execution and shows methods to overcome both these challenges.

\section{Consensus via Proof-of-Execution}\label{sec:poe}
The \emph{Proof-of-Execution consensus protocol} (\PoE{}) shares the primary-backup design of \PBFT{} and utilizes speculative execution to minimize client latencies in the normal case (when the primary is non-faulty and communication is reliable). To simplify presentation, we will present a non-optimized version of \PoE{}, after which we take an in-depth look at optimizing the complexity of \PoE{} in Section~\ref{sec:complexity}.  

Our presentation of \PoE{} is broken-up in six parts. First, in Section~\ref{ss:poe_nc}, we describe the \emph{normal-case protocol} that is used by the primary to propose consensus decision. Next, in Section~\ref{ss:poe_failure}, we look at how replica and network failures can disrupt the normal-case protocol. Third, in Section~\ref{ss:poe_cc}, we describe the \emph{check-commit protocol} to deal with failures that do not disrupt the normal-case. Then, in Section~\ref{ss:poe_vc}, we describe the \emph{view-change protocol} to deal with failures that disrupt the normal-case. After presenting these three protocols in full detail, we will prove in Section~\ref{ss:poe_proofs} that \PoE{} provides \emph{weak consensus} and we will prove in Section~\ref{ss:poe_client} that \PoE{} provides reliable service to clients.

\subsection{The Normal-Case Protocol}\label{ss:poe_nc}

\PoE{} operates in views and in view $\View$ the replica $\Primary$ with $\ID{\Primary} = \View \bmod \n$ is the \emph{primary} that coordinates the normal-case protocol. 

Consider a client $\Client$ that wants to request transaction $\T$. For now, we assume that $\Client$ knows that $\Primary$ is the current primary, we refer to Section~\ref{ss:poe_client} for the case in which the primary is unknown to $\Client$. To prevent any party to forge requests by client $\Client$, the client $\Client$ signs any transactions it wants to request before sending them to the current primary. Hence, to request $\T$, the client will send $\Cert{\T}{\Client}$ to primary $\Primary$.

After primary $\Primary$ receives $\Cert{\T}{\Client}$, a transaction $\T$ requested and signed by client $\Client$, it can propose $\T$. To \emph{propose $\T$ as the $\rn$-th transaction}, the primary broadcasts a message \Message{Propose}{\Cert{\T}{\Client}, \View, \rn} to all other replicas.

 Next, all replicas exchange their local state to determine whether the primary properly proposed a decision. In \PoE{}, the replicas do so in \emph{one} phase of all-to-all communication. Upon arrival of the \emph{first} proposal for round $\rn$ of view $\View$ via some \MName{Propose} message $m = \Message{Propose}{\Cert{\T}{\Client}, \View, \rn}$, each (non-faulty) replica $\Replica$ that received $m$ will enter the \emph{prepare phase} for $m$. As the first step of the prepare phase, $\Replica$ broadcasts a message \Message{Prepare}{m}. Next, $\Replica$ waits until it receives \MName{Prepare} messages identical to \Message{Propose}{m} from at-least $\nf$ distinct replicas. After receiving these $\nf$ messages, the proposal $m$ is \emph{prepared}.  As a proof of this prepared state for $m$, $\Replica$ stores a \emph{prepared certificate} $\PC{m}$ consisting of these $\nf$ \MName{Prepare} messages.

When replica $\Replica$ prepared proposal $m$, it schedules $\T$ for speculative execution. Whenever all preceding transactions are already executed, $\T$ will be executed by $\Replica$ yielding some result $r$. Next, $\Replica$ will inform the client $\Client$ via an \Message{Inform}{\Cert{\T}{\Client}, \View, \rn, r} of the outcome. Finally, client $\Client$ waits for a \emph{proof-of-execution} for $\Cert{\T}{\Client}$ consisting of identical \Message{Inform}{\Cert{\T}{\Client}, \View, \rn, r} messages from $\nf$ distinct replicas. When $\Client$ receives this proof-of-execution, it considers $\T$ executed. As we shall prove later on, the existence of this $(\View, \rn)$-proof-of-execution for $\Cert{\T}{\Client}$ guarantees that the \emph{speculative} execution of $\T$ will be preserved by all replicas (and will never rollback).

The pseudo-code for this normal-case protocol can be found in Figure~\ref{fig:poe_nc} and an illustration of the working of this protocol can be found in Figure~\ref{fig:poe_nc_ill}.

\begin{figure}[t]
    \begin{myprotocol}
        \TITLE{Client role}{used by client $\Client$ to request transaction $\T$}
        \STATE Send $\Cert{\T}{\Client}$ to the primary $\Primary$.
        \STATE Await a \emph{$(\View, \rn)$-proof-of-execution for $\Cert{\T}{\Client}$} consisting of identical messages \Message{Inform}{\Cert{\T}{\Client}, \View, \rn, r} from $\nf$ distinct replicas.
        \STATE Considers $\T$ executed, with result $r$, as the $\rn$-th transaction.\label{fig:poe_nc:cc}
        \SPACE
        \TITLE{Primary role}{running at the primary $\Primary$ of view $\View$}
        \STATE Let view $\View$ start after execution of the $\rn$-th transaction.
        \WHILE{$\Primary$ is the primary}
            \STATE Await receipt of well-formed client requests $\Cert{\T}{\Client}$.
            \STATE Broadcast \Message{Propose}{\Cert{\T}{\Client}, \View, \rn} to all replicas.\label{fig:poe_nc:propose}
            \STATE $\rn \GETS \rn + 1$.
        \ENDWHILE
        \SPACE
        \TITLE{Backup role}{running at every replica $\Replica \in \Replicas$}\label{fig:poe_nc:pa_backup}
        \EVENT{$\Replica$ receives message $m = \Message{Propose}{\Cert{\T}{\Client}, \View, \rn}$ such that:
            \begin{algenumerate}
                \item $\View$ is the current view;
                \item $m$ is signed by the primary of view $\View$;
                \item $\Replica$ did not prepare a $\rn$-th proposal in view $\View$; and
                \item $\Cert{\T}{\Client}$ is a well-formed client request
            \end{algenumerate}
        }\label{fig:poe_nc:sup}
            \STATE Prepare $m$ as the proposal for round $\rn$ in view $\View$.
            \STATE Broadcast \Message{Prepare}{m} to all replicas.\label{fig:poe_nc:prepare}
        \ENDEVENT
        \EVENT{$\Replica$ receives $\nf$ messages $m_i = \Message{Prepare}{m}$ such that:
            \begin{algenumerate}
                \item $\View$ is the current view;
                \item $m$ is a well-formed proposal \Message{Propose}{\Cert{\T}{\Client}, \View, \rn};
                \item each message $m_i$ is signed by a distinct replica; and
                \item $\Replica$ started the prepare phase for $m$
            \end{algenumerate}
        }\label{fig:poe_nc:prepared}
            \STATE Wait until execution of all rounds preceding $\rn$.\label{fig:poe_nc:wait}
            \STATE Store \emph{prepared certificate} $\PC{m} = \{ m_i \mid 1 \leq i \leq \nf \}$.\label{fig:poe_nc:store}
            \STATE Execute $\T$ as the $\rn$-th transaction, yielding result $r$.\label{fig:poe_nc:execute}
            \STATE Send \Message{Inform}{\Cert{\T}{\Client}, \View, \rn, r} to $\Client$.\label{fig:poe_nc:inform}
        \ENDEVENT
    \end{myprotocol}
    \caption{The normal-case protocol in \PoE{}.}\label{fig:poe_nc}
\end{figure}

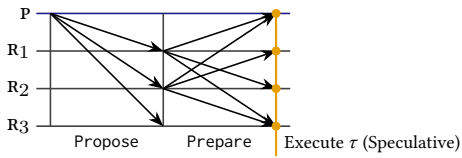
\begin{figure}[t!]
    \centering
    \begin{tikzpicture}[yscale=0.5,xscale=0.75]
        \draw[thick,draw=black!75] (1.75,   0) edge ++(4.5, 0)
                                   (1.75,   1) edge ++(4.5, 0)
                                   (1.75,   2) edge ++(4.5, 0)
                                   (1.75,   3) edge[blue!50!black!90] ++(4.5, 0);

        \draw[thin,draw=black!75] (2, 0) edge ++(0, 3)
                                  (4, 0) edge ++(0, 3)
                                  (6, 0) edge ++(0, 3);

        \node[left] at (1.8, 0) {$\Replica_3$};
        \node[left] at (1.8, 1) {$\Replica_2$};
        \node[left] at (1.8, 2) {$\Replica_1$};
        \node[left] at (1.8, 3) {$\Replica[p]$};

        \path[->] (2, 3) edge (4, 2)
                         edge (4, 1)
                         edge (4, 0)
                           
                  (4, 2) edge (6, 0)
                         edge (6, 1)
                         edge (6, 3)

                  (4, 1) edge (6, 0)
                         edge (6, 2)
                         edge (6, 3);
                                 
        \node[dot,colA] at (6, 0) {};
        \node[dot,colA] at (6, 1) {};
        \node[dot,colA] at (6, 2) {};
        \node[dot,colA] at (6, 3) {};
        
        \path (6, 3) edge[thick,colA] (6, -0.8);
        \node[label,below right,align=left] at (6, 0) {Execute $\T$ (Speculative)\\\phantom{Decide $\T$}};

        \node[label,below,yshift=3pt] at (3, 0) {\vphantom{Execute $\T$}\MName{Propose}};
        \node[label,below,yshift=3pt] at (5, 0) {\vphantom{Execute $\T$}\MName{Prepare}};
    \end{tikzpicture}
    \caption{A schematic representation of the normal-case protocol of \PoE{}: the primary $\Replica[p]$ proposes transaction $\T$ to all replicas via a \MName{Propose} message. Next, replicas prepare $\T$ via a one-phase all-to-all message exchange. Notice that replicas do not explicitly \emph{decide} on $\T$ in the normal case, but do \emph{execute} $\T$. In this example, replica $\Replica_3$ is faulty and does not participate.}\label{fig:poe_nc_ill}
\end{figure}

The correctness of \PoE{} is based on the following properties of the normal-case protocol:

\begin{theorem}\label{thm:main_nc}
Round $\rn$ of view $\View$ of the normal-case protocol of \PoE{} satisfies the following two properties.
\begin{enumerate}
\item\label{thm:main_nc:fst} If non-faulty replicas $\Replica_i$, $i \in \{1, 2\}$, prepared proposals $m_i = \Message{PrePrepare}{\Cert{\T_i}{\Client_i}, \View, \rn}$, then $m_1 = m_2$. 
\item\label{thm:main_nc:snd}  If a non-faulty primary $\Primary$ proposes $m = \Message{Propose}{\Cert{\T}{\Client}, \View, \rn}$, communication is reliable, transaction execution is deterministic, and all non-faulty replicas executed the same sequence of $\rn-1$ transactions, then $\Client$ will receive a $(\View, \rn)$-proof-of-execution for $\Cert{\T}{\Client}$.
\end{enumerate}
\end{theorem}
\begin{proof} We prove the two statements separately.

We prove the first statement by contradiction. Assume that non-faulty replicas $\Replica_1$ and $\Replica_2$ prepared proposals $m_1$ and $m_2$ with $m_1 \neq m_2$. As $\Replica_i$, $i \in \{1,2\}$, prepared proposal $m_i$, it must have received messages \Message{Prepare}{m_i} from $\nf$ distinct replicas (\lfref{fig:poe_nc}{prepared}). Let $S_i$  be the set of $\nf$ replicas from which $\Replica_i$ received these messages and let $T_i = S_i \difference \Faulty$ be the non-faulty replicas in $S_i$. By construction, we have $\abs{T_i} \geq \nf - \f$. As each non-faulty replica only sends \MName{prepare} messages for a \emph{single} proposal in round $\rn$ of view $\View$ (\lfref{fig:poe_nc}{sup}), $m_1 \neq m_2$ implies that $S_1 \intersect S_2 = \emptyset$. Hence, we must have $\abs{S_1 \union S_2} = \abs{S_1} + \abs{S_2} \geq 2(\nf - \f)$. As all replicas in $S_1 \union S_2$ are non-faulty, we must also have $\abs{S_1 \union S_2} \leq \nf$. Hence, we must have $2(\nf - \f) \leq \nf$, which implies $\nf \leq 2\f$. As $\n = \nf + \f$, this implies $\n \leq 3\f$, a contradiction. Consequently, we conclude that $m_1 = m_2$ must hold.

Next, we prove the second statement. A non-faulty primary $\Primary$ will broadcast $m$ to all non-faulty replicas (\lfref{fig:poe_nc}{propose}). As communication is reliable, all $\nf$ non-faulty replicas will receive $m$ as the first proposal of round $\rn$ of view $\View$ (\lfref{fig:poe_nc}{sup}) and broadcast a message \Message{Prepare}{m}. As communication is reliable, all $\nf$ non-faulty replicas will receive \Message{Prepare}{m} from these $\nf$ non-faulty replicas (\lfref{fig:poe_nc}{prepared}) and execute $\T$ (\lfref{fig:poe_nc}{execute}). As all non-faulty replicas executed the same sequence of $\rn-1$ transactions before executing $\T$, each non-faulty replica has the same state before executing $\T$. Consequently, due to deterministic execution of $\T$, all non-faulty replicas will obtain the same result $r$ from execution of $\T$ and send the same \Message{Inform}{\Cert{\T}{\Client}, \View, \rn, r} to $\Client$ (\lfref{fig:poe_nc}{inform}).
\end{proof}

\PoE{} only requires \emph{two} phases of communication before execution commences, which is one less phase than \PBFT{}. Consequently, the processing time of transactions within \PoE{} is sharply reduced from at-least $3\delta$ to $2\delta$, which will also reduce the latency clients perceive. Finally, the elimination of a phase of communication eliminates one round of messages, reducing the bandwidth cost for the normal-case of \PoE{} (and, hence, allowing for an increase in throughput). As shown in Theorem~\subref{thm:main_nc}{snd}, the elimination of a phase in \PoE{} does \emph{not} affect the service that is provided under normal conditions: clients still have a strong guarantee of service whenever the primary is non-faulty and communication is reliable.

The normal-case protocol of \PoE{} provides only few guarantees, however. If non-faulty replica $\Replica$ speculatively executes some transaction $\T$ proposed via $m = \Message{Propose}{\Cert{\T}{\Client}, \View, \rn}$, then, due to Theorem~\subref{thm:main_nc}{fst}, $\Replica$ \emph{only} has the guarantee that no other transaction $\T'$ is executed in round $\rn$ of $\View$. In specific, $\Replica$ has no guarantees that any other non-faulty replicas prepared $m$ or executed $\T$ and has no guarantees that $\Client$ received a $(\View, \rn)$-proof-of-execution for $\Cert{\T}{\Client}$. As we shall see in Section~\ref{ss:poe_failure} and Section~\ref{ss:poe_vc}, a subsequent failure of view $\View$ can easily lead to a situation in which the proposal $m$ is \emph{not preserved}, due to which $\Replica$ needs to rollback $\T$. As we shall show in the remainder of Section~\ref{sec:poe}, the design of \PoE{} does provide the strong guarantee that $m$ will always be guaranteed \emph{if} a proof-of-execution for $m$ \emph{could} be received by client $\Client$. To provide this strong guarantee (and to deal with certain kinds of failures), non-faulty replicas rely on the check-commit protocol of Section~\ref{ss:poe_cc} and the view-change protocol of Section~\ref{ss:poe_vc}. Next, we look at the operations of \PoE{} during failures: we look at how \PoE{} deals with faulty (and possibly malicious) primaries and how \PoE{} recovers from periods of unreliable communication.

\subsection{Failure of the Normal-Case Protocol}\label{ss:poe_failure}

The normal-case protocol described in Section~\ref{ss:poe_nc} is designed to efficiently make consensus decisions and provide clients with proof-of-executions when operating under normal conditions. If the normal conditions are not met, then the normal-case protocol can fail in several ways, each following directly from the conditions stipulated in Theorem~\subref{thm:main_nc}{snd}:

\begin{example}\label{ex:failure}
Consider round $\rn$ of view $\View$ in a deployment of \PoE{}. The normal-case protocol of \PoE{} can be disrupted in round $\rn$ in the following ways:
\begin{enumerate}
    \item Any \emph{non-primary faulty replica} can behave malicious by not participating or by sending invalid messages. Under normal conditions, this will not disrupt the normal-case protocol, however, as Theorem~\subref{thm:main_nc}{snd} does not depend on the behavior of any faulty replicas (that are not the primary).
    \item\label{ex:failure:prim} A \emph{malicious primary} can choose to send different proposals for round $\rn$ of view $\View$ to different non-faulty replicas or can choose to send no proposals to some non-faulty replicas. Due to Theorem~\subref{thm:main_nc}{fst}, at-most one proposal will be prepared in round $\rn$ of view $\View$, this independent of the behavior of any faulty replicas. Hence, this malicious behavior can have only two outcomes:
    \begin{enumerate}
        \item\label{ex:failure:prim:dark}  A proposal $m = \Message{Propose}{\Cert{\T}{\Client}, \View, \rn}$ will be prepared by some non-faulty replicas. The non-faulty replicas that did not receive $m$ are \emph{left in the dark}, as they will not be able to prepare in round $\rn$ of view $\View$ and, consequently, are stuck. In this case, the client can still receive a \emph{proof-of-execution}: if at-least $\nf - \f$ non-faulty replicas prepared $m$, then the $\f$ faulty replicas can choose to also execute $\T$ and send \MName{Inform} messages to the client, this to assure the client receives $\nf$ identical \MName{Inform} messages.
        \item\label{ex:failure:prim:none} No proposal will be prepared by non-faulty replicas, in which case the primary \emph{disrupts the progress} of the normal-case protocol and, consequently, prevents any transactions from being executed.
    \end{enumerate}
    \item Due to \emph{unreliable communication}, messages can get lost (or arbitrarily delayed due to which receiving replicas consider them lost). Consequently, unreliable communication can prevent replicas from receiving proposals even from non-faulty primaries. Furthermore, unreliable communication can prevent replicas from receiving sufficient \MName{Prepare} messages to finish their prepare phases.
\end{enumerate}
\end{example}

Using digital signatures, it is rather straightforward to detect that a primary is sending conflicting proposals for a given round $\rn$ of view $\View$, as non-faulty replicas forward these conflicting proposals to each other during the prepare phase. Unfortunately, this  does not extend to other malicious behavior, as we shall show next.

\begin{example}\label{ex:distin}
Consider a system with $\Replicas = \{ \Primary, \Replica_1, \Replica_2, \Replica_3 \}$ such that $\Primary$ is the current primary of some view $\View$. We consider the following three cases:
\begin{enumerate}
    \item The primary $\Primary$ is faulty and does \emph{not} send any proposal to $\Replica_3$. Hence, eventually $\Replica_3$  detects a primary failure. To alert all other replicas of this failure, $\Replica_3$ broadcasts message \Message{Failure}{\View}.
    \item The primary $\Primary$ is non-faulty and sends a proposal to $\Replica_3$. Unfortunately, $\Replica_3$ is faulty and pretends that the primary failed to propose. Consequently, $\Replica_3$ broadcasts message \Message{Failure}{\View}.
    \item The primary $\Primary$ is non-faulty and sends a proposal to $\Replica_3$. Unfortunately, communication is unreliable and this proposal is lost. This loss is interpreted by $\Replica_3$ as a failure of the primary to propose. To alert all other replicas of this failure, $\Replica_3$ broadcasts message \Message{Failure}{\View}.
\end{enumerate}
We have sketched these three cases in Figure~\ref{fig:poe_nc_failures}. As one can see, the replicas $\Replica_1$ and $\Replica_2$ receive the \emph{exact same set of messages} in all three cases and, hence, observe identical behavior and cannot distinguish between the three cases.
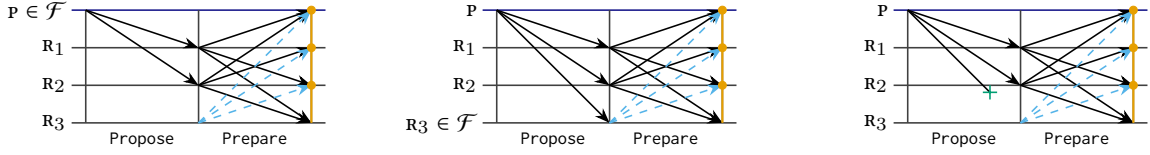
\begin{figure*}[t!]
    \centering
    \begin{tikzpicture}[yscale=0.5,xscale=0.75]
        \draw[thick,draw=black!75] (1.75,   0) edge ++(4.5, 0)
                                   (1.75,   1) edge ++(4.5, 0)
                                   (1.75,   2) edge ++(4.5, 0)
                                   (1.75,   3) edge[blue!50!black!90] ++(4.5, 0);

        \draw[thin,draw=black!75] (2, 0) edge ++(0, 3)
                                  (4,   0) edge ++(0, 3)
                                  (6, 0) edge ++(0, 3);

        \node[left] at (1.8, 0) {$\Replica_3$};
        \node[left] at (1.8, 1) {$\Replica_2$};
        \node[left] at (1.8, 2) {$\Replica_1$};
        \node[left] at (1.8, 3) {$\Replica[p] \in \Faulty$};

        \path[->] (2, 3) edge (4, 2)
                         edge (4, 1)
                           
                  (4, 2) edge (6, 0)
                         edge (6, 1)
                         edge (6, 3)

                  (4, 1) edge (6, 0)
                         edge (6, 2)
                         edge (6, 3)
                           
                  (4, 0)   edge[colB,dashed] (6, 1)
                           edge[colB,dashed] (6, 2)
                           edge[colB,dashed] (6, 3);
                                 
        \node[dot,colA] at (6, 1) {};
        \node[dot,colA] at (6, 2) {};
        \node[dot,colA] at (6, 3) {};
        
        \path (6, 3) edge[thick,colA] (6, -0);

        \node[label,below,yshift=3pt] at (3, 0) {\vphantom{Execute $\T$}\MName{Propose}};
        \node[label,below,yshift=3pt] at (5, 0) {\vphantom{Execute $\T$}\MName{Prepare}};
    \end{tikzpicture}\qquad\quad
    \begin{tikzpicture}[yscale=0.5,xscale=0.75]
        \draw[thick,draw=black!75] (1.75,   0) edge ++(4.5, 0)
                                   (1.75,   1) edge ++(4.5, 0)
                                   (1.75,   2) edge ++(4.5, 0)
                                   (1.75,   3) edge[blue!50!black!90] ++(4.5, 0);

        \draw[thin,draw=black!75] (2, 0) edge ++(0, 3)
                                  (4, 0) edge ++(0, 3)
                                  (6, 0) edge ++(0, 3);

        \node[left] at (1.8, 0) {$\Replica_3 \in \Faulty$};
        \node[left] at (1.8, 1) {$\Replica_2$};
        \node[left] at (1.8, 2) {$\Replica_1$};
        \node[left] at (1.8, 3) {$\Replica[p]$};

        \path[->] (2, 3) edge (4, 2)
                         edge (4, 1)
                         edge (4, 0)
                           
                  (4, 2) edge (6, 0)
                         edge (6, 1)
                         edge (6, 3)

                  (4, 1) edge (6, 0)
                         edge (6, 2)
                         edge (6, 3)
                           
                  (4, 0)  edge[colB,dashed] (6, 1)
                          edge[colB,dashed] (6, 2)
                          edge[colB,dashed] (6, 3);
                                 
        \node[dot,colA] at (6, 1) {};
        \node[dot,colA] at (6, 2) {};
        \node[dot,colA] at (6, 3) {};
        
        \path (6, 3) edge[thick,colA] (6, -0);

        \node[label,below,yshift=3pt] at (3, 0) {\vphantom{Execute $\T$}\MName{Propose}};
        \node[label,below,yshift=3pt] at (5, 0) {\vphantom{Execute $\T$}\MName{Prepare}};
    \end{tikzpicture}\qquad\quad
    \begin{tikzpicture}[yscale=0.5,xscale=0.75]
        \draw[thick,draw=black!75] (1.75,   0) edge ++(4.5, 0)
                                   (1.75,   1) edge ++(4.5, 0)
                                   (1.75,   2) edge ++(4.5, 0)
                                   (1.75,   3) edge[blue!50!black!90] ++(4.5, 0);

        \draw[thin,draw=black!75] (2, 0) edge ++(0, 3)
                                  (4, 0) edge ++(0, 3)
                                  (6, 0) edge ++(0, 3);

        \node[left] at (1.8, 0) {$\Replica_3$};
        \node[left] at (1.8, 1) {$\Replica_2$};
        \node[left] at (1.8, 2) {$\Replica_1$};
        \node[left] at (1.8, 3) {$\Replica[p]$};

        \path[->] (2, 3) edge (4, 2)
                         edge (4, 1)
                         edge[shorten >=0.5cm,-{Rays[colC]}] (4, 0)
                           
                  (4, 2) edge (6, 0)
                         edge (6, 1)
                         edge (6, 3)

                  (4, 1) edge (6, 0)
                         edge (6, 2)
                         edge (6, 3)
                           
                  (4, 0)  edge[colB,dashed] (6, 1)
                          edge[colB,dashed] (6, 2)
                          edge[colB,dashed] (6, 3);
                                 
        \node[dot,colA] at (6, 1) {};
        \node[dot,colA] at (6, 2) {};
        \node[dot,colA] at (6, 3) {};
        
        \path (6, 3) edge[thick,colA] (6, -0);
        \node[left] at (1.8, 0) {\phantom{$\Replica_3 \in \Faulty$}};

        \node[label,below,yshift=3pt] at (3, 0) {\vphantom{Execute $\T$}\MName{Propose}};
        \node[label,below,yshift=3pt] at (5, 0) {\vphantom{Execute $\T$}\MName{Prepare}};
    \end{tikzpicture}
    \caption{A schematic representation of three failures in the normal-case protocol of \PoE{}. \emph{Left}, a faulty primary $\Primary$ does not propose to $\Replica_3$. \emph{Middle}, a faulty replica $\Replica_3$ pretends that the primary $\Primary$ did not propose. \emph{Right}, unreliable communication prevents the delivery of a proposal to $\Replica_3$. In all three cases, $\Replica_3$ alerts other replicas of failure via a \MName{Failure} message (dashed arrow), while the replicas $\Replica_1$ and $\Replica_2$ observe identical behavior.}\label{fig:poe_nc_failures}
\end{figure*}
\end{example}

As Example~\ref{ex:failure} illustrates, any disruption of the normal-case protocol of \PoE{} is caused by a faulty primary or by unreliable communication. If communication is unreliable, then there is no way to guarantee continuous service~\cite{flp}. Hence, replicas assume failure of the current primary if the normal-case protocol  is disrupted, while the design of \PoE{} guarantees that unreliable communication does not affect the correctness of \PoE{} and that the normal-case protocol of \PoE{} will be able to recover when communication becomes reliable.

As Example~\subref{ex:failure}{prim} illustrated, a faulty primary can cause two kinds of disruptions. First, the primary can leave \emph{replicas in the dark} without disrupting the progress of the normal-case protocol (Example~\subref{ex:failure}{prim:dark}) and we use the \emph{check-commit protocol} of Section~\ref{ss:poe_cc} to deal with such behavior. Second, the behavior of the primary can \emph{disrupt the progress} of the normal-case protocol (Example~\subref{ex:failure}{prim:none}) and we use the \emph{view-change protocol} of Section~\ref{ss:poe_vc} to deal with such behavior.

\subsection{The Check-Commit Protocol}\label{ss:poe_cc}

The main  purpose of the check-commit protocol is to \emph{commit} proposals $m = \Message{Propose}{\Cert{\T}{\Client}, \View, \rn}$, after which non-faulty replicas that speculatively executed $\T$ have the guarantee that $\T$ will never rollback. Furthermore, we use the check-commit protocol to assure that non-fauulty replicas cannot be left in the dark: the check-commit protocol assures that all non-faulty replicas receive prepared certificates for $m$ if any replica can commit $m$.

For non-faulty replicas that successfully prepared via the normal-case protocol of Section~\ref{ss:poe_nc}, the check-commit protocol operates in a \emph{single phase} of communication. Let $m = \Message{Propose}{\Cert{\T}{\Client}, \View, \rn}$ be some proposal. The non-faulty replica $\Replica$ uses the check-commit protocol for $m$ to indicate its willingness to commit to $m$ and to determine whether a sufficient number of other replicas are willing to commit $m$. Replica $\Replica$ is willing to commit to $m$ if the following conditions are met:
\begin{enumerate}
    \item $\Replica$ prepared $m$ and speculatively executed $\T$ and, hence, stored a prepared certificate for $m$ and informed the client;
    \item $\Replica$ committed in all rounds before round $\rn$;\footnote{To maximize throughput, one can choose for a design with out-of-order commit steps (in which round $\rn+1$ can commit before round $\rn$). We have omitted such a design in this presentation of \PoE{}, as it would complicate view-changes and substantially increase the complexity of the correctness proofs of \PoE{}.} and
    \item $\Replica$ is currently still in view $\View$.
\end{enumerate}
If these conditions are met, then $\Replica$ starts the process to commit to $m$ by broadcasting a message \Message{CheckCommit}{\PC{m}} with $\PC{m}$ the prepared certificate for $m$ stored by $\Replica$. Then each replica $\Replica$ waits until it receives well-formed \MName{CheckCommit} messages for proposal $m$ from at-least $\nf$ distinct replicas. After receiving these $\nf$ messages, the proposal $m$ is \emph{committed}. As a proof of this committed state for $m$, $\Replica$ stores a \emph{commit certificate} $\CC{m}$ consisting of these $\nf$ \MName{CheckCommit} messages.

Let $m = \Message{Propose}{\Cert{\T}{\Client}, \View, \rn}$ be a proposal. If any replica receives a well-formed message \Message{CheckCommit}{\PC{m}}, then it can use the provided prepared certificate $\PC{m}$ for $m$ to prepare $m$ and execute $\T$ (if it has not yet done so).\footnote{If digital signatures are not used on \MName{Prepare} messages, then prepared certificates cannot be reliably forwarded. In that case, a replica needs to receive $\nf - \f > \f$ identical \MName{CheckCommit} messages for proposal $m$ from distinct replicas before it can use the provided information to prepare $m$.  We refer to Section~\ref{ss:fine} and Section~\ref{ss:mac} for further details. We note that if communication is reliable, then a replica is guaranteed to receive $\nf -\f$ identical \MName{CheckCommit} messages \emph{unless} the behavior of the primary disrupts the progress of the normal-case protocol, in which case a view-change will happen (see Section~\ref{ss:poe_vc}).}

The pseudo-code for this check-commit protocol can be found in Figure~\ref{fig:poe_cc} and an illustration of the working of this protocol can be found in Figure~\ref{fig:poe_cc_ill}.

\begin{figure}[t]
    \begin{myprotocol}
        \TITLE{Check-commit role}{running at every replica $\Replica \in \Replicas$}
        \EVENT{$\Replica$ prepared $m = \Message{Propose}{\Cert{\T}{\Client}, \View, \rn}$ and executed $\T$}
            \STATE Wait until all previous rounds have a commit certificate.\label{fig:poe_cc:wait}
            \IF{$\View$ is the current view}
                \STATE Let $\PC{m}$ be the prepared certificate stored for $m$.
                \STATE Broadcast \Message{CheckCommit}{\PC{m}} to all replicas.\label{fig:poe_cc:broad}
            \ENDIF
        \ENDEVENT
        \EVENT{$\Replica$ receives a well-formed message \Message{CheckCommit}{\PC{m}}}\label{fig:poe_cc:prepare}
            \STATE Let $m = \Message{Propose}{\Cert{\T}{\Client}, \View', \rn'}$ and $\PC{m} = \{ m_1, \dots, m_{\nf} \}$.
            \IF{$\View = \View'$ and $\Replica$ did not prepare $m$}
                \STATE Prepare $m$ and execute $\T$ using \lsfref{fig:poe_nc}{store}{inform} with \MName{Prepare} messages $m_1, \dots, m_{\nf}$.\label{fig:poe_cc:prepare_exec}
            \ENDIF
        \ENDEVENT
        \EVENT{$\Replica$ receives $\nf$ messages $m_i = \Message{CheckCommit}{\PC{m}_i}$ such that:
            \begin{algenumerate}
                \item $m$ is a well-formed proposal \Message{Propose}{\Cert{\T}{\Client}, \View, \rn};
                \item $\View$ is the current view;
                \item $\PC{m}_i$ is a well-formed prepared certificate for $m$;
                \item each message is signed by a distinct replica; and
                \item $\Replica$ prepared $m$
            \end{algenumerate}
        }\label{fig:poe_cc:cped}
            \STATE Wait until all previous rounds have a commit certificate.\label{fig:poe_cc:wait_store}
            \STATE Store \emph{commit certificate} $\CC{m} = \{ \PC{m}_i \mid 1 \leq i \leq \nf \}$.\label{fig:poe_cc:store}
        \ENDEVENT
    \end{myprotocol}
    \caption{The check-commit protocol in \PoE{}.}\label{fig:poe_cc}
\end{figure}

\begin{figure}[t!]
    \centering
    \begin{tikzpicture}[yscale=0.5,xscale=0.75]
        \draw[thick,draw=black!75] (1.75,   0) edge ++(4.7, 0)
                                   (1.75,   1) edge ++(4.7, 0)
                                   (1.75,   2) edge ++(4.7, 0)
                                   (1.75,   3) edge ++(4.7, 0);

        \draw[thin,draw=black!75] (2, 0) edge ++(0, 3)
                                  (4.2, 0) edge ++(0, 3)
                                  (6.2, 0) edge ++(0, 3);

        \node[left] at (1.8, 0) {$\Replica_4$};
        \node[left] at (1.8, 1) {$\Replica_3$};
        \node[left] at (1.8, 2) {$\Replica_2$};
        \node[left] at (1.8, 3) {$\Replica_1$};

        \path[->] (2, 2) edge (4, 3)
                         edge (4, 1)
                         edge (4, 0)
                           
                  (2.2, 3) edge (4.2, 0)
                           edge (4.2, 1)
                           edge (4.2, 2)

                  (2.8, 1) edge (4.8, 0)
                           edge (4.8, 2)
                           edge (4.8, 3)

                  (4.2, 0) edge[dashed] (6.2, 1)
                           edge[dashed] (6.2, 2)
                           edge[dashed] (6.2, 3);

        \node[dot,colF] at (2.2, 3) {};
        \node[dot,colF] at (2, 2) {};
        \node[dot,colF] at (2.8, 1) {};
        \node[dot,colF] at (4.2, 0) {};

        \node[dot,colA] at (4.8, 3) {};
        \node[dot,colA] at (4.8, 2) {};
        \node[dot,colA] at (4.2, 1) {};
        \node[dot,colA] at (4.8, 0) {};

        \node[dot,colF]  (la) at (8, 2) {};
        \node[right=4pt] at (la) {Prepare and execute $\T$.};

        \node[dot,colA]  (la) at (8, 1) {};
        \node[right=4pt] at (la) {Commit $\T$ (Decide).};

        \node[label,below,yshift=3pt] at (3.1, 0) {\vphantom{Execute $\T$}\MName{CheckCommit}};
        \node[label,below,yshift=3pt] at (5.2, 0) {\vphantom{Execute $\T$}\MName{CheckCommit}};
    \end{tikzpicture}
    \caption{A schematic representation of the check-commit protocol of \PoE{}. In this illustration, replicas $\Replica_1$, $\Replica_2$, and $\Replica_3$ initiate the check-commit protocol due to preparing and executing a transaction $\T$ (which they finish at different times), whereas $\Replica_4$ learns $\T$ via the \Name{CheckCommit} message it receives from replicas $\Replica_1$ and $\Replica_2$. The replicas explicitly \emph{decide} on $\T$ upon finishing the protocol.}\label{fig:poe_cc_ill}
\end{figure}
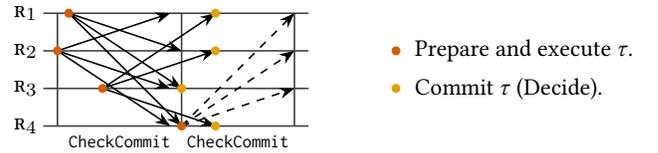

The correctness of \PoE{} is based on the following properties of the check-commit protocol:

\begin{theorem}\label{thm:main_cc}
Assume that communication is reliable, transaction execution is deterministic, and all non-faulty replicas are in view $\View$. Round $\rn$ of view $\View$ of the check-commit protocol of \PoE{} satisfies the following three properties:
\begin{enumerate}
\item\label{thm:main_cc:fst} if a client $\Client$ receives a $(\View, \rn)$-proof-of-execution for $\Cert{\T}{\Client}$, then all non-faulty replicas will prepare some proposal proposing $\Cert{\T}{\Client}$ in round $\rn$;
\item\label{thm:main_cc:snd} if all non-faulty replicas executed the same sequence of $\rn-1$ transactions and a non-faulty replica stored a commit certificate for proposal $m =\Message{Propose}{\Cert{\T}{\Client}, \View, \rn}$, then all non-faulty replicas will store a commit certificate for $m$ and $\Client$ will receive a $(\View, \rn)$-proof-of-execution for $\Cert{\T}{\Client}$; and
\item\label{thm:main_cc:trd} if there exists a commit certificate $\CC{m}$ for proposal $m =\Message{Propose}{\Cert{\T}{\Client}, \View, \rn}$, then at-least $\nf-\f$ non-faulty replicas stored prepared certificates for proposal $m$, executed $m$, and stored commit certificates for proposals in every round $\rn'$ preceding round $\rn$ ($\rn' < \rn$).
\end{enumerate}
\end{theorem}
\begin{proof} We prove the three statements separately.

Client $\Client$ only receives a $(\View, \rn)$-proof-of-execution for $\Cert{\T}{\Client}$ if at-least $\nf$ distinct replicas signed some message \Message{Inform}{\Cert{\T}{\Client}, \View, \rn, r}. As there are at-most $\f$ faulty replicas and $\nf > 2\f$, there exist at-least $\nf - \f \geq \f+1$ non-faulty replicas that must have sent these \MName{Inform} messages to $\Client$. Choose such a non-faulty replica $\Replica[q]$. Replica $\Replica[q]$ will only do so after preparing some proposal $m =\Message{Propose}{\Cert{\T}{\Client}, \View, \rn}$ and executing $\T$ (\lfref{fig:poe_nc}{inform}). Consequently, $\Replica[q]$ satisfies the conditions to start the commit-phase for $m$ and will broadcast a well-formed message \Message{CheckCommit}{\PC{m}_{\Replica[q]}} to all replicas (\lfref{fig:poe_cc}{broad}). Hence, as communication is reliable, all non-faulty replicas will receive a well-formed message \Message{CheckCommit}{\PC{m}_{\Replica[q]}} and will be able to prepare $m$ (\lfref{fig:poe_cc}{prepare_exec}).

Next, we prove the second statement. Assume that non-faulty replica $\Replica$ stored a commit certificate $\CC{m}$. This certificate is based on well-formed \MName{CheckCommit} messages from $\nf$ distinct replicas (\lfref{fig:poe_cc}{cped}). Consequently, there exists a non-faulty replica $\Replica[q]$ that must have sent one of these \MName{CheckCommit} messages to $\Replica$. This replica $\Replica[q]$ must have broadcasted its well-formed message \Message{CheckCommit}{\PC{m}_{\Replica[q]}} to all replicas (\lfref{fig:poe_cc}{broad}). Hence, as communication is reliable, all non-faulty replicas will receive a well-formed message \Message{CheckCommit}{\PC{m}_{\Replica[q]}}, will be able to prepare $m$ (\lfref{fig:poe_cc}{prepare_exec}), will receive well-formed messages \Message{CheckCommit}{\PC{m}_i} from all $\nf$ non-faulty replicas (\lfref{fig:poe_cc}{broad}), and will be able to commit $m$ (\lfref{fig:poe_cc}{cped}). Finally, we can use the same argument as in the proof of Theorem~\subref{thm:main_nc}{snd} to prove that $\Client$ will receive a $(\View, \rn)$-proof-of-execution for $\Cert{\T}{\Client}$.

Finally, we prove the third statement. The commit certificate $\CC{m}$ can only exist if $\nf$ distinct replicas signed the message \Message{CheckCommit}{m}. Hence, there are at-least $\nf-\f$ non-faulty replicas that signed message \Message{CheckCommit}{m}. These $\nf-\f$ non-faulty replicas will only sign \Message{CheckCommit}{m} after they prepared and executed $m$ and committed all previous rounds (\lfref{fig:poe_cc}{wait}), completing the proof.
\end{proof}

\subsection{The View-Change Protocol}\label{ss:poe_vc}
To deal with disruptions of the normal-case protocol of \PoE{}, \PoE{} employs a view-change protocol.  This protocol has two goals:
\begin{enumerate}
\renewcommand{\theenumi}{G\arabic{enumi}}
\item\label{vcg:preserve} The view-change protocol must always \emph{preserve requests $\Cert{\T}{\Client}$ with a proof-of-execution}: if client $\Client$ could have observed successful execution of transaction $\T$, then the view-change protocol must guarantee that the execution of $\T$ is preserved by the system and included in \emph{all} future views. Furthermore, the view-change protocol must always \emph{preserve committed transactions}.
\item\label{vcg:recover} The view-change protocol must \emph{resume the normal-case protocol} when communication is reliable: if communication is reliable and the normal-case protocol is disrupted to the point where the check-commit protocol can no longer assure that all non-faulty replicas commit, then the view-change protocol eventually puts all non-faulty replicas in the same future view $\View'$ in which the normal-case protocol can operate without disruptions.
\end{enumerate}
These goals reflect the guarantees provided by weak consensus: Goal~\ref{vcg:preserve} will be used to provide \emph{non-divergence}, whereas Goal~\ref{vcg:recover} will be used to provide \emph{termination} (when communication is sufficiently reliable).

The view-change protocol for view $\View$ operates in three stages:
\begin{enumerate}
\item Replicas enter the \emph{failure detection stage} when they detect failure of view $\View$. In this first stage, each replica that detects failure of the current primary will \emph{alert} all other replicas of this failure. These failure alert messages are not only used to assure that all replicas detect failure of the primary, but will also assure sufficient \emph{synchronization} of the replicas to guarantee the view-change protocol will succeed whenever communication is reliable. 

\item When replicas receive failure alerts for the same primary from $\nf$ distinct replicas, they enter the \emph{new-view proposal stage}. During this second stage, replicas provide the \emph{new primary}, the replica $\Primary'$ with $\ID{\Primary'} = (\View + 1) \bmod \n$, with a summary of their internal state. After the new primary $\Primary'$ receives such summary of sufficient replicas, $\Primary'$ can determine the state in which the next view starts, which it then proposes and broadcast to all other replicas. All replicas wait for this new-view proposal from $\Primary'$ (and detect failure of view $\View+1$ if no valid new-view proposal arrives on time).
\item When replicas receive a valid new-view proposal from the new primary, they enter the \emph{new-view accept stage}. In this third and final stage, the replicas will validate the new-view proposal they received, update their local state based on the information in the new-view proposal, and start the normal-case protocol of \PoE{} for this new view.
\end{enumerate}
The pseudo-code for the view-change protocol can be found in Figure~\ref{fig:poe_vc}, and an illustration of the working of this protocol can be found in Figure~\ref{fig:poe_vc_ill}. Next, we will detail each of the three stages in more detail.

\begin{figure}[t]
    \begin{myprotocol}
        \TITLE{Failure detection stage}{running at every replica $\Replica \in \Replicas$}
        \EVENT{$\Replica$ detects failure of view $\View$}\label{fig:poe_vc:detect}
            \IF{$\Replica$ did not previously detect failure of view $\View$}
                \STATE Broadcast \Message{Failure}{\View} to all replicas and
                periodically rebroadcast until the new-view proposal stage is entered.\label{fig:poe_vc:request}
            \ENDIF
        \ENDEVENT
        \EVENT{$\Replica$ receives messages \Message{Failure}{\View'} with $\View' \geq \View$ and\\
                \phantom{\textbf{event }}signed by $\f+1$ distinct replicas}
            \STATE $\Replica$ detects failure of view $\View$.\label{fig:poe_vc:amplify}
        \ENDEVENT
        \SPACE
        \TITLE{New-view proposal stage}{running at every replica $\Replica \in \Replicas$}
        \EVENT{$\Replica$ receives messages \Message{Failure}{\View'} with $\View' \geq \View$ and\\
                \phantom{\textbf{event }}signed by $\nf$ distinct replicas}\label{fig:poe_vc:ep}
            \STATE Halt the normal-case protocol of Section~\ref{ss:poe_nc} for view $\View$.
            \STATE Halt the check-commit protocol of Section~\ref{ss:poe_cc} for view $\View$.
            \STATE Let $\CC{m}$ be the last commit certificate stored by $\Replica$ and let $\VE$ be the set of prepared certificates $\PC{m'}$ for all proposals $m'$ that $\Replica$ executed (without rollback) after proposal $m$.\label{fig:poe_vc:include_ve}
            \STATE Send \Message{ViewState}{\View, \CC{m}, \VE} to replica $\Primary'$, $\ID{\Primary'} = (\View + 1) \bmod \n$.
            \STATE Await a valid \MName{NewView} message for view $\View+1$. If no such message arrives on time, then detect failure of view $\View+1$.\label{fig:poe_vc:failure_next}
        \ENDEVENT
        \EVENT{$\Replica$ receives well-formed messages \Message{ViewState}{\View, \CC{m_i}_i, \VE_i},\\
                \phantom{\textbf{event }}$1 \leq i \leq \nf$, signed by $\nf$ distinct replicas}\label{fig:poe_vc:nve}
            \IF{$\ID{\Replica} = (\View + 1) \bmod \n$ ($\Replica$ is the primary of view $\View+1$)}
                \STATE Let $\VI = \{ \Message{ViewState}{\View, \CC{m_i}_i, \VE_i} \mid 1 \leq i \leq \nf \}$.
                \STATE Broadcast \Message{NewView}{\View + 1, \VI} to all replicas.\label{fig:poe_vc:request_nv}
            \ENDIF
        \ENDEVENT
        \SPACE
        \TITLE{New-view accept stage}{running at every replica $\Replica \in \Replicas$}
        \EVENT{$\Replica$ receives well-formed message \Message{NewView}{\View + 1, \VI}\\
                \phantom{\textbf{event }}from replica  $\Primary'$, $\ID{\Primary'} = (\View + 1) \bmod \n$}\label{fig:poe_vc:nv_msg}
            \STATE Update the internal state in accordance to $\VI$ and start the normal-case protocol of Section~\ref{ss:poe_nc} for view $\View + 1$.\label{fig:poe_vc:update}
        \ENDEVENT
    \end{myprotocol}
    \caption{The view-change protocol of \PoE{}.}\label{fig:poe_vc}
\end{figure}

\begin{figure}[t!]
    \centering
    \begin{tikzpicture}[yscale=0.5,xscale=0.75]
        \draw[thick,draw=black!75] (1.75,   0) edge ++(9.5, 0)
                                   (1.75,   1) edge ++(9.5, 0)
                                   (1.75,   2) edge ++(9.5, 0)
                                   (1.75,   3) edge[blue!50!black!90] ++(9.5, 0);

        \draw[thin,draw=black!75] (2, 0) edge ++(0, 3)
                                  (5, 0) edge ++(0, 3)
                                  (7, 0) edge ++(0, 3)
                                  (9, 0) edge ++(0, 3)
                                  (11, 0) edge ++(0, 3);

        \node[left] at (1.8, 0) {$\Replica[b] \in \Faulty$};
        \node[left] at (1.8, 1) {$\Replica_2$};
        \node[left] at (1.8, 2) {$\Replica_1$};
        \node[left] at (1.8, 3) {$\Replica[p]'$};

        \path[->] (2, 3) edge[colB,dashed] (4, 2)
                         edge[colB,dashed] (4, 1)
                         edge[colB,dashed] (4, 0)
                           
                  (3, 2) edge[colB,dashed] (5, 0)
                         edge[colB,dashed] (5, 1)
                         edge[colB,dashed] (5, 3)

                  (5, 1) edge[colB,dashed] (7, 0)
                           edge[colB,dashed] (7, 2)
                           edge[colB,dashed] (7, 3)

                  (9, 3) edge[<-] (7, 1)
                         edge[<-] (7, 2)

                  (9, 3) edge (11, 2)
                         edge (11, 1)
                         edge (11, 0);
                                                          
        \node[dot,colD] at (9, 3) {};
        \node[sdot,colE] at (11, 0) {};
        \node[sdot,colE] at (11, 1) {};
        \node[sdot,colE] at (11, 2) {};

        \node[label,below,yshift=3pt] at (3.5, 0) {\vphantom{Execute $\T$}\MName{Failure}};
        \node[label,below,yshift=3pt,align=center] (x) at (6, 0) {\vphantom{Execute $\T$}\MName{Failure}};
        \node[label,below] at (x) {(Join)};
        \node[label,below,yshift=3pt] at (8, 0) {\vphantom{\MName{Failure}}\MName{ViewState}};
        \node[label,below,yshift=3pt] at (10, 0) {\vphantom{\MName{Failure}}\MName{NewView}};
    \end{tikzpicture}
    \caption{A schematic representation of the view-change protocol of \PoE{}. The current primary $\Replica[b]$ is faulty and
needs to be replaced. The next primary, $\Replica[p]'$, and the replica $\Replica_1$ detected this failure first and alerted all replicas via \MName{Failure} messages. The replica $\Replica_2$ joins in on this failure. After replicas receive $\nf = 3$ \MName{Failure} messages, they send their state to $\Primary'$ via \MName{ViewState} messages. Finally, $\Primary'$ uses $\nf$ such \MName{ViewState} messages to propose a new view via a \MName{NewView} message.}\label{fig:poe_vc_ill}
\end{figure}
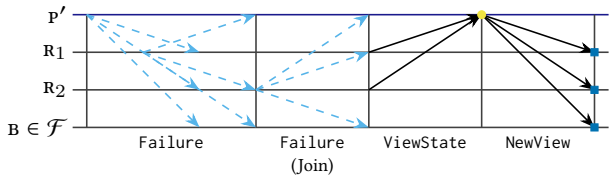

\paragraph{The Failure Detection Stage}

In the failure detection stage, replicas detect failure of view $\View$ and alert other replicas of this failure. To alert other replicas of a failure of view $\View$, replicas will broadcast messages \Message{Failure}{\View}. Before a replica $\Replica$ enters the failure detection stage, replica $\Replica$ needs to detect failure of the primary. Replica $\Replica$ can do so in two ways.

First, $\Replica$ can set a timer whenever it \emph{expects a proposal} for some round $\rn$ from the current primary. If this timer expires and no proposal for round $\rn$ was finished (proposed, executed, and committed), then $\Replica$ detects failure. Replica $\Replica$ can expect a proposal whenever it forwarded a valid client request to the current primary or whenever it receives valid \MName{Prepare} messages for round $\rn$ from non-faulty replicas (e.g., by receiving such \MName{Prepare} messages from $\f+1$ distinct replicas) without receiving any corresponding \MName{Propose} messages.

Second, $\Replica$ can receive failure alerts for the current (or future) view of at-least $\f+1$ distinct other replicas. As there are at-most $\f$ faulty replicas, at-least one of these alerts must have originated from some non-faulty replica $\Replica[q]$. In this case, $\Replica$ can simply use this observation to detect failure.

\paragraph{The New-View Proposal Stage}
Replicas enter the new-view proposal stage for view $\View$ after they receive messages \Message{Failure}{\View'}, $\View' \geq \View$, from $\nf$ distinct replicas. This condition will \emph{synchronize} the view-change in all non-faulty replicas whenever communication is reliable:

\begin{lemma}\label{lem:sync}
Assume communication is reliable and has message delay $\delta$. If the first non-faulty replica to enter the new-view proposal stage does so at time $t$, then all non-faulty replicas will enter the new-view proposal stage before-or-at $t + 2\delta$.
\end{lemma}
\begin{proof}
Let $\Replica$ be the first non-faulty replica that enters the new-view proposal stage. As replica $\Replica$ entered the new-view proposal stage at $t$, it received messages \Message{Failure}{\View'}, $\View' \geq \View$, from $\nf$ distinct replicas before-or-at $t$ (\lfref{fig:poe_vc}{ep}). As there are at-most $\f$ faulty replicas and $\nf > 2\f$, at-least $\nf - \f \geq \f+1$ of these messages originated from non-faulty replicas, whom always broadcast their \MName{Failure} messages (\lfref{fig:poe_vc}{request}). Hence, if communication is reliable, then all replicas will receive at-least $\f+1$ \MName{Failure} messages within at-most a message-delay $\delta$. Consequently, all non-faulty replicas will have detected failure of view $\View$ at-most at $t + \delta$ and broadcast \MName{Failure} messages themselves (\lfref{fig:poe_vc}{amplify}), due to which all replicas will receive $\nf$ \MName{Failure} messages at-most at $t + 2\delta$ and enter the new-view proposal stage for view $\View$ (\lfref{fig:poe_vc}{ep}).
\end{proof}

Consider a period of reliable communication in an asynchronous environment. In this environment, the message delay $\delta$ as used in Lemma~\ref{lem:sync} is bounded by a value unknown to the non-faulty replicas. In Lemma~\ref{lem:poe_timing_proof} and Theorem~\ref{thm:poe_timing}, we show how non-faulty replicas can determine a sufficiently high upper bound for $\delta$.

When a non-faulty replica $\Replica$ enters the new-view proposal stage for view $\View$, $\Replica$ first halts its participation in the normal-case protocol of Section~\ref{ss:poe_nc} and the check-commit protocol of Section~\ref{ss:poe_cc} for view $\View$. Next, $\Replica$ constructs a summary of its internal state. To do so, $\Replica$ constructs the set $\VE$ that holds the prepared certificates $\PC{m'}$ of proposals executed (without rollback) by $\Replica$ and stored \emph{since} the last commit certificate $\CC{m}$ stored by $\Replica$. To simplify presentation, we assume that each replica $\Replica$ has a \emph{dummy commit certificate} for round $\rn = 0$ (that does not propose any request), which $\Replica$ uses when it has not yet committed proposals. The pair $(\CC{m}, \VE)$ will serve as the summary of the current state of $\Replica$. Finally, $\Replica$ sends  $(\CC{m}, \VE)$ to the next primary, the replica $\Primary'$ with $\ID{\Primary'} = (\View + 1) \bmod \n$, via a \Message{ViewState}{\View, \CC{m}, \VE} message.\footnote{In \PoE{}, we distinguish between \MName{Failure} messages, which are \emph{small} and broadcasted to all other replicas, that are used to detect failures; and \MName{ViewState} messages, which are \emph{large} and send only to the next primary, that are used to construct a \MName{NewView} message. If bandwidth is not a limiting factor, then one can simply broadcast \MName{ViewState} messages to all replicas to fulfill both roles.}

The next primary, the replica $\Primary'$ with $\ID{\Primary'} = (\View + 1) \bmod \n$, will wait until it receives well-formed messages \[\VI = \{ \Message{ViewState}{\View, \CC{m_i}_i, \VE_i} \mid 1\leq i \leq \nf \},\] signed by $\nf$ distinct replicas. After $\Primary'$ has received these messages $\VI$, $\Primary'$ broadcasts the message \Message{NewView}{\View+1, \VI} to all replicas. This message announces a new-view whose initial state is based on the information in $\VI$. To assure timely arrival of these messages, we use the following assumption:

\begin{assumption}\label{ass:size}
All \PoE{} messages have a predetermined bounded size. In specific, there is a known upper bound on the size of proposals of the form \Message{Propose}{\Cert{\T}{\Client}, \View, \rn} and there is an upper bound on the number of rounds with prepared proposals that are not yet committed. We refer to this upper bound as the \emph{window size}.
 \end{assumption}

Using Assumption~\ref{ass:size}, we can assume, without loss of generality, that message delivery times are independent of the message size and are fully determined by some unknown message delay.

\begin{lemma}\label{lem:sync_nv}
Assume communication is reliable and has message delay $\delta$. If the first non-faulty replica to enter the new-view proposal stage does so at time $t$, then any non-faulty next primary will be able to deliver a new-view proposal before $t + 4\delta$.
\end{lemma}
\begin{proof}
Due to Lemma~\ref{lem:sync}, all non-faulty replicas will have entered the new-view proposal stage at $t + 2\delta$. Hence, the next primary will have received sufficient \MName{ViewState} messages at $t + 3\delta$ to propose a new-view (\lfref{fig:poe_vc}{nve}), and a well-formed new-view proposal will be broadcast at-or-before $t+3\delta$ to all replicas (\lfref{fig:poe_vc}{request_nv}). Consequently, all replicas will receive this new-view proposal before $t + 4\delta$.
\end{proof}

Due to synchronized entry of the new-view proposal stage, every non-faulty replica $\Replica$ will expect a timely new-view proposal. If no such proposal is received, then $\Replica$ will detect failure of view $\View+1$.

\paragraph{The New-View Accept Stage}

Replicas enter the new-view accept stage after they receive a well-formed \Message{NewView}{\View + 1, \VI} message from the primary of view $\View+1$, the replica $\Primary'$ with $\ID{\Primary'} = (\View + 1) \bmod \n$. Based on the information included in $\VI$, each replica will determine the state in which $\View+1$ starts. This state consists of the set of transactions that have been proposed before view $\View +1$ and, hence, determines at which round the primary $\Primary'$ of $\View+1$ can start proposing. We make the following key assumption:

\begin{assumption}\label{ass:vc}
If there exists a commit certificate $\CC{m}$ for proposal $m = \Message{Propose}{\Cert{\T}{\Client}, \View, \rn}$, then every commit certificate $\CC{m'}$ for round $\rn$ is a commit certificate for a proposal proposing $\Cert{\T}{\Client}$.
\end{assumption}

As part of proving the correctness of the view-change protocol, we will prove that Assumption~\ref{ass:vc} holds. Before we do so, we first show how we interpret the unique state represented by a well-formed proposal $n = \Message{NewView}{\View + 1, \VI}$ using Assumption~\ref{ass:vc}. Furthermore, we will show that this derived state satisfies Goal~\ref{vcg:preserve}. Let $\VI = \{ \Message{ViewState}{\View, \CC{m_i}_i, \VE_i} \mid 1 \leq i \leq \nf \}$ be the set of \MName{ViewState} messages included in $n$. Let
\begin{align*}
    \NVcc{\rn} &= \{ \Cert{\T}{\Client} \mid \exists\View\ \exists \CC{m}\ (m = \Message{Propose}{\Cert{\T}{\Client}, \View, \rn}) \},
\intertext{be the set of client requests $\Cert{\T}{\Client}$ that have been committed in round $\rn$ (a view $\View$ exists for which a commit certificate $\CC{m}$ exists with $m = \Message{Propose}{\Cert{\T}{\Client}, \View, \rn}$). Let}
    \NVP{n}{\rn} &= \{ m \mid (\PC{m} \in \VE_i) \land (m = \Message{Propose}{\Cert{\T}{\Client}, \View, \rn}) \}
\intertext{be the set of all proposals included in $\VI$ for round $\rn$, let}
    \NVm{n}{\rn} &= \{ \Cert{\T}{\Client} \mid  (\Message{Propose}{\Cert{\T}{\Client}, \View, \rn} \in \NVP{n}{\rn}) \land{}\\
                &\phantom{= \{ \Cert{\T}{\Client} \mid{}} \View = \max\{ \View' \mid \Message{Propose}{\Cert{\T'}{\Client'}, \View', \rn} \in \NVP{n}{\rn} \} \}
\intertext{be the set of client requests proposed in the most-recent view of any proposals in $\NVP{n}{\rn}$, let}
    \NVrncc{n} &= \max\{\rn_i \mid m_i = \Message{Propose}{\Cert{\T_i}{\Client_i}, \View_i, \rn_i} \}
\intertext{be the latest round for which a commit certificate is included in $\VI$, and let}
    \NVrn{n} &= \max(\{\rn \mid \NVP{n}{\rn} \neq \emptyset\} \union \NVrncc{n})
\end{align*}
be the latest round for which $n$ includes proposals. We will use the sets $\NVcc{\cdot}$ and $\NVm{n}{\cdot}$ to define the unique sequence of client transactions preserved by new-view proposal $n$. To do so, we use the following technical result:

\begin{lemma}\label{lem:unique_vc}
Let $n =  \Message{NewView}{\View + 1, \VI}$ be a well-formed new-view proposal. We have:
\begin{enumerate}
\item if Assumption~\ref{ass:vc} holds, then $\NVcc{\rn}$ is a singleton set ($\abs{\NVcc{\rn}} = 1$) for all $\rn \leq \NVrncc{n}$ ; and
\item $\NVm{n}{\rn}$ is a singleton set ($\abs{\NVm{n}{\rn}} = 1$) for all $\NVrncc{n} < \rn \leq \NVrn{n}$.
\end{enumerate}
\end{lemma}
\begin{proof}
First, we prove $\abs{\NVcc{\rn}} = 1$ for all $\rn \leq \NVrncc{n}$. By the definition of $\NVrncc{n}$, there exists a commit certificate for round $\NVrncc{n}$. By Theorem~\subref{thm:main_cc}{trd}, there exist commit certificates for all rounds $\rn \leq \NVrncc{n}$. Hence, for all $\rn \leq \NVrncc{n}$, $\NVcc{\rn} \neq \emptyset$. Finally, by Assumption~\ref{ass:vc}, we conclude $\abs{\NVcc{\rn}} = 1$.

Next, we prove $\abs{\NVm{n}{\rn}} = 1$ for all $\NVrncc{n} < \rn \leq \NVrn{n}$. By the definition of $\NVrn{n}$, we must have $\NVP{n}{\NVrn{n}} \neq \emptyset$ and $\NVm{n}{\NVrn{n}} \neq \emptyset$. Let $\Cert{\T}{\Client} \in \NVm{n}{\NVrn{n}}$. By the definition of $\NVm{n}{\NVrn{n}}$, there exists a message $m_{\MName{vs}} = \Message{ViewState}{\View, \CC{m}, \VE} \in \VI$ such that there exists a prepared certificate $\PC{m'} \in \VE$ with $m' = \Message{Propose}{\Cert{\T}{\Client}, \View, \NVrn{n}}$. Let $\Replica$ be the replica that signed this message $m_{\MName{vs}}$ and let $\rn'$ be the round for which $m$ was proposed. By the definition of $\NVrncc{n}$, we have $\rn' \leq \NVrncc{n}$. As $\Message{NewView}{\View + 1, \VI}$ is well-formed, also the message $m_{\MName{vs}}$ is well-formed. As non-faulty replicas only execute proposals for round $\rn$ after they executed proposals for all preceding rounds (\lfref{fig:poe_nc}{execute}) and $\VE$ is well-formed, it must contain proposals for all rounds $\rn$, $\rn' \leq \NVrncc{n} < \rn \leq \NVrn{n}$. Hence, for all $\NVrncc{n} < \rn \leq \NVrn{n}$, $\NVP{n}{\rn} \neq \emptyset$ which implies $\NVm{n}{\rn} \neq \emptyset$. Due to Theorem~\subref{thm:main_nc}{fst}, we also have $\abs{\NVm{n}{\rn}} \leq 1$ and we conclude $\abs{\NVm{n}{\rn}} = 1$.
\end{proof}

Due to Lemma~\ref{lem:unique_vc}, the sequence of client requests 
\begin{multline*}\Ledger(n) = \Cert{\T_1}{\Client_1}, \dots, \Cert{\T_{\NVrn{n}}}{\Client_{\NVrn{n}}}\text{ with }\\
 \Cert{\T_\rn}{\Client_\rn} \in \begin{cases} \NVcc{\rn} &\text{if $1 \leq \rn \leq \NVrncc{n}$;}\\
                                              \NVm{n}{\rn} &\text{if $\NVrncc{n} < \rn \leq \NVrn{n}$.}\end{cases}
\end{multline*}
is uniquely defined by the new-view proposal $n$ and specifies the state in which view $\View + 1$ starts. Replica $\Replica$ will update its internal state (\lfref{fig:poe_vc}{update}) in accordance to $\Ledger(n)$ in the following way:
\begin{enumerate}
\item $\Replica$ will \emph{rollback} every client request it executed and that is not included in $\Ledger(n)$;
\item $\Replica$ will obtain a commit certificate for each round $\rn$, $1 \leq \rn \leq \NVrncc{n}$, for which it does not yet have a commit certificate, execute the newly obtained requests in order, and inform the client of the result,
\item $\Replica$ will expect the new primary to \emph{repropose} client requests $\Cert{\T_{\rn}}{\Client_{\rn}}$, $\NVrncc{n} < \rn \leq \NVrn{n}$, in round $\rn$ of view $\View + 1$. If the new primary fails to do so, then failure of view $\View + 1$ is detected.
\end{enumerate}

The last step assures that all non-faulty replicas that receive a new-view proposal for view $\View+1$ will only proceed in this new view if they all received \emph{compatible} new-view proposals that each represent the same \emph{unique} ledger.

After updating its internal state, replica $\Replica$ will start the normal-case protocol for view $\View+1$ by accepting any proposal from $\Primary'$ for rounds after $\NVrn{n}$.

Next, we prove that the view-change protocol outlined above satisfies Goal~\ref{vcg:preserve} and Goal~\ref{vcg:recover}.

\paragraph{The View-Change Protocol Satisfies Goal~\ref{vcg:preserve}}
A client request $\Cert{\T}{\Client}$ needs to be preserved by the view-change protocol as the $\rn$-th request if it either has a $(\View, \rn)$-proof-of-execution or a replica stored a commit certificate $\CC{m}$ for some proposal $m = \Message{Propose}{\Cert{\T}{\Client}, \View, \rn}$. If $\Cert{\T}{\Client}$ has a $(\View, \rn)$-proof-of-execution, then there must be a set of identical messages \Message{Inform}{\Cert{\T}{\Client}, \View, \rn, r} signed by $\nf$ distinct replicas. Likewise, if $\Cert{\T}{\Client}$ has a commit certificate $\CC{m}$, then there must be a set of identical messages \Message{CheckCommit}{m} signed by $\nf$ distinct replicas. As there are at-most $\f$ faulty replicas, at-least $\nf - \f$ of these messages must originate from non-faulty replicas, which will only produce these messages after they prepared and executed proposal $m$.  Hence, a \emph{necessary condition} for the preservation of $\Cert{\T}{\Client}$ is the existence of a proposal $m$ that is executed by at-least $\nf - \f$ non-faulty replicas. Next, we prove that the view-change protocol preserves such requests:

\begin{theorem}\label{thm:vc_all}
Let $\View$ be the first view in which a proposal $m = \Message{Propose}{\Cert{\T}{\Client}, \View, \rn}$ for round $\rn$ was executed by $\nf - \f$ non-faulty replicas, let $m' = \Message{Propose}{\Cert{\T'}{\Client'}, \View', \rn}$ be any proposal for round $\rn$ with $\View \leq \View'$, and let $n = \Message{NewView}{\View[w] + 1, \VI}$, $\View \leq \View' \leq \View[w]$, be a well-formed new-view proposal. The following properties hold:
\begin{enumerate}
    \item\label{thm:vc_all:no_commit} there exist no commit certificates $\CC{\Message{Propose}{\Cert{\T''}{\Client''}, \View'', \rn}}$ with $\View'' < \View$;
    \item\label{thm:vc_all:future_sign} in views $\View'$, $\View < \View'$,  non-faulty replicas only sign \Message{Prepare}{m'} if $\Cert{\T'}{\Client'} = \Cert{\T}{\Client}$;
    \item\label{thm:vc_all:prepared_nodiv} if there exists a prepared certificate $\PC{m'}$, then $\Cert{\T'}{\Client'} = \Cert{\T}{\Client}$;
    \item\label{thm:vc_all:nv_rounds} if $\nf - \f$ non-faulty replicas executed $m'$, then $\rn \leq \NVrn{n}$;
    \item\label{thm:vc_all:nv_rounds_cc} if $\nf - \f$ non-faulty replicas committed $m'$, then $\rn \leq \NVrncc{n}$; and
    \item\label{thm:vc_all:nv_ledger} $\Ledger(n)[\rn] = \Cert{\T}{\Client}$;
\end{enumerate}
\end{theorem}
\begin{proof}
First, we prove~\inref{thm:vc_all:no_commit} by contradiction. Assume there exists a commit certificate $\CC{\Message{Propose}{\Cert{\T''}{\Client''}, \View'', \rn}}$ with $\View'' < \View$. Hence, there exist messages \Message{CheckCommit}{\PC{m''}_i}, $1 \leq i \leq \nf$, signed by $\nf$ distinct replicas (\lfref{fig:poe_cc}{store}). At-least $\nf-\f$ of these messages are signed by non-faulty replicas. As non-faulty replicas only construct and sign messages \Message{CheckCommit}{\PC{m''}_i} if they prepared and executed $m''$ (\lfref{fig:poe_cc}{broad}) and $\View$ was the first view in which $\nf-\f$ non-faulty replicas executed a proposal for round $\rn$, we must have $\View \leq \View''$, a contradiction, and we conclude that $\CC{m''}$ does not exist.

\smallskip

We prove~\inref{thm:vc_all:future_sign}--\inref{thm:vc_all:nv_ledger} by induction on the current view $\View[w]$. As the base case, we prove that each of the statements~\inref{thm:vc_all:future_sign}--\inref{thm:vc_all:nv_ledger} hold in view $\View[w] = \View$.
 
\inref{thm:vc_all:future_sign}. We have $\View[w] = \View' = \View$. Hence, the statement voidly holds.

\inref{thm:vc_all:prepared_nodiv}. As $\View[w] = \View$, we only need to consider $\View = \View'$. If $\View = \View'$, then, by Theorem~\subref{thm:main_nc}{fst}, we have $m' = m$ and $\Cert{\T'}{\Client'} = \Cert{\T}{\Client}$.

\inref{thm:vc_all:nv_rounds} and~\inref{thm:vc_all:nv_rounds_cc}.
As $\View[w] = \View$, we can use~\inref{thm:vc_all:no_commit} and the proof of~\inref{thm:vc_all:prepared_nodiv}, to derive that $m' = m$. Let $n = \Message{NewView}{\View[w] + 1, \VI}$, $\View[w] = \View$, be a well-formed new-view proposal. Let $C$ be the set of $\nf-\f$ non-faulty replicas that executed $m$ (proof for~\inref{thm:vc_all:nv_rounds}) or committed $m$ (proof for~\inref{thm:vc_all:nv_rounds_cc}), let $S$ be the set of $\nf$ distinct replicas that signed the \MName{ViewState} messages included in $\VI$, and let $T = S \difference \Faulty$ be the non-faulty replicas in $S$. By construction, we have $\abs{T} \geq \nf - \f$.  Hence, using the same contradiction argument as used in the proof of Theorem~\subref{thm:main_nc}{fst}, we can conclude that $(C \intersect T) \neq \emptyset$. Let $\Replica \in (C \intersect T)$ be such a non-faulty replica that executed $m$ (proof for~\inref{thm:vc_all:nv_rounds}) or committed $m$ (proof for~\inref{thm:vc_all:nv_rounds_cc}) and that signed a message $\Message{ViewState}{\View, \CC{m_{\Replica}}, \VE} \in \VI$ with $m_{\Replica} = \Message{Propose}{\Cert{\T_{\Replica}}{\Client_{\Replica}}, \View_{\Replica}, \rn_{\Replica}}$.

If $\Replica$ committed $m$, then we must have $\rn \leq \rn_{\Replica}$ as non-faulty replicas commit proposals in order (\lfref{fig:poe_cc}{wait_store}) and $m_{\Replica}$ is the last proposal $\Replica$ committed (\lfref{fig:poe_vc}{include_ve}). Hence, by the definition of $\NVrncc{n}$ and $\NVrn{n}$, we have $\rn \leq \rn_{\Replica} \leq \NVrncc{n} \leq \NVrn{n}$. 

If $\Replica$ executed $m$ without committing $m$, then $\rn > \rn_{\Replica}$. As $\Replica$ executed $m$, $\Replica$ must have stored a prepared certificate $\PC{m}$ (\lfref{fig:poe_nc}{store}) and we conclude $\PC{m} \in \VE$ (\lfref{fig:poe_vc}{include_ve}). By the definition of $\NVP{n}{\rn}$ and the definition of $\NVrn{n}$, we conclude $m \in \NVP{n}{\rn}$ and $\rn \leq \NVrn{n}$.

\inref{thm:vc_all:nv_ledger}. As $\nf-\f$ non-faulty replicas executed $m$ in view $\View[w] = \View$, we have $\rn \leq \NVrn{n}$ by~\inref{thm:vc_all:nv_rounds_cc}. Hence, $\Ledger(n)[\rn]$ is defined.

If $\rn \leq \NVrncc{n}$, then, by the definition of $\Ledger(n)$ and $\NVcc{\rn}$, we have $\Ledger(n)[\rn] = \Cert{\T_n}{\Client_n}$ with $\Cert{\T_n}{\Client_n} \in \NVcc{\rn}$ and there exists a commit certificate $\CC{m_n}$ for some proposal $m_n = \Message{Propose}{\Cert{\T_n}{\Client_n}, \View_n, \rn}$. By~\inref{thm:vc_all:no_commit}, we conclude that $\View \leq \View_n$ and, hence, $\View_n = \View$. Due to Theorem~\subref{thm:main_cc}{trd}, there exists a prepared certificate $\PC{m_n}$ and, by~\inref{thm:vc_all:prepared_nodiv}, we conclude $\Cert{\T_n}{\Client_n} = \Cert{\T}{\Client}$.

If $\NVrncc{n} < \rn \leq \NVrn{n}$, then, by the definition of $\Ledger(n)$, $\NVm{n}$, and $\NVP{n}{\rn}$, we have $\Ledger(n)[\rn] = \Cert{\T_n}{\Client_n}$ with $\Cert{\T_n}{\Client_n} \in \NVm{n}{\rn}$ and there exists a prepared certificate $\PC{m_n}$ for some proposal $m_n = \Message{Propose}{\Cert{\T_n}{\Client_n}, \View_n, \rn}$ with $m_n \in \NVP{n}{\rn}$. Using the same reasoning as in the proof of~\inref{thm:vc_all:nv_rounds}, there exists a non-faulty replica $\Replica$ that executed $m$ and signed a message $\Message{ViewState}{\View, \CC{m_{\Replica}}, \VE} \in \VI$. As $\NVrncc{n} < \rn$, $\Replica$ did not commit $m$, $\PC{m} \in \VE$, and $m \in \NVP{n}{\rn}$. By the definition of $\NVm{n}{\rn}$, we conclude $\View \leq \View_n$. Hence, $\View_n = \View$ and, by~\inref{thm:vc_all:prepared_nodiv}, we conclude $\Cert{\T_n}{\Client_n} = \Cert{\T}{\Client}$.

\smallskip

As the induction hypothesis, we assume that~\inref{thm:vc_all:future_sign}--\inref{thm:vc_all:nv_ledger} hold in every view $j$, $\View \leq j < \View[w]$. Now consider view $\View[w]$ and let $n_{\View[w]} = \Message{NewView}{\View[w], \VI_{\View[w]}}$ be a well-formed new-view proposal that can be used to enter view $\View[w]$ (\lfref{fig:poe_vc}{nv_msg}). Next, we prove the inductive step for each of the statements~\inref{thm:vc_all:future_sign}--\inref{thm:vc_all:nv_ledger}. 

\inref{thm:vc_all:future_sign}. Let \Message{Prepare}{m'}, $\View < \View' \leq \View[w]$, be a message signed by a non-faulty replica. We apply induction hypothesis~\inref{thm:vc_all:nv_rounds} on $n_{\View[w]}$ and conclude $\rn \leq \NVrn{n_{\View[w]}}$. If, furthermore, $\rn \leq \NVrncc{n_{\View[w]}}$, then the primary of view $\View[w]$ cannot propose for round $\rn$ and we must have $\View \leq \View' < \View[w]$. We apply induction hypothesis~\inref{thm:vc_all:future_sign} on $\Message{Prepare}{m'}$ to conclude $\Cert{\T'}{\Client'} = \Cert{\T}{\Client}$. Otherwise, if $\NVrncc{n_{\View[w]}} < \rn \leq \NVrn{n_{\View[w]}}$, then the primary of view $\View[w]$ can only propose $\Cert{\T}{\Client}$ for round $\rn$ (\lfref{fig:poe_vc}{update}), and we conclude $\Cert{\T'}{\Client'} = \Cert{\T}{\Client}$.

\inref{thm:vc_all:prepared_nodiv}. Let $\PC{m'}$ be a prepared certificate, $\View \leq \View' \leq \View[w]$. If $\View' < \View[w]$, then we apply induction hypothesis~\inref{thm:vc_all:prepared_nodiv} to conclude that $\Cert{\T'}{\Client'} = \Cert{\T}{\Client}$. Otherwise, if $\View' = \View$, then there exist messages \Message{Prepare}{m} signed by $\nf$ distinct replicas (\lfref{fig:poe_nc}{store}). At-least $\nf-\f$ of these messages are signed by non-faulty replicas. By~\inref{thm:vc_all:future_sign}, these non-faulty replicas will only sign $m'$ if $\Cert{\T'}{\Client'} = \Cert{\T}{\Client}$.

\inref{thm:vc_all:nv_rounds} and~\inref{thm:vc_all:nv_rounds_cc}.
Let $n = \Message{NewView}{\View[w] + 1, \VI}$ be a well-formed new-view proposal. Let $C$ be the set of $\nf-\f$ non-faulty replicas that executed $m'$ (proof for~\inref{thm:vc_all:nv_rounds}) or committed $m'$ (proof for~\inref{thm:vc_all:nv_rounds_cc}), let $S$ be the set of $\nf$ distinct replicas that signed the \MName{ViewState} messages included in $\VI$, and let $T = S \difference \Faulty$ be the non-faulty replicas in $S$. By construction, we have $\abs{T} \geq \nf - \f$.  Hence, using the same contradiction argument as used in the proof of Theorem~\subref{thm:main_nc}{fst}, we can conclude that $(C \intersect T) \neq \emptyset$. Let $\Replica \in (C \intersect T)$ be such a non-faulty replica that executed $m'$ (proof for~\inref{thm:vc_all:nv_rounds}) or committed $m'$ (proof for~\inref{thm:vc_all:nv_rounds_cc}) and that signed a message $\Message{ViewState}{\View, \CC{m_{\Replica}}, \VE} \in \VI$ with $m_{\Replica} = \Message{Propose}{\Cert{\T_{\Replica}}{\Client_{\Replica}}, \View_{\Replica}, \rn_{\Replica}}$. Let $\View''$ be the last view in which $\Replica$ executed transactions. As $\Replica$ executed $m'$, we have $\View \leq \View' \leq \View'' \leq \View[w]$. 

If $\View' = \View''$, then $\Replica$ is guaranteed to execute $m'$ in view $\View''$ (proof of~\inref{thm:vc_all:nv_rounds}) or to commit $m'$ in view $\View''$ (proof of~\inref{thm:vc_all:nv_rounds_cc}). If $\View' < \View''$, then $\Replica$ used a well-formed new-view message $n'' = \Message{NewView}{\View[w], \VI_{\View[w]}}$ to enter view $\View'' \leq \View[w]$ (\lfref{fig:poe_vc}{nv_msg}). For the proof of~\inref{thm:vc_all:nv_rounds}, we apply induction hypothesis~\inref{thm:vc_all:nv_rounds} on $n''$ to conclude that $\rn \leq \NVrn{n''}$. As $\Replica$ executed transactions in view $\View''$, $\Replica$ is guaranteed to have executed some proposal for round $\rn$ while updating its internal state (\lfref{fig:poe_vc}{update}). For the proof of~\inref{thm:vc_all:nv_rounds_cc}, we apply induction hypothesis~\inref{thm:vc_all:nv_rounds_cc} on $n''$ to conclude that $\rn \leq \NVrncc{n''}$. As $\Replica$ executed transactions in view $\View''$, $\Replica$ is guaranteed to have committed some proposal for round $\rn$ while updating its internal state (\lfref{fig:poe_vc}{update}).  

If $\Replica$ is guaranteed to have committed some proposal $m''$ in round $\rn$ of view $\View''$ or in round $\rn$ when entering view $\View''$, then we must have $\rn \leq \rn_{\Replica}$ as non-faulty replicas commit proposals in order (\lfref{fig:poe_cc}{wait_store}) and $m_{\Replica}$ is the last proposal $\Replica$ committed (\lfref{fig:poe_vc}{include_ve}). Hence, by the definition of $\NVrncc{n}$ and $\NVrn{n}$, we have $\rn \leq \rn_{\Replica} \leq \NVrncc{n} \leq \NVrn{n}$. 

If $\Replica$ is guaranteed to have executed some proposal $m''$ in round $\rn$ of view $\View''$ or in round $\rn$ when entering view $\View''$ without committing $m''$, then $\rn > \rn_{\Replica}$. As $\Replica$ executed $m''$, $\Replica$ must have stored a prepared certificate $\PC{m''}$ (\lfref{fig:poe_nc}{store}) and we conclude $\PC{m''} \in \VE$ (\lfref{fig:poe_vc}{include_ve}). By the definition of $\NVP{n}{\rn}$ and the definition of $\NVrn{n}$, we conclude $m \in \NVP{n}{\rn}$ and $\rn \leq \NVrn{n}$.

\inref{thm:vc_all:nv_ledger}
 As $\nf-\f$ non-faulty replicas executed $m$ in view $\View$, we have $\rn \leq \NVrn{n}$ by~\inref{thm:vc_all:nv_rounds_cc}. Hence, $\Ledger(n)[\rn]$ is defined.

If $\rn \leq \NVrncc{n}$, then, by the definition of $\Ledger(n)$ and $\NVcc{\rn}$, we have $\Ledger(n)[\rn] = \Cert{\T_n}{\Client_n}$ with $\Cert{\T_n}{\Client_n} \in \NVcc{\rn}$ and there exists a commit certificate $\CC{m_n}$ for some proposal $m_n = \Message{Propose}{\Cert{\T_n}{\Client_n}, \View_n, \rn}$. By~\inref{thm:vc_all:no_commit}, we conclude that $\View \leq \View_n$. Due to Theorem~\subref{thm:main_cc}{trd}, there exists a prepared certificate $\PC{m_n}$ and, by~\inref{thm:vc_all:prepared_nodiv}, we conclude $\Cert{\T_n}{\Client_n} = \Cert{\T}{\Client}$.

If $\NVrncc{n} < \rn \leq \NVrn{n}$, then, by the definition of $\Ledger(n)$, $\NVm{n}{\rn}$, and $\NVP{n}{\rn}$, we have $\Ledger(n)[\rn] = \Cert{\T_n}{\Client_n}$ with $\Cert{\T_n}{\Client_n} \in \NVm{n}{\rn}$ and there exists a prepared certificate $\PC{m_n}$ for some proposal $m_n = \Message{Propose}{\Cert{\T_n}{\Client_n}, \View_n, \rn}$ with $m_n \in \NVP{n}{\rn}$. As $\nf-\f$ replicas executed $m$, we can use the same reasoning as in the proof of~\inref{thm:vc_all:nv_rounds} to obtain a non-faulty replica $\Replica$ that executed $m$ and signed a message $\Message{ViewState}{\View, \CC{m_{\Replica}}, \VE} \in \VI$ with $m_{\Replica} = \Message{Propose}{\Cert{\T_{\Replica}}{\Client_{\Replica}}, \View_{\Replica}, \rn_{\Replica}}$. Let $\View''$ be the last view in which $\Replica$ executed transactions. As $\Replica$ executed $m$, we have $\View \leq \View'' \leq \View[w]$.

If $\View = \View''$, then $\Replica$ executed some proposal $m'' = m$ in view $\View''$. If $\View' < \View''$, then $\Replica$ used some new-view message $n'' = \Message{NewView}{\View[w], \VI_{\View[w]}}$ to enter view $\View'' \leq \View[w]$  (\lfref{fig:poe_vc}{nv_msg}). As $\Replica$ executed transactions in view $\View''$, $\Replica$ is guaranteed to have executed some proposal $m''$ for round $\rn$ while updating its internal state (\lfref{fig:poe_vc}{update}). As $\NVrncc{n} < \rn$, $\Replica$ did not commit $m''$ in round $\rn$ , $\PC{m''} \in \VE$, and $m'' \in \NVP{n}{\rn}$. By the definition of $\NVm{n}{\rn}$, we conclude $\View \leq \View_n$. By~\inref{thm:vc_all:prepared_nodiv}, we conclude $\Cert{\T_n}{\Client_n} = \Cert{\T}{\Client}$.
\end{proof}

Theorem~\ref{thm:vc_all} not only proves that the view-change protocol satisfies Goal~\ref{vcg:preserve}, it also proves that Assumption~\ref{ass:vc} holds.

\begin{corollary}[Assumption~\ref{ass:vc}]\label{cor:ass}
If there exists a commit certificate $\CC{m}$ for proposal $m = \Message{Propose}{\Cert{\T}{\Client}, \View, \rn}$, then every commit certificate $\CC{m'}$ for round $\rn$ is a commit certificate for a proposal proposing $\Cert{\T}{\Client}$.
\end{corollary}
\begin{proof}
By Theorem~\subref{thm:main_cc}{trd}, at-least $\nf - \f$ non-faulty replicas stored a prepared certificate for $m$ and executed $m$ whenever a commit certificate $\CC{m}$ exists. Let $m' = \Message{Propose}{\Cert{\T'}{\Client'}, \View', \rn}$, $\View' \leq \View$, be the first proposal for round $\rn$ that is executed by $\nf - \f$ non-faulty replicas. By Theorem~\subref{thm:vc_all}{prepared_nodiv}, we conclude that $\Cert{\T'}{\Client'} = \Cert{\T}{\Client}$.
\end{proof}

Theorem~\ref{thm:vc_all} and Corollary~\ref{cor:ass} assure that one can derive a unique ledger from each well-formed new-view proposal $n = \Message{NewView}{\View + 1, \VI}$. To resume the normal-case protocol, we also need the guarantee that each non-faulty replica derives the \emph{same} ledger from the new-view proposal they receive, even if they receive different new-view proposals. For rounds for which a commit certificate exists, Theorem~\ref{thm:vc_all} already provides this guarantee, whereas for rounds for which only prepared certificates exists, the \emph{repropose} mechanisms will enforce this guarantee. Finally, to resume the normal-case protocol, each individual replica also needs to be able to derive from $n$ the exact content of $\Ledger(n)$. Unfortunately, the message $n$ itself only contains all necessary information to derive $\Ledger(n)[\rn]$ for all $\rn$, $\NVrncc{n} \leq \rn \leq \NVrn{n}$. Fortunately, also the remainder of $\Ledger(n)$ can be derived:

\begin{lemma}\label{lem:missing_cc}
Let $n =  \Message{NewView}{\View + 1, \VI}$ be a well-formed new-view proposal. If communication is reliable, then every non-faulty replica can derive $\Ledger(n)$ (this independent of their internal state).
\end{lemma}
\begin{proof}
Let $\Replica$ be a non-faulty replica that receives $n$ and is unaware of a client request $\Ledger(n)[\rn]$, $\rn < \NVrncc{n}$. By the definition of $\NVrncc{n}$, there exists a message $m_{\MName{vs}} = \Message{ViewState}{\View, \CC{m}, \VE}$ with $m = \Message{Propose}{\Cert{\T}{\Client}, \View[w], \NVrncc{n}}$, $\View[w] \leq \View$. By Theorem~\subref{thm:main_cc}{trd}, the existence of $\CC{m}$ assures that at-least $\nf-\f$ non-faulty replicas stored commit certificates for some proposal $m' = \Message{Propose}{\Cert{\T'}{\Client'}, \View[w]', \rn}$.

Hence, $\Replica$ can query all replicas for the transaction they committed in round $\rn$. To do so, $\Replica$ uses the query protocol outlined in Figure~\ref{fig:query_protocol}. As the first step, $\Replica$ broadcasts a message \Message{QueryCC}{\rn} to all replicas (\lfref{fig:query_protocol}{broad}). Non-faulty replicas will respond with the message \Message{RespondCC}{\PC{m''}, \CC{m''}} if they committed some proposal $m''$ in round $\rn$ and will not respond otherwise (\lfref{fig:query_protocol}{response}). By Corollary~\ref{cor:ass}, no commit certificates for round $\rn$ can exist that proposes a request other than $\Cert{\T'}{\Client'}$. Hence, using any of the received messages \Message{RespondCC}{\CC{m''}}, $\Replica$ will be able to prepare, execute, and commit a proposal for $\Cert{\T'}{\Client'}$ in round $\rn$ (\lfref{fig:query_protocol}{perform}).
\end{proof}

\begin{figure}[t]
    \begin{myprotocol}
        \TITLE{Query role}{running at replica $\Replica$}
        \STATE Broadcast \Message{QueryCC}{\rn} to all replicas.\label{fig:query_protocol:broad}
        \EVENT{$\Replica$ receives a well-formed message \Message{RespondCC}{\PC{m},\CC{m}}}\label{fig:query_protocol:respond}
            \STATE Let $m = \Message{Propose}{\Cert{\T}{\Client}, \View', \rn'}$ and $\PC{m} = \{ m_1, \dots, m_{\nf} \}$.
            \IF{$\Replica$ did not commit in round $\rn'$}\label{fig:query_protocol:perform}
                \STATE Wait until all previous rounds have a commit certificate.
                \STATE Prepare $m$ and execute $\T$ using \lsfref{fig:poe_nc}{store}{inform} with \MName{Prepare} messages $m_1, \dots, m_{\nf}$.
                \STATE Store \emph{commit certificate} $\CC{m}$.\label{fig:query_protocol:store_cc}
            \ENDIF
        \ENDEVENT
        \SPACE
        \TITLE{Query response role}{running at every replica $\Replica[q] \in \Replicas$}
        \EVENT{$\Replica[q]$ receives messages \Message{QueryCC}{\rn} from $\Replica$}
            \IF{$\Replica[q]$ stored prepared certificate $\PC{m}$ and commit certificate $\CC{m}$\\
                \phantom{\textbf{event }}for some proposal $m = \Message{Propose}{\Cert{\T}{\Client}, \View, \rn}$}
                \STATE Send \Message{RespondCC}{\PC{m},\CC{m}} to replica $\Replica$.\label{fig:query_protocol:response}
            \ENDIF
        \ENDEVENT
        \end{myprotocol}
    \caption{The query protocol of Lemma~\ref{lem:missing_cc} used by replicas $\Replica \in \Replicas$ when recovering a missing committed proposal for round $\rn$.}\label{fig:query_protocol}
\end{figure}

\paragraph{The View-Change Protocol Satisfies Goal~\ref{vcg:recover}}

In Lemma~\ref{lem:sync} and Lemma~\ref{lem:sync_nv}, we already outlined the main timing-based synchronization steps used by the view-change protocol to guarantee Goal~\ref{vcg:recover}. This mechanism relies on the availability of some message delay value $\delta$ known to all non-faulty replicas that is long enough to guarantee the delivery of any message used by \PoE{} (see also Assumption~\ref{ass:size}). In practice, non-faulty replicas are not aware of this bound and must use some internal message delay estimate.

Let $\Replica$ be a non-faulty replica that uses internal message delay estimate $\Delta(\View, \Replica)$ in view $\View$. If $\Delta(\View, \Replica)$ is too small, then $\Replica$ can erroneously detect failure of $\View+1$ (e.g., at \lfref{fig:poe_vc}{failure_next}). To assure that $\Replica$ will not erroneously detect failure of consecutive views, $\Replica$ needs to eventually use an internal message delay estimate $\Delta(\View + i, \Replica)$ in view $\View + i$ for which $\Delta(\View + i, \Replica) \geq \delta$ holds. To do so, we assume that all non-faulty replicas $\Replica[q]$ have some \emph{backoff function} $f(i)$ for which $f(i) \geq i$ holds and use the internal message delay estimate $\Delta(\View + i, \Replica[q]) = f(i) \cdot \Delta(\View, \Replica[q])$. We have:

\begin{lemma}\label{lem:poe_timing_proof}
Let $\Delta(\View, \Replica)$ be the internal message delay estimate of replica $\Replica$ at view $\View$, let $f(i)$ be the backoff function used by $\Replica$, and let $\delta$ be a message delay value. There exists a $j$ such that $\Delta(\View + j, \Replica) \geq \delta$.
\end{lemma}
\begin{proof}
By definition, we have $\Delta(\View + i, \Replica) = f(i) \cdot \Delta(\View, \Replica) \geq i \cdot \Delta(\View, \Replica)$. Hence, for all $j \geq \smash{\left\lceil \tfrac{\delta}{\Delta(\View, \Replica)} \right\rceil}$, we have $\Delta(\View + j, \Replica) \geq \delta$. 
\end{proof}

Frequently, an exponential function is used as the backoff function, leading to so-called \emph{exponential backoff}:

\begin{example}\label{ex:vc_bc}
Consider a system with $\Replicas = \{ \Replica_1, \Replica_2, \Replica_3, \Replica_4 \}$ in an environment with a message delay of $\delta = \SI{30}{\milli\second}$. Assume that all replicas are non-faulty and that we have the following internal message delay estimates in view $\View$:
\begin{align*}
	\Delta(\View, \Replica_1) &= \SI{7}{\milli\second};&
	\Delta(\View, \Replica_2) &= \SI{3}{\milli\second};\\
	\Delta(\View, \Replica_3) &= \SI{5}{\milli\second};&
	\Delta(\View, \Replica_4) &= \SI{1}{\milli\second},
\intertext{and assume that all replicas use the \emph{exponential backup function} $f(i) = 2^i$. We have:}
	\Delta(\View + 3, \Replica_1) &= 2^3 \cdot 7 = \SI{56}{\milli\second};&
	\Delta(\View + 3, \Replica_2) &= 2^3 \cdot 3 = \SI{24}{\milli\second};\\
	\Delta(\View + 3, \Replica_3) &= 2^3 \cdot 5 = \SI{40}{\milli\second};&
	\Delta(\View + 3, \Replica_4) &= 2^3 \cdot 1 = \SI{8}{\milli\second}.
\end{align*}
Hence, within three \emph{backoffs}, two out of four replicas have an internal message delay estimate that is sufficiently large. The other replicas will require four and five backoffs, respectively.
\end{example}

Next, we shall use backoff-based internal message delay estimates to prove that the view-change protocol satisfies Goal~\ref{vcg:recover}.

\begin{theorem}\label{thm:poe_timing} 
If non-faulty replicas use backoff-based internal message delay estimates, the normal-case protocol was disrupted, and communication becomes reliable, then the view-change protocol guarantees that the normal-case protocol will eventually be reestablished. 
\end{theorem}
\begin{proof}
Assume the normal-case protocol was disrupted in some view $\View$. Consequently, non-faulty replicas will perform view-changes for consecutive views $\View' \geq \View$ until the normal-case protocol is reestablished (\lfref{fig:poe_vc}{failure_next}). During the view-change for such view $\View'$, non-faulty replicas periodically rebroadcast their messages \Message{Failure}{\View'} until they reach the new-view proposal stage for $\View'$ (\lfref{fig:poe_vc}{request}). Hence, each non-faulty replica will eventually broadcast messages \Message{Failure}{\View[w]} for some view $\View[w] \geq \View$ whenever communication becomes reliable. As all $\nf$ non-faulty replicas are guaranteed to receive these messages from all $\nf$ non-faulty replicas when communication becomes reliable, all non-faulty replicas will enter the new-view proposal stage for consecutive views $\View' \geq \View$ when communication becomes reliable (\lfref{fig:poe_vc}{ep}).

Assume that communication becomes reliable with message delay $\delta$ unknown to all non-faulty replicas. Without loss of generality, we can assume that each non-faulty replica $\Replica \in (\Replicas \difference \Faulty)$ uses a backoff function $f_{\Replica}$ and uses the internal message delay estimate $\Delta(\View + i, \Replica) = f_{\Replica}(i) \cdot \Delta(\View, \Replica)$ in view $\View + i$. Hence, by Lemma~\ref{lem:poe_timing_proof}, there exists a $j_{\Replica}$ such that, for all views $j \geq j_{\Replica}$, $\Delta(\View + j, \Replica) \geq \delta$. Consequently, in all views $\View + j'$ with $j' \geq \max\{ j_{\Replica} \mid \Replica \in (\Replicas \difference \Faulty) \}$, the internal message delay estimate of all non-faulty replicas is at-least $\delta$.

Choose $i \geq \max\{ j_{\Replica} \mid \Replica \in (\Replicas \difference \Faulty) \}$ such that view $\View + i$ is a view for which all non-faulty replicas entered the new-view proposal stage after communication became reliable. Let $\Primary_{j}$ be the primary of view $\View + i + j$  (with $\ID{\Primary_{j}} = (\View + i + j) \bmod \n$). Finally, choose the first view $\View + i + \sigma$ for which $\Primary_{\sigma}$ is non-faulty. As there are at-most $\f$ faulty replicas, we must have $0 \leq \sigma \leq \f$. If the view-change protocol reestablish the normal-case protocol in any view before view $\View + i + \sigma$, then the statement of this theorem voidly holds.

To complete the proof, we assume that the view-change protocol did not reestablish the normal-case protocol before view $\View + i + \sigma$ and we show that the view-change protocol will reestablish the normal-case protocol for view $\View + i + \sigma$. Let $t(\Replica)$ be the time at which a non-faulty replica $\Replica \in (\Replicas \difference \Faulty)$ enters the new-view proposal stage of view $\View + i + \sigma - 1$ and let $t_{\min} = \min\{ t(\Replica) \mid \Replica \in (\Replicas \difference \Faulty) \}$ be the time at which the first non-faulty replica enters this new-view proposal stage. By Lemma~\ref{lem:sync_nv}, all replicas will receive an identical new-view proposal \Message{NewView}{\View + i + \sigma, \VI} (\lfref{fig:poe_vc}{request}) at-or-before $t_{\min} + 4\delta$. Finally, let $\Replica[q]$ be any non-faulty replica. Due to the internal message delay estimate used by $\Replica[q]$, replica $\Replica[q]$ will expect \Message{NewView}{\View + i + \sigma, \VI} before \[t(\Replica[q]) + 4\Delta(\View + i + \sigma, \Replica[q]) \geq t + 4 \Delta(\View + i +\sigma, \Replica[q]).\] As $i \geq j(\Replica)$, we can conclude $t_{\min} + 4 \Delta(\View+ i + \sigma, \Replica[q]) \geq t + 4\delta$, that $\Replica[q]$ receives \Message{NewView}{\View + i + \sigma, \VI}, and that $\Replica[q]$ enters the new-view accept stage. Hence, all non-faulty replicas will enter the new-view accept stage with the same new-view proposal \Message{NewView}{\View + i + \sigma, \VI}, will update their internal state accordingly, and reestablish the normal-case protocol in view $\View + i + \sigma$.
\end{proof}

Next, we further illustrate Theorem~\ref{thm:poe_timing} in an environment with exponential backoff.

\begin{example}
Consider the situation of Example~\ref{ex:vc_bc}. Due to some disruptions, view $\View$ fails and all replicas participate in consecutive view-changes (until one of these view-changes succeeds). Let $t_j(\Replica_i)$ be the time at which replica $\Replica_i$, $1 \leq i \leq 4$, enters the new-view proposal phase for view $\View + j$ and let $t_{j,\min} = \min\{ t_j(\Replica_i) \mid 1 \leq i \leq 4 \}$ be the time at which the first replica enters the new-view proposal phase for view $\View + j$.

If all replicas are non-faulty, then one can expect a valid new-view proposal for view $\View + j + 1$ at $t_{j,\min} + 4\delta$, while replica $\Replica_i$, $1\leq i \leq 4$, will expect a valid new-view proposal for view $\View + j + 1$ before \[t_j(\Replica_i) + 4\Delta(\View + j, \Replica_i) \geq t_j + 4\cdot 2^j \cdot \Delta(\View, \Replica_i).\] Consequently, the new-view proposal for view $\View+j+1$ is guaranteed to arrive on time for replica $\Replica_i$ whenever $t_{j,\min} + 4\cdot 2^j \cdot \Delta(\View, \Replica_i) \geq t_{j,\min}  + 4\delta$ and, hence, when $2^j \cdot \Delta(\View, \Replica_i) \geq \delta$, which simplifies to $j \geq \smash{\log_2\left\lceil\tfrac{\delta}{\Delta(\View, \Replica_i)}\right\rceil}$. Filling in yields
\begin{align*}
j \geq \log_2\left\lceil\tfrac{\delta}{\Delta(\View, \Replica_1)}\right\rceil = \log_2\left\lceil\tfrac{30}{7}\right\rceil = 3 &&\text{(for replica $\Replica_1$)};\\
j \geq \log_2\left\lceil\tfrac{\delta}{\Delta(\View, \Replica_2)}\right\rceil = \log_2\left\lceil\tfrac{30}{3}\right\rceil = 4 &&\text{(for replica $\Replica_2$)};\\
j \geq \log_2\left\lceil\tfrac{\delta}{\Delta(\View, \Replica_3)}\right\rceil = \log_2\left\lceil\tfrac{30}{5}\right\rceil = 3 &&\text{(for replica $\Replica_3$)};\\
j \geq \log_2\left\lceil\tfrac{\delta}{\Delta(\View, \Replica_4)}\right\rceil = \log_2\left\lceil\tfrac{30}{1}\right\rceil = 5 &&\text{(for replica $\Replica_4$)}.
\end{align*}
We conclude that within at-most five consecutive view-changes, all replicas will have a sufficiently large internal message delay estimate to guarantee a successful view-change.
\end{example}

\begin{remark}
Theorem~\ref{thm:poe_timing} will use arbitrarily large internal message delay estimates to assure that all non-faulty replicas are eventually sufficiently synchronized. In most practical deployments, one can utilize a reasonable upper bound on any message estimate (e.g., \SI{10}{\second} when sending small messages over a wide-area network), as this upper bound will always hold whenever the network is operating correctly (and communication is reliable).
\end{remark}

\subsection{\PoE{} Provides Consensus}\label{ss:poe_proofs}

In Sections~\ref{ss:poe_nc}--\ref{ss:poe_vc}, we have laid out the design of \PoE{}. From the details presented, one can already derive how \PoE{} provides non-divergence (Theorem~\subref{thm:main_nc}{fst} and Theorem~\ref{thm:vc_all}) and termination (Theorem~\subref{thm:main_cc}{fst}  and Theorem~\ref{thm:poe_timing}). To prove that \PoE{} provides weak consensus, we also need to detail how \PoE{} provides \emph{non-triviality}. To do so, we introduce a mechanism that allows non-faulty clients to force replication of their transactions whenever communication is sufficiently reliable.

Consider a client $\Client$ that wants to request transaction $\T$. In the normal-case protocol, it is assumed that $\Client$ knows primary $\Primary$ of the current view, in which case $\Client$ will simply send $\Cert{\T}{\Client}$ to $\Primary$. If this primary $\Primary$ is non-faulty and communication is reliable, then this normal-case protocol will assure a proof-of-execution (Theorem~\subref{thm:main_nc}{snd}). To deal with deviations of the normal-case (e.g., an unknown primary or a faulty primary), \PoE{} allows client $\Client$ to send its request $\Cert{\T}{\Client}$ to \emph{any} replica $\Replica$. If $\Replica$ is non-faulty, then it will forward $\Cert{\T}{\Client}$ to the primary of the current view, after which $\Replica$ expects a timely proposal of $\Cert{\T}{\Client}$ via the normal-case protocol. If no such timely proposal arrives, $\Replica$ will detect failure of the current view (\lfref{fig:poe_vc}{detect}), and, after a successful view-change, forward $\Cert{\T}{\Client}$ to the next primary.

Using the above request-forward mechanism, clients can force consensus on requests whenever communication is sufficiently reliable: to force consensus on request $\Cert{\T}{\Client}$, client $\Client$ simply sends this request to all replicas. Consequently, all non-faulty replicas will forward $\Cert{\T}{\Client}$ to the current primary and expect a timely proposal of $\Cert{\T}{\Client}$. If this does not happen, then all non-faulty replicas will detect failure of the current view. As there are $\nf$ non-faulty replicas, this is sufficient to trigger a view-change (\lfref{fig:poe_vc}{ep}). After a new view is established, the non-faulty replicas forward $\Cert{\T}{\Client}$ to the next primary, this until the request is proposed and executed. 

To prevent abuse of the request-forward mechanism by malicious clients, \PoE{} uses two heuristics. First, malicious clients can try to suppress requests by other clients by continuously sending requests to all replicas. To deal with such behavior, non-faulty replicas can limit the rate at which they forward requests of a single client to the primary (assuring that the primary has the room to timely propose requests for all clients). Second, malicious clients can send distinct, conflicting, requests to different replicas. In this case, the primary has to choose one of these transactions to propose first, due to which the primary might be unable to propose all client requests forwarded to it within a timely fashion. To prevent failure detection of the view due to this malicious client behavior, any replica $\Replica$ that forwarded a request of client $\Client$ to the current primary will consider any subsequent proposal of a transaction requested by  $\Client$ as a timely proposal by the primary (even if that proposal does not match the request forwarded by $\Replica$).

Using the request-forward mechanism, we are finally able to prove that \PoE{} provides weak consensus:

\begin{theorem}\label{thm:cons_serv}
If \PoE{} is operated in a system with $\n > 3\f$, then \PoE{} provides weak consensus.
\end{theorem}
\begin{proof}
We say that a non-faulty replica $\Replica$ decides on $\Cert{\T}{\Client}$ in round $\rn$ whenever it stores a commit certificate $\CC{m}$ with $m = \Message{Propose}{\Cert{\T}{\Client}, \View, \rn}$. Due to Theorem~\subref{thm:main_cc}{trd}, the existence of $\CC{m}$ implies that at-least $\nf - \f$ non-faulty replicas executed $m$. Using these facts, we will prove that \PoE{} provides the three guarantees of consensus.

\paragraph{\textbf{Termination} (in periods of reliable communication)}
Due to Theorem~\subref{thm:main_cc}{trd}, the existence of $\CC{m}$ implies that at-least $\nf - \f$ non-faulty replicas executed $m$. If communication is reliable, then Theorem~\subref{thm:main_cc}{snd} guarantees termination. If communication is unreliable and leads to a view failure, then Theorem~\ref{thm:poe_timing} assures that a new view, view $\View[w]$, will be established whenever communications becomes reliable and Theorem~\subref{thm:vc_all}{nv_ledger} and Lemma~\ref{lem:missing_cc} assure that all replicas participating in view $\View[w]$ will recover $\Cert{\T}{\Client}$ as the $\rn$-th decision.

Now consider non-faulty replica $\Replica[q]$ that is not able to participate in the new view $\View[w]$, e.g., due to faulty behavior of the primary of view $\View[w]$. Replica $\Replica[q]$ can make a decision for round $\rn$ in two ways. First, if the primary of view $\View[w]$ can propose in round $\rn$, then either a proposal $m' = \Message{Propose}{\Cert{\T}{\Client}, \View[w], \rn}$ gets committed by a non-faulty replica, in which case Theorem~\subref{thm:main_cc}{fst} guarantees that $\Replica[q]$ commits $m'$, or the failure of view $\View[w]$ will be detected. If the primary of view $\View[w]$ cannot propose in round $\rn$, then $\Replica[q]$ can use Lemma~\ref{lem:missing_cc} to derive $\Cert{\T}{\Client}$ as the $\rn$-th request whenever $\Replica[q]$ detects the existence of a commit certificate for round $\rn$. Replica $\Replica[q]$ can detect the existence of commit certificates for round $\rn$, $\rn < \rn'$, after receiving well-formed \MName{CheckCommit} messages for round $\rn' $ from at-least a single non-faulty replica (by receiving such messages from $\f+1$ distinct replicas).

\paragraph{\textbf{Non-Divergence}}
Let $\View'$ be the first view in which a proposal $m' = \Message{Propose}{\Cert{\T'}{\Client'}, \View', \rn}$ was executed by $\nf - \f$ non-faulty replicas. By Theorem~\subref{thm:vc_all}{no_commit}, $\View' \leq \View$. Hence, by Theorem~\subref{thm:vc_all}{prepared_nodiv}, we must have $\Cert{\T'}{\Client'} = \Cert{\T}{\Client}$.

\paragraph{\textbf{Non-Triviality} (in periods of reliable communication)} 
We say that a decision needs to be made for request $\Cert{\T}{\Client}$ of client $\Client$ whenever $\Client$ is non-faulty. In this case, a non-faulty replica $\Replica[q]$ learns that a decision needs to be made for $\Cert{\T}{\Client}$ whenever it receives $\Cert{\T}{\Client}$. As non-faulty clients will repeatedly send $\Cert{\T}{\Client}$ to all replicas whenever they do not receive a proof-of-execution for $\T$, every non-faulty replica will receive $\Cert{\T}{\Client}$ when communication becomes reliable. Assume that communication becomes reliable in view $\View$. Either the primary of view $\View$ will propose $\Cert{\T}{\Client}$, or all non-faulty replicas will trigger view changes until reaching a view whose primary will propose $\Cert{\T}{\Client}$. As there are at-most $\f$ faulty replicas, at-most $\f$ such view changes will happen before $\Cert{\T}{\Client}$ is proposed.
\end{proof}

\begin{remark}
\PoE{} is designed to be able to deal with many types of failures, e.g., Byzantine failures, network failures, and crashes. First, as proven in Theorem~\ref{thm:cons_serv}, \PoE{} provides both \emph{consensus} and \emph{client services} if at-most $\f$ concurrent replicas fail (e.g., Byzantine replicas or crashed replicas), $\n > 3\f$. Second, as a consequence of Theorem~\subref{thm:main_nc}{fst} and Theorem~\ref{thm:vc_all}, \PoE{} will always provide \emph{non-divergence} if at-most $\f$ replicas are Byzantine. Third, \PoE{} can recover form any number of non-Byzantine replica crashes \emph{as long as} these crashes are \emph{recoverable}: if replicas store any certificates they produce in permanent store (\lfref{fig:poe_nc}{store} and \lfref{fig:poe_cc}{store}) such that they can always recover these certificates after a crash, then non-Byzantine crashes can only lead to service disruption (similar to network failure), and service will be recovered when sufficient replicas have recovered.
\end{remark}

\subsection{\PoE{} Provides Client Service}\label{ss:poe_client}

Theorem~\ref{thm:cons_serv} proves that \PoE{} provides weak consensus. Hence, due to the non-triviality guarantee of \PoE{} (as proven in Theorem~\ref{thm:cons_serv}), every client can force a consensus decision on their requests. Hence, to provide \emph{client services}, we only need to show that \PoE{} guarantees that a consensus decision for request $\Cert{\T}{\Client}$ will lead to a proof-of-execution for $\T$. If communication is reliable and no view changes happen, then either Theorem~\subref{thm:main_nc}{snd} or Theorem~\subref{thm:main_cc}{snd} guarantees such proof-of-execution. Due to the speculative design of \PoE{}, view changes can result in consensus decisions \emph{without} an accompanying proof-of-execution:

\begin{example}\label{ex:poe_fail}
Consider a system with $\n = 3\f+1$ replicas. We partition the non-faulty replicas in sets $A$, $B$, and $C$ with $\abs{A} = \abs{B} = \f$ and $\abs{C} = 1$. Now let $\Cert{\T}{\Client}$ be a client request and consider the following sequence of events in round $\rn$ of views $\View$, $\View + 1$, and $\View + 2$.

\paragraph{View $\View$} Due to \emph{unreliable communication} the $\f$ replicas in $B$ become unreachable while the primary $\Primary_{\View}$ of view $\View$ proposes $m_{\View} = \Message{Propose}{\Cert{\T}{\Client}, \View, \rn}$. Hence, only the $\nf$ replicas in $A \union C \union \Faulty$ receive $m_{\View}$ and exchange \MName{Prepare} messages. Consequently, the $\nf - \f$ non-faulty replicas in $A \union C$ are able to execute $\T$ and inform the client $\Client$. The faulty replicas in $\Faulty$ decide to \emph{not} inform the client. Hence, the client $\Client$ only receives $\f+1$ messages of the form \Message{Inform}{\Cert{\T}{\Client}, \View, \rn, r}, which is insufficient for a proof-of-execution.

Next, the replicas in $A \union C$ exchange \MName{CheckCommit} messages, while the faulty replicas only send \MName{CheckCommit} messages to $A$. 	Consequently, the replicas in $A$ are able to construct commit certificates $\CC{m_{\View}}$, while the replicas in $C$ are unable to do so. When communication becomes reliable, the $\f+1$ non-faulty replicas in $B \union C$ will be able to successfully trigger a view change.

\paragraph{View $\View +1$} 
Due to \emph{unreliable communication}, the $\f$ non-faulty replicas in $A$ become unreachable. Hence, the primary $\Primary_{\View+1}$ of view $\View+1$ ends up proposing a new view based on the information provided by the replicas in $B \union C \union \Faulty$. In specific, the non-faulty replicas in $B$ only provide information on rounds before $\rn$, the non-faulty replica in $C$ provides a prepared certificate for $m_{\View}$, and the faulty replicas in $\Faulty$ \emph{lie} and only provide information on rounds before $\rn$. Hence, any replica that enters view $\View +1$ will rollback transaction $\T$ and will expect the primary $\Primary_{\View+1}$ to propose  $m_{\View+1} = \Message{Propose}{\Cert{\T}{\Client}, \View+1, \rn}$.

After primary $\Primary_{\View+1}$ proposes $m_{\View+1}$, the  $\nf$ replicas in $B \union C \union \Faulty$ receive $m_{\View+1}$ and exchange \MName{Prepare} messages. Consequently, the $\nf - \f$ non-faulty replicas in $B \union C$ are able to execute $\T$ and inform the client $\Client$. The faulty replicas in $\Faulty$ once again decide to \emph{not} inform the client. Hence, at this point, the client received $\f+1$ messages of the form \Message{Inform}{\Cert{\T}{\Client}, \View, \rn, r} and $\f+1$ messages of the form \Message{Inform}{\Cert{\T}{\Client}, \View+1, \rn, r}, which is still insufficient for a proof-of-execution.

Next, the replicas in $B \union C$ exchange \MName{CheckCommit} messages, while the faulty replicas only send \MName{CheckCommit} messages to $B$. 	Consequently, the replicas in $B$ are able to construct commit certificates $\CC{m_{\View+1}}$, while the replicas in $C$ are unable to do so. When communication becomes reliable, the $\f+1$ non-faulty replicas in $A \union C$ will be able to successfully trigger a view change.

\paragraph{View $\View +2$} 
The (faulty) primary $\Primary_{\View+2}$ of view $\View+2$ will receive commit certificates $\CC{m_{\View}}$ from the replicas in $A$, commit certificates $\CC{m_{\View+1}}$ from the replicas in $B$, a prepared certificate for $m_{\View+1}$ from $C$, and, again, the faulty replicas in $\Faulty$ \emph{lie} and only provide information on rounds before $\rn$. Using this information, the primary $\Primary_{\View+2}$ decides to construct \emph{two} distinct new-view proposals $n_A$ and $n_B$ that are based on the information provided by the $\nf$ replicas in $A \union C \union \Faulty$ and $B \union C \union \Faulty$, respectively.

The replicas in $A$ receive $n_A$, the replicas in $B$ receive $n_B$, and replica $C$ also receives $n_B$ (due to which $C$ only has to store an additional commit certificate). Hence, after the view-change, at-most $\f+1$ non-faulty replicas hold identical commit certificates for round $\rn$,  while \emph{all} non-faulty replicas hold a commit certificate for a proposal proposing $\Cert{\T}{\Client}$ in round $\rn$ and no non-faulty replica will send additional \MName{Inform} messages for the request $\Cert{\T}{\Client}$ to client $\Client$. As such, the client $\Client$ fails to obtain a proof-of-execution for a fully-decided request.

As the final step in this example, the (faulty) primary $\Primary_{\View+2}$ successfully proposes $m = \Message{Propose}{\Cert{\T'}{\Client'}, \View+2, \rn+1}$. Due to this proposal, all non-faulty replicas end up with commit certificates $\CC{m}$, thereby assuring that future view changes will omit any information on round $\rn$ and, thus, assuring that no proof-of-execution will be produced in future views.
\end{example}

To deal with the issue raised in Example~\ref{ex:poe_fail}, we utilize a recovery protocol by which clients can obtain a proof-of-commit that \emph{implies} a proof-of-execution. This recovery protocol is based on a straightforward observation: if a client $\Client$ can observe that a non-faulty replica \emph{committed} a proposal $\Message{Propose}{\Cert{\T}{\Client}, \View, \rn}$ for request $\Cert{\T}{\Client}$, then, by Theorem~\ref{thm:cons_serv}, this observation provides proof that all non-faulty replicas will eventually execute $\Cert{\T}{\Client}$ in round $\rn$.

Note that as part of providing non-triviality, a client $\Client$ that requires execution for a request $\Cert{\T}{\Client}$ will eventually send $\Cert{\T}{\Client}$ to all replicas, this to enforce an eventual consensus decision. When a non-faulty replica $\Replica$ receives $\Cert{\T}{\Client}$ after the replica stored a commit certificate for a proposal proposing $\Cert{\T}{\Client}$, then, as the first step of the proof-of-commit recovery protocol, replica $\Replica$ responds with a message \Message{InformCC}{\Cert{\T}{\Client}, \rn, r}, in which $r$ was the original execution result $\Replica$ obtained while executing $\T$ in round $\rn$. As the second step of the proof-of-commit recovery protocol, client $\Client$ will considers $\T$ executed after it receives a \emph{proof-of-commit} for $\Cert{\T}{\Client}$ consisting of identical \Message{InformCC}{\Cert{\T}{\Client}, \rn, r} messages from $\f+1$ distinct replicas. The pseudo-code for the proof-of-commit recovery protocol can be found in Figure~\ref{fig:poe_poc}.\footnote{Non-faulty replicas in \PoE{} can opt to \emph{always} send an \MName{InformCC} message to clients after they commit a proposal: \MName{InformCC} messages carry the same information to the client as normal client replies do in consensus protocols such as \PBFT{}. Such a design will trade an increase of network bandwidth for a decrease of latencies when replicas are left in the dark.}

\begin{figure}[t]
    \begin{myprotocol}
        \TITLE{Client recovery role}{used by client $\Client$ to force transaction $\T$}
        \STATE Send $\Cert{\T}{\Client}$ to all replicas $\Replica \in \Replicas$.
        \STATE Await a \emph{$\rn$-proof-of-commit for $\Cert{\T}{\Client}$} consisting of identical messages \Message{InformCC}{\Cert{\T}{\Client}, \rn, r} from $\f+1$ distinct replicas.\label{fig:poe_poc:await}
        \STATE Considers $\T$ executed, with result $r$, as the $\rn$-th transaction.\label{fig:poe_poc:cc}
        \SPACE
        \TITLE{Replica recovery role}{running at every replica $\Replica \in \Replicas$}\label{fig:poe_poc:replica}
        \EVENT{$\Replica$ receives $\Cert{\T}{\Client}$}
            \IF{$\Replica$ stored commit certificate $\PC{m}$ for some proposal\\
            \phantom{\textbf{event }}$m = \Message{Propose}{\Cert{\T}{\Client}, \View, \rn}$}
                \STATE Let $r$ be the result $\Replica$ obtained while executing $\T$ in round $\rn$.
                \STATE Send \Message{InformCC}{\Cert{\T}{\Client}, \rn, r} to $\Client$.\label{fig:poe_poc:inform}
            \ENDIF
        \ENDEVENT
    \end{myprotocol}
    \caption{The proof-of-commit recovery protocol in \PoE{}.}\label{fig:poe_poc}
\end{figure}

Using the proof-of-commit recovery protocol, we can finally prove that \PoE{} provides client services: 
\begin{theorem}
If \PoE{} is operated in a system with $\n > 3\f$, then \PoE{} provides client service whenever communication is reliable.
\end{theorem}
\begin{proof}
When communication is reliable and client $\Client$ requests $\Cert{\T}{\Client}$, then, by Theorem~\ref{thm:cons_serv}, client $\Client$ can force execution of $\Cert{\T}{\Client}$. If client $\Client$ receives a $(\View, \rn)$-proof-of-execution for $\T$, then client $\Client$ considers $\T$ executed and, due to Theorem~\ref{thm:vc_all}, the executed state observed by the client will be preserved.

Next, we consider the case in which client $\Client$ does not receive a proof-of-execution for $\T$. Due to Theorem~\ref{thm:cons_serv}, all non-faulty replicas will eventually commit a proposal for $\Cert{\T}{\Client}$ in some round $\rn$. By Theorem~\subref{thm:main_cc}{trd} and Corollary~\ref{cor:ass}, all non-faulty replicas executed the same sequence of $\rn-1$ transactions before executing $\T$. Hence, if client $\Client$ resends $\Cert{\T}{\Client}$ to the non-faulty replicas, then these non-faulty replicas will eventually all be able to send identical messages \Message{InformCC}{\Cert{\T}{\Client}, \rn, r} to the client (\lfref{fig:poe_poc}{inform}). Consequently, the client $\Client$ will receive at-least $\nf > \f$ such messages and client $\Client$ will consider $\T$ executed (\lfref{fig:poe_poc}{cc}). Due to Theorem~\ref{thm:vc_all}, the executed state observed by the client will be preserved. Finally, as there are at-most $\f$ faulty replicas, the faulty replicas can produce at-most \Message{InformCC}{\Cert{\T}{\Client}, \rn', r'} messages distinct form \Message{InformCC}{\Cert{\T}{\Client}, \rn, r} and, hence, cannot produce an invalid proof-of-commit.
\end{proof}

\begin{remark}
\PoE{} does not enforce that client requests are proposed \emph{in order}: if a client $\Client$ first sends $\Cert{\T_1}{\Client}$ to the current primary $\Primary$ and then sends $\Cert{\T_2}{\Client}$, then nothing prevents $\Primary$ of first proposing $\Cert{\T_2}{\Client}$ in round $\rn_2$ and then proposing $\Cert{\T_1}{\Client}$ in round $\rn_1$ ($\rn_2 < \rn_1$). If a client requires ordered execution, then it can simply request $\T_2$ only after it receives a proof-of-execution of $\T_1$. If the application that utilizes \PoE{} requires ordered execution of all client requests, then one can require that valid client requests come with a baked-in counter indicating its order and that subsequent proposals use subsequent counter values.
\end{remark}

\section{On the Complexity of \PoE{}}\label{sec:complexity}

In Section~\ref{sec:poe} we provided an in-detail description and proof of correctness of the Proof-of-Execution consensus protocol (\PoE{}). Next, we will take a deep dive into the exact message complexity of this protocol. Furthermore, as we kept the presentation of \PoE{} described in Section~\ref{sec:poe} as simple as possible, we will also introduce techniques one can employ to obtain \PoE{} variants with a reduced message complexity. We have summarized the complexity of \PoE{} and its variants in Figure~\ref{fig:complex_summary}.

\begin{figure*}
\begin{minipage}{\textwidth}
\renewcommand{\thefootnote}{\alph{footnote}}
\centering
\small
\begin{tabular}{lccccc}
\toprule
Element&Example&\multicolumn{3}{c}{Size}&Amount\\
\cmidrule(l{4pt}r{4pt}){3-5}
&&Standard \PoE{}&Using Digests&Using Threshold Certificates &  \\
&&(Section~\ref{ss:std})&(Section~\ref{ss:fine})&(Section~\ref{ss:tsc})&\\
\midrule
Client request & $\Cert{\T}{\Client}$ & $\CSize$ & $\CSize$ & $\CSize$ & 1\\
\midrule
\MName{Propose} message& $m = \Message{Propose}{\Cert{\T}{\Client}, \View, \rn}$ & $\BigO{\CSize}$& $\BigO{\CSize}$& $\BigO{\CSize}$& $\n-1$ \\
\MName{Prepare} message& \Message{Prepare}{m} & $\BigO{\CSize}$ & $\BigO{1}$ & $\BigO{1}$ & $\n(\n-1)$ \footnotemark[1]\\
\MName{Inform} message & \Message{Inform}{\Cert{\T}{\Client}, \View, \rn, r} & $\BigO{\CSize + \Size{r}}$ & $\BigO{1 + \Size{r}}$ & $\BigO{1 + \Size{r}}$ & $\n$\\
\midrule
\MName{CheckCommit} message & \Message{CheckCommit}{\PC{m}} & $\BigO{\nf + \CSize}$ & $\BigO{1}$ \footnotemark[2] & $\BigO{1}$ \footnotemark[2] & $\n(\n-1)$\\
\midrule
\MName{Failure} message & \Message{Failure}{\View} & $\BigO{1}$ & $\BigO{1}$ & $\BigO{1}$ & $\n(\n-1)$\\
\MName{ViewState} message & \Message{ViewState}{\View, \CC{m}, \VE} & $\BigO{\nf^2 + \WSize(\nf + \CSize)}$ & $\BigO{\WSize\nf}$ \footnotemark[2] & $\BigO{\WSize}$ \footnotemark[2] & $(\n-1)$\\
\MName{NewView} message & \Message{NewView}{\View + 1, \VI} &
$\BigO{\nf\left(\nf^2 + \WSize(\nf + \CSize)\right)}$ &
$\BigO{\WSize\nf^2}$ \footnotemark[2] &
$\BigO{\WSize\nf}$ \footnotemark[2]& $(\n-1)$\\
\midrule
\MName{QueryCC} message & \Message{QueryCC}{\rn} & $\BigO{1}$ & $\BigO{1}$ & $\BigO{1}$ & $\n$ \footnotemark[3] \\
\MName{RespondCC} message &\Message{RespondCC}{\PC{m}, \CC{m}} & $\BigO{\nf^2+\CSize}$ & $\BigO{\nf}$ \footnotemark[2] & $\BigO{1}$ \footnotemark[2] & $\n$ \footnotemark[3]\\
\MName{InformCC} message & \Message{InformCC}{\Cert{\T}{\Client}, \rn, r} & $\BigO{\CSize + \Size{r}}$ & $\BigO{1 + \Size{r}}$ & $\BigO{1 + \Size{r}}$ & $\n$ \footnotemark[3]\\
\midrule
Prepared certificate & $\PC{m}$ & $\BigO{\nf + \CSize}$ & $\BigO{\nf}$ \footnotemark[4] & $\BigO{1}$ \footnotemark[4] & \\
Commit certificate & $\CC{m}$ & $\BigO{\nf^2 + \CSize}$ & $\BigO{\nf}$ \footnotemark[4] & $\BigO{1}$ \footnotemark[4] & \\
\bottomrule
\end{tabular}
\footnotetext[1]{The \MName{Prepare} message of the primary can be merged into the \MName{Propose} it broadcasts, reducing the amount of \MName{Prepare} messages to $(\n-1)^2$.}
\footnotetext[2]{The complexities noted are for versions that minimize message size at the cost of additional recovery steps for replicas to obtain missing client requests. See Remark~\ref{rem:simple_recovery} for details.}
\footnotetext[3]{These messages are only required for individual recovery processes and are not used when communication is sufficiently reliable.}
\footnotetext[4]{As client requests are not part of these certificates, replicas need to store a copy of the client request alongside these certificates.}
\end{minipage}
\caption{The message complexity (of messages exchanged) and storage complexity (of certificates stored) in the standard variants of \PoE{}. We assume that client requests have a size bounded by $\CSize$, that the primary can propose at-most $\WSize$ requests out-of-order after the last commit certificate (the window size), and that the execution result $r$ has size $\Size{r}$.}\label{fig:complex_summary}
\end{figure*}

\subsection{The Complexity of Standard \PoE{}}\label{ss:std}

As the first step, we take a look at the complexity of the \emph{standard} variant of \PoE{} as it is presented in Section~\ref{sec:poe}. To simplify presentation, we have assumed in Section~\ref{sec:poe} that replicas send messages to themselves whenever they broadcast messages. In practice, these messages can be eliminated.

In the normal case protocol of Figure~\ref{fig:poe_nc}, the primary broadcasts $\n - 1$ messages $m = \Message{Propose}{\Cert{\T}{\Client}, \View, \rn}$ whose size is directly determined by the size $\CSize = \Size{\Cert{\T}{\Client}}$ of the proposed client request $\Cert{\T}{\Client}$. Next, all $\n$ replicas broadcast messages \Message{Prepare}{m} to all other $(\n-1)$ replicas, resulting in $\n(\n-1)$ messages of size $\BigO{\CSize}$. Finally, all $\n$ replicas respond to the client via messages \Message{Inform}{\Cert{\T}{\Client}, \View, \rn, r}, whose size is directly determined by the size $\CSize$ of a client request and the size $\Size{r}$ of the result $r$.

In the check-commit protocol of Figure~\ref{fig:poe_cc}, all replicas broadcast a message \Message{CheckCommit}{\PC{m}}, in which $\PC{m}$ is a prepared certificate for proposal $m$. This prepared certificate $\PC{m}$ consists of $\nf$ messages \Message{Prepare}{m} signed by $\nf$ distinct replicas. As these $\nf$ messages are identical, only a single copy needs to be included (with size $\BigO{\CSize}$) together with $\nf$ signatures that have a constant size each.\footnote{To simplify presentation in Section~\ref{sec:poe}, the \MName{CheckCommit} messages carry prepared certificates. Without affecting the correctness of \PoE{} or the complexity of recovery, these prepared certificates can be eliminated in favor of constant-sized \MName{CheckCommit} messages. We refer to Section~\ref{ss:fine} and Section~\ref{ss:mac} for further details.}

In the view-change protocol of Figure~\ref{fig:poe_vc}, all at-most $\n$ replicas that detect failure of view $\View$ broadcast a message \Message{Failure}{\View} of constant size to all other $(\n-1)$ replicas.  All at-most $\n$ replicas that enter the new-view proposal stage in view $\View$ will send a single message of the form \Message{ViewState}{\View, \CC{m}, \VE}, with $\CC{m}$ a commit certificate and $\VE$ a set of prepared certificates, to the new primary. The commit certificate $\CC{m}$ consists of $\nf$ messages \Message{Commit}{\PC{m}_i}, $1 \leq i \leq \nf$, signed by $\nf$ distinct replicas. As each replica can construct a prepared certificate $\PC{m}_i$ distinct from all other prepared certificates and the signature on \Message{Commit}{\PC{m}_i} can only be verified when the message \Message{Commit}{\PC{m}_i} is known, a commit certificate needs to store the $\nf$ individual prepared certificates: one can only eliminate $\nf-1$ redundant copies of the original proposal $m$. Hence, each commit certificate consists of a proposal $m$, $\nf^2$ signatures of prepare message \Message{Prepare}{m}, and $\nf$ signatures of messages \Message{Commit}{\PC{m}_i}, $1 \leq i \leq \nf$. By Assumption~\ref{ass:size}, the number of prepared certificates in $\VE$ is upper bounded by some \emph{window size} $\WSize$. Finally, the new primary will broadcast a single message of the form \Message{NewView}{\View + 1, \VI}, with $\VI$ a set of $\nf$ \MName{ViewState} messages, to all $(\n-1)$ other replicas.  

Finally, in the query protocol of Figure~\ref{fig:query_protocol}, each replica that needs to recover the status of round $\rn$ after a period of unreliable communication will broadcast constant-sized messages of the form \Message{QueryCC}{\rn} and get responses \Message{RespondCC}{\PC{m}, \CC{m}} with $\PC{m}$ a prepared certificate and $\CC{m}$ a commit certificate. In the proof-of-commit recovery protocol of Figure~\ref{fig:poe_poc}, each replica will send a message \Message{InformCC}{\Cert{\T}{\Client}, \rn, r} whose size is determined by $\CSize$ to clients that are recovering from unreliable communication.

\subsection{Reducing Message Sizes with Digests}\label{ss:fine}

A close look at \PoE{} shows that in the original design of \PoE{}, full copies of \MName{Propose} messages are included in \MName{Prepare}, \MName{CheckCommit}, \MName{ViewState}, and \MName{NewView} messages. Furthermore, full copies of the client request are included in \MName{Inform} and \MName{InformCC} messages.

When client requests are large, this will incur a substantial communication cost. The typical way to support large client requests is by replacing client requests $\Cert{\T}{\Client}$ by a constant-sized message digest $\Digest{\Cert{\T}{\Client}}$ obtained from $\Cert{\T}{\Client}$ using a strong cryptographic hash function $\Digest{\cdot}$~\cite{cryptobook}. To do so, one simply replaces the message \Message{Propose}{\Cert{\T}{\Client}, \View, \rn} the primary broadcasts when proposing transaction $\T$, requested by client $\Client$, as the $\rn$-th transaction in some view $\View$ by the pair $(m_d = \Message{Propose}{d, \View, \rn}, \Cert{\T}{\Client})$ such that $d = \Digest{\Cert{\T}{\Client}}$. All other messages will use the constant-sized proposal $m_d$, which only includes the constant-sized digest $d$, in all messages that contain a \MName{Propose} message. Furthermore, also the client request $\Cert{\T}{\Client}$ in \MName{Inform} and \MName{InformCC} messages can be replaced by the constant-sized digest $d$.

Receipt of a message \Message{CheckCommit}{\PC{m_d}} is no longer sufficient to prepare and execute client request $\Cert{\T}{\Client}$ (\lfref{fig:poe_cc}{prepare_exec}): after receiving message \Message{CheckCommit}{\PC{m_d}}, one only has a digest $d = \Digest{\Cert{\T}{\Client}}$ of client request $\Cert{\T}{\Client}$, but not the client request itself. Fortunately, the existence of prepared certificate $\PC{m_d}$ guarantees that at-least $\nf-\f$ non-faulty replicas have sent \Message{Prepare}{m_d} and, hence, have received a pair $(m_d, \Cert{\T}{\Client})$. Hence, after a replica $\Replica[q]$ receives prepared certificate $\PC{m_d}$, it can query the replicas whose signed prepared certificates are in $\PC{m_d}$ for the missing client request. As $\Replica[q]$ can derive the digest $d$ from $\PC{m_d}$, it can verify whether any request $\Cert{\T'}{\Client'}$ it received in response to its query is valid by verifying whether $d = \Digest{\Cert{\T'}{\Client'}}$ (under normal cryptographic assumptions). Similar strategies can be made to recover missing client requests due to \MName{NewView} messages. As such, the usage of digests to eliminate copies of client requests trades a  reduction of message complexity with a more-complex recovery path.

\begin{remark}
Although digests reduce the impact of \emph{very large} client requests on throughput, they do not eliminate the need for a limit on the size of client requests as specified in Assumption~\ref{ass:size}: if the size of client requests is not limited, then non-faulty replicas cannot expect a timely response of \emph{any} forwarded client requests, as the primary could always be busy with sending an arbitrarily large proposal.
\end{remark}

The second source of large messages in \PoE{} are \MName{CheckCommit} messages, which impacts the size of commit certificates and of \MName{ViewState}, \MName{NewView}, and \MName{RespondCC} messages. To further reduce the size of \MName{CheckCommit} messages, one can replace messages of the form \Message{CheckCommit}{\PC{m}}, with $\PC{m}$ a prepared certificate for proposal $m$, by messages of the form \Message{CheckCommit}{m_d} (where $m_d$ is the proposal obtained from $m$ when the client request is represented by a digest). Doing so will reduce the size of a \MName{CheckCommit} message from $\BigO{\nf + \CSize}$ to $\BigO{1}$. Due to this change, all commit certificates $\CC{m_d}$ will consists of $\nf$ messages \Message{CheckCommit}{m_d} signed by $\nf$ distinct replicas. As these $\nf$ messages are identical, only a single copy needs to be included in $\CC{m_d}$ (with size $\BigO{1}$) together with $\nf$ signatures that have a constant size each, thereby reducing the size of commit certificates from $\BigO{\smash{\nf^2 + \CSize}}$ to $\BigO{\nf}$.

Eliminating prepared certificates from \MName{CheckCommit}  messages, by replacing messages of the form \Message{CheckCommit}{\PC{m}} by messages of the form \Message{CheckCommit}{m_d}, does impact the recovery mechanism of \lfref{fig:poe_cc}{prepare} by which replicas can prepare and execute proposals: a \emph{single} message \Message{CheckCommit}{m_d} not only lacks a client request, it also lacks a prepared certificate that provides the necessary information to obtain $\Cert{\T}{\Client}$, to prepare $m_d$, and to execute $\Cert{\T}{\Client}$. To deal with this lack of information, a replica $\Replica[q]$ that wants to use the recovery mechanism of \lfref{fig:poe_cc}{prepare} has to wait for identical \Message{CheckCommit}{m_d} messages signed by $\f+1$ distinct replicas, after which $\Replica[q]$ has a guarantee that at-leas one of these signing replicas is non-faulty and can be queried for the necessary prepared certificate $\PC{m_d}$ and client request $\Cert{\T}{\Client}$. (For completeness, we note that this change in the number of required \MName{CheckCommit} messages necessary to recover prepared certificates does not invalidate the core properties of the check-commit protocol proven in Theorem~\ref{thm:main_cc}). As such, also this simplification of \MName{CheckCommit} messages trades a reduction of message complexity with a more-complex recovery path.

\begin{remark}\label{rem:simple_recovery}
One can strike other balances between message complexity and the complexity of the recovery path than the approach outlined above. E.g., we can retain most message size benefits, while keeping the recovery path simple with a slightly different approach. In specific,
\begin{enumerate}
\item To fully retain the recovery path of \lfref{fig:poe_cc}{prepare_exec}, while also reducing the size of commit certificates to $\BigO{\nf}$, one can replace messages of the form \Message{CheckCommit}{\PC{m}} by triples of the form $(\Message{CheckCommit}{m_d}, \PC{m_d}, \Cert{\T}{\Client})$, in which the signed constant-sized message \Message{CheckCommit}{m_d} is used in the construction of commit certificates of size $\BigO{\nf}$, while the prepared certificate $\PC{m_d}$ and the client request $\Cert{\T}{\Client}$ can be used for the recovery path of \lfref{fig:poe_cc}{prepare_exec}.
\item To assure that the new primary has access to all necessary client requests, while also reducing the size of \MName{ViewState} and, consequently, \MName{NewView} messages, one can replace messages of the form \Message{ViewState}{\View, \CC{m}, \VE} by pairs of the form $(\Message{ViewState}{\View, \CC{m_d}, \VE_d}, S)$ in which $S$ is the set of client requests proposed by $m_d$ and by the proposals $m_d'$ with $\PC{m_d'} \in \VE$.
\item To assure that replicas have knowledge of all client requests necessary to start a new view, while also reducing the size of \MName{NewView} mesages, the new primary can replace messages of the form \Message{NewView}{\View+1, \VI} by pairs of the form $(\Message{NewView}{\View+1, \VI_d}, \Cert{\T}{\Client})$, with $n_d = \Message{NewView}{\View+1, \VI_d}$ and  $\Ledger{n_d}[\NVrncc{n_d}] = \Cert{\T}{\Client}$ the client request proposed by any commit certificate included in $\VI$ in the last round $\NVrncc{n_d}$ for which $\VI$ includes commit certificates. We note that the \MName{NewView} message $n_d$ only requires this single client request to convey sufficient information: all client requests for rounds $\rn$, $\NVrncc{n_d} < \rn \leq \NVrn{n_d}$ need to be re-proposed by the new primary, while all replicas can validate whether the new primary proposed the right requests via the digests available in $n_d$.
\end{enumerate}
\end{remark}

\subsection{Threshold Signatures for Certificates}\label{ss:tsc}

A further look at the usage of certificates in \PoE{} shows that the prepared and commit certificates, which are used in several messages, will grow large in deployments with many replicas and will cause high storage costs and communication costs in such deployments. E.g., consider the prepared and commit certificates for a \MName{Propose} message $m = \Message{Propose}{\Cert{\T}{\Client}, \View, \rn}$. Even when using digests, both $\PC{m}$ and $\CC{m}$ will consists of a set of $\nf$ distinct (constant-sized) signatures. Hence, these certificates have a size that is linear with respect to the number of (non-faulty) replicas.

To further reduce the size of certificates, one can employ \emph{threshold signatures}~\cite{rsasign,eccsign}.
\begin{definition}
A \emph{$x:y$ threshold signature scheme} is a digital signature scheme for a set $X$ of $\abs{X} = x$ participants in which any set $Y \subseteq X$ of $\abs{Y} = y$ participants can cooperate to produce a certificate for any given value $v$. To do so, each participant in $Y$ produces their signature share for $v$. Using only valid signature shares for $v$ produced by at-least $y$ distinct participants, anyone can produce a valid certificate for $v$.

In specific, each participant $p \in X$ will receive a distinct private key $k(p)$ that $p$ can use to sign any value $v$, resulting in a \emph{signature share} $\PCert{v}{p}$. As with traditional public-key cryptography, the signature share of $p$ can only be produced using $k(p)$ and anyone can verify the authenticity of signature share $\PCert{v}{p}$ using the public key associated with $k(p)$. Given a set of signature shares $\{ k(p) \mid p \in Y \}$, one can produce a single constant-sized \emph{threshold signature} that certifies that value $v$ was signed by at-least $y$ distinct participants from the set  $X$ of participants in the threshold signature scheme.
\end{definition}

In the setting of \PoE{}, one can use a $\n:\nf$ threshold signature scheme in which each replica is a participant. Next, instead of using normal digital signatures to sign a message $m'$, each replica $\Replica$ uses their private key $p(\Replica)$ to produce signature shares $\PCert{m}{\Replica}$ for $m'$. These signature shares can be used to provide  authenticated communication. If a replica receives $\nf$ signature shares for a given \MName{Prepare} message for proposal $m$, then one can use these signature shares to produce a prepared certificate for $m$ of constant size. Likewise, one can use $\nf$ signature shares for a given \MName{CheckCommit} message to produce a commit certificate of constant size.

\subsection{A Linear Proof-of-Execution}\label{ss:lin}

Although using digests (Section~\ref{ss:fine}) or threshold certificates (Section~\ref{ss:tsc}) drastically reduces the size of individual messages, they do not change the amount of messages being sent: all replicas will broadcast \MName{Prepare} and \MName{CheckCommit} messages in the normal case, a \emph{quadratic amount}. Next, we show how to further apply threshold signatures to assure that each phase of communication only uses a linear amount of messages. We note that such applications of threshold signatures are common among modern variants of \PBFT{} such as \SBFT{}~\cite{sbft}, \LinBFT{}~\cite{linbft}, and \HS{}~\cite{hotstuff}. Hence, we shall mainly focus on the novel design necessary to make \emph{all} normal-case communication of \PoE{} linear, this including the decentralized check-commit protocol.  Next, we detail the design of Linear-\PoE{}.

The normal-case of \PoE{} consists of \emph{two} phases of quadratic communication: in the \emph{first} phase, \MName{Prepare} messages are exchanged, while in the \emph{second} phase \MName{CheckCommit} messages are exchanged. These two phases are rather different in their design: the \emph{prepare phase} (the first phase) follows the traditional primary-backup design of \PBFT{} and only has to succeed \emph{when the primary is non-faulty}, whereas the \emph{check-commit phase} (the second phase) is fully decentralized in the sense that it should succeed \emph{independent of any non-faulty behavior} in all rounds that do not lead to primary failure.

To make the prepare phase linear, we can apply the well-known transformation from \emph{all-to-all} communication to \emph{all-to-one-to-all} communication using threshold signatures (e.g.,~\cite{linbft,sbft,hotstuff,mc_2021}). To do so, we replace the prepare phase by two subphases. In specific, upon arrival of the \emph{first} proposal for round $\rn$ of view $\View$ via some \MName{Propose} message pair $(m_d = \Message{Propose}{d, \View, \rn}, \Cert{\T}{\Client})$, with $d = \Digest{\Cert{\T}{\Client}}$, each (non-faulty) replica $\Replica$ that received this pair will enter the \emph{prepare phase} for $m_d$. As the first subphase of the prepare phase, $\Replica$ supports the proposal $m_d$ by sending a reply message \Message{Support}{\PCert{\Message{Prepare}{m_d}}{\Replica}} to the primary. As the last subphase, the primary collects well-formed messages of the form \Message{Support}{\PCert{\Message{Prepare}{m_d}}{\Replica[q]}}, $\Replica[q] \in S, S \subset \Replicas$, from a set of $\nf = \abs{S}$ replicas. Next, the primary uses the provided signature shares to produce a constant-sized prepared certificate $\PC{m_d}$. Finally, the primary broadcasts $\PC{m_d}$ to all replicas via a message \Message{Certify}{\PC{m_d}}. After replicas receive \Message{Certify}{\PC{m_d}}, they use the prepared certificate $\PC{m_d}$ included in \Message{Certify}{\PC{m_d}} to prepare $m$ and execute $\T$ using \lsfref{fig:poe_nc}{store}{inform}. We note that the linear version of the prepare phase relies on a single replica, the primary, to aggregate \MName{Support} messages into a prepared certificate. Primaries will only fail to do so if they are faulty or communication is unreliable. Hence, the normal view-change protocol can deal with failures of this linear version of the prepare phase. An illustration of the working of this linear version of the normal-case protocol can be found in Figure~\ref{fig:poe_lin_nc_ill}.

\begin{figure}[t!]
    \centering
    \begin{tikzpicture}[yscale=0.5,xscale=0.75]
        \draw[thick,draw=black!75] (1.75,   0) edge ++(6.5, 0)
                                   (1.75,   1) edge ++(6.5, 0)
                                   (1.75,   2) edge ++(6.5, 0)
                                   (1.75,   3) edge[blue!50!black!90] ++(6.5, 0);

        \draw[thin,draw=black!75] (2, 0) edge ++(0, 3)
                                  (4, 0) edge ++(0, 3)
                                  (6, 0) edge ++(0, 3)
                                  (8, 0) edge ++(0, 3);

        \node[left] at (1.8, 0) {$\Replica_3$};
        \node[left] at (1.8, 1) {$\Replica_2$};
        \node[left] at (1.8, 2) {$\Replica_1$};
        \node[left] at (1.8, 3) {$\Replica[p]$};

        \path[->] (2, 3) edge (4, 2)
                         edge (4, 1)
                         edge (4, 0)
                           
                  (4, 2) edge (6, 3)
                  (4, 1) edge (6, 3)
                  
                  (6, 3) edge (8, 2)
                         edge (8, 1)
                         edge (8, 0);
                                 
        \node[dot,colA] at (8, 0) {};
        \node[dot,colA] at (8, 1) {};
        \node[dot,colA] at (8, 2) {};
        \node[dot,colA] at (8, 3) {};
        
        \path (8, 3) edge[thick,colA] (8, -0.8);
        \node[label,below right,align=left] at (8, 0) {Execute $\T$\\\phantom{Decide $\T$}};

        \node[label,below,yshift=3pt] at (3, 0) {\vphantom{Execute $\T$}\MName{Propose}};
        \node[label,below,yshift=3pt] at (5, 0) {\vphantom{Execute $\T$}\MName{Support}};
        \node[label,below,yshift=3pt] at (7, 0) {\vphantom{Execute $\T$}\MName{Certify}};
    \end{tikzpicture}
    \caption{A schematic representation of the \emph{linear} normal-case protocol of \PoE{}: the primary $\Replica[p]$ proposes transaction $\T$ to all replicas via a \MName{Propose} message $m$. Next, all other replicas respond with \MName{Support} message holding a signature share for a message of the form \Message{Prepare}{m}. The primary combines $\nf$ of these signature shares to construct a constant-sized prepared certificate $\PC{m}$ holding a constant-sized threshold signature. Finally, the primary sends the prepared certificate $\PC{m}$ to all replicas via a \MName{Certify} message, after which replicas can execute $\T$. In this example, replica $\Replica_3$ is faulty and does not participate.}\label{fig:poe_lin_nc_ill}
\end{figure}
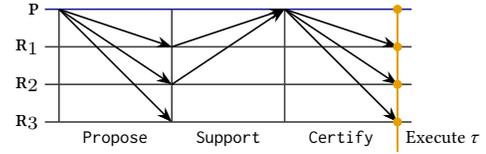

Next, we look at the check-commit phase. We reiterate that the check-commit phase operates rather differently than the prepare phase (and than the commit phases in most \PBFT{}-style primary-backup consensus protocols). Indeed, the \emph{decentralized} design of the check-commit protocol of Figure~\ref{fig:poe_cc} is crucial for its correctness: the check-commit protocol of Figure~\ref{fig:poe_cc} guarantees that any replicas that are left in the dark by the primary (without the primary disrupting the progress of the normal-case protocol) will be able to recover any missing proposals via the check-commit protocol (this independent of the behavior of any faulty replicas).

Consequently, we cannot simply make the check-commit protocol linear by assigning a fixed replica $\Replica[q]$ that aggregates signature shares of the form $\PCert{\Message{Commit}{m_d}}{\Replica}$ into commit certificates $\CC{m_d}$: replica $\Replica[q]$ can be faulty and will have the power to keep individual replicas in the dark, thereby breaking a main correctness guarantee of the check-commit protocol. 

To make the ckeck-commit protocol linear, we employ three techniques:
\begin{description}
\item[aggregator rotation.] In round $\rn$, the replica $\Replica[q]$ with $\ID{\Replica[q]} = \rn \bmod \n$ will be the aggregator that receives signature shares. Hence, every round has a different aggregator, assuring that $\nf$ out of $\n$ consecutive rounds have a non-faulty aggregator (when communication is reliable) and can succeed even with interference of other replicas.

\item[aggregated multi-round check-commits.] As the previous rounds can have faulty aggregators, due to which these rounds did not properly finish their check-commit steps, the aggregator $\Replica[q]$ of round $\rn$ will perform a \emph{multi-round} check-commit for round $\rn$ and the preceding $\n-1$ rounds (hence, a multi-round check-commit for all rounds \emph{since} the previous round for which $\Replica[q]$ was the aggregator). In this way, the correct behavior of  $\Replica[q]$ can guarantee that it can always perform a successful check-commit when it is the aggregator. To assure that the aggregated check-commit messages for round $\rn$ \emph{can} have a constant size, we do not check-commit on an individual  \MName{Propose} message, but on the digest $\Digest{\PC{m_{d,1}}, \dots, \PC{m_{d,\n}}}$ of $\n$ propose certificates, one for each \MName{Propose} message of the last $\n$ rounds. (For simplicity, we assume that all replicas have agreed on default propose certificates and proposals for any rounds before the first round).

\item[recovery certificates.] The aggregator $\Replica[q]$ of round $\rn$ is only guaranteed to be able to construct a commit certificate for $\Digest{\PC{m_{d,1}}, \dots, \PC{m_{d,\n}}}$ if \emph{all} non-faulty replicas have a prepared certificates for the last $\n$ rounds. A core functionality of the check-commit protocol is to provide such prepared certificates to non-faulty replicas that are missing them (if at-least $\f+1$ non-faulty replicas obtained such prepared certificates, as otherwise the view-change protocol will take care of recovery). To be able to provide such recovery, the linear check-commit protocol will employ a recovery step that uses a $\n:\f+1$ threshold signature scheme to produce \emph{recovery certificates} that can be used by non-faulty replicas to reliably obtain any missing prepared certificates. This $\n:\f+1$ threshold signature scheme is distinct from the normal $\n:\nf$ threshold signature scheme we use to generate prepared and commit certificates and we write $\RCert{v}{\Replica}$ to denote a \emph{signature share} produced by $\Replica$ using this $\n:\f+1$ threshold signature scheme.
\end{description}

The pseudo-code for this \emph{linear} check-commit protocol can be found in Figure~\ref{fig:lin_poe_cc} and an illustration of the working of this linear version of the check-commit protocol can be found in Figure~\ref{fig:lin_poe_cc_ill}.

\begin{figure}[t]
    \begin{myprotocol}
        \TITLE{Check-commit role}{running at every replica $\Replica \in \Replicas$}
        \EVENT{$\Replica$ prepared $m = \Message{Propose}{\Cert{\T}{\Client}, \View, \rn}$ and executed $\T$}
            \STATE Wait until all rounds up-to round $\rn - \n$ have a commit certificate.\label{fig:lin_poe_cc:wait}
            \IF{$\View$ is the current view}
                \STATE Let $D = \Digest{\PC{m_{d,1}}, \dots, \PC{m_{d,\n}}}$ be the digest for the prepared certificates stored in the last $\n$ rounds (the rounds $\rn-(\n-1),\dots,\rn$).
                \STATE Send \Message{SupportCC}{\rn, D, \PCert{D}{\Replica}, \RCert{D}{\Replica}} to the aggregator $\Replica[q]$ of round $\rn$ (replica $\Replica[q]$ with $\ID{\Replica[q]} = \rn \bmod \n$).
                .\label{fig:lin_poe_cc:broad}
            \ENDIF
        \ENDEVENT
        \EVENT{$\Replica$ with $\ID{\Replica} = \rn \bmod \n$ receives $\f+1$ well-formed messages\\
        \phantom{\textbf{event }}$m_i = \Message{SupportCC}{\rn, D, \PCert{D}{\Replica_i}, \RCert{D}{\Replica_i}}$, $1 \leq i \leq \f+1$, from\\
        \phantom{\textbf{event }}$\f+1$ distinct replicas}
            \STATE Broadcast \Message{RecoveryCC}{\rn, D, \RC{D}} to all replicas, in which $\RC{D}$ is the \emph{recovery certificate} constructed using the signature shares $\RCert{D}{\Replica_i}$, $1 \leq i \leq \f+1$.
        \ENDEVENT
        \EVENT{$\Replica$ receives a message $m_r =\Message{RecoveryCC}{\rn, D, \RC{D}}$ from \\
        \phantom{\textbf{event }}aggregator $\Replica[q]$ of round $\rn$ ($\Replica[q]$ with $\ID{\Replica[q]} = \rn \bmod \n$)}
            \IF{$\Replica$ cannot construct $D$ with its local prepared certificates}
                \STATE Query $\Replica[q]$ for the missing prepared certificates. If $\Replica[q]$ does not have these, then $\Replica[q]$ can query the $\f+1$ distinct replicas (of which at-least one is non-faulty) from which $\Replica[q]$ received the \MName{SupportCC} messages used to construct $m_r$ for the prepared certificates they used to construct $D$.
            \ENDIF
        \ENDEVENT
        \EVENT{$\Replica$ with $\ID{\Replica} = \rn \bmod \n$ receives $\nf$ well-formed messages\\
        \phantom{\textbf{event }}$m_i = \Message{SupportCC}{\rn, D, \PCert{D}{\Replica_i}, \RCert{D}{\Replica_i}}$, $1 \leq i \leq \nf$, from\\
        \phantom{\textbf{event }}$\nf$ distinct replicas}
            \STATE Broadcast \Message{CertifyCC}{\rn, D, \CC{\rn}} to all replicas, in which $\CC{\rn}$ is the \emph{commit certificate} constructed using the signature shares $\PCert{D}{\Replica_i}$, $1 \leq i \leq \nf$.
        \ENDEVENT
        \EVENT{$\Replica$ receives a message $m_c =\Message{CertifyCC}{\rn, D, \CC{\rn}}$ from \\
        \phantom{\textbf{event }}aggregator $\Replica[q]$ of round $\rn$ ($\Replica[q]$ with $\ID{\Replica[q]} = \rn \bmod \n$)}
            \STATE Wait until all rounds up-to round $\rn-\n$ have a commit certificate.\label{fig:lin_poe_cc:wait_store}
            \STATE Store \emph{commit certificate} $\CC{\rn}$ for round $\rn$, which can also be used as a commit certificate for the $\n-1$ rounds preceding round $\rn$.\label{fig:lin_poe_cc:store}
        \ENDEVENT
    \end{myprotocol}
    \caption{The linear check-commit protocol in \PoE{}.}\label{fig:lin_poe_cc}
\end{figure}

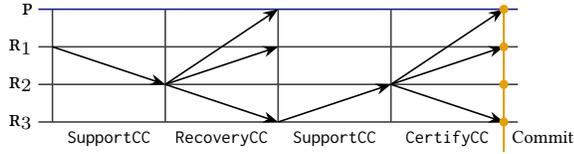
\begin{figure}[t!]
    \centering
    \begin{tikzpicture}[yscale=0.5,xscale=0.75]
        \draw[thick,draw=black!75] (1.75,   0) edge ++(8.5, 0)
                                   (1.75,   1) edge ++(8.5, 0)
                                   (1.75,   2) edge ++(8.5, 0)
                                   (1.75,   3) edge[blue!50!black!90] ++(8.5, 0);

        \draw[thin,draw=black!75] (2, 0) edge ++(0, 3)
                                  (4, 0) edge ++(0, 3)
                                  (6, 0) edge ++(0, 3)
                                  (8, 0) edge ++(0, 3)
                                  (10, 0) edge ++(0, 3);

        \node[left] at (1.8, 0) {$\Replica_3$};
        \node[left] at (1.8, 1) {$\Replica_2$};
        \node[left] at (1.8, 2) {$\Replica_1$};
        \node[left] at (1.8, 3) {$\Replica[p]$};

        \path[->]
                  (2, 2) edge (4, 1)
                 
                  (4, 1) edge (6, 3)
                         edge (6, 2)
                         edge (6, 0)
                         
                  (6, 0) edge (8, 1)
                  
                  (8, 1) edge (10, 3)
                         edge (10, 2)
                         edge (10, 0);

        \node[dot,colA] at (10, 0) {};
        \node[dot,colA] at (10, 1) {};
        \node[dot,colA] at (10, 2) {};
        \node[dot,colA] at (10, 3) {};
        
        \path (10, 3) edge[thick,colA] (10, -0.8);
        \node[label,below right,align=left] at (10, 0) {Commit\\\phantom{Decide $\T$}};

        \node[label,below,yshift=3pt] at (3, 0) {\vphantom{Execute $\T$}\MName{SupportCC}};
        \node[label,below,yshift=3pt] at (5, 0) {\vphantom{Execute $\T$}\MName{RecoveryCC}};
        \node[label,below,yshift=3pt] at (7, 0) {\vphantom{Execute $\T$}\MName{SupportCC}};
        \node[label,below,yshift=3pt] at (9, 0) {\vphantom{Execute $\T$}\MName{CertifyCC}};
    \end{tikzpicture}
    \caption{A schematic representation of the \emph{linear} check-commit protocol of \PoE{}: in round $\rn = 6$, all replicas send a \MName{SupportCC} message holding a signature share for the digest $D$ of the last $\n$ proposals to the aggregator of round $\rn$ (replica $\Replica[q]$ with $\ID{\Replica[q]} = \rn \bmod \n$). In this case, all replicas send to replica $\Replica_2$. The aggregator $\Replica[q]$ uses the first $\f+1$ of these messages to construct a recovery certificate, which can be used for replicas in the dark to recover missing rounds. Next, the aggregator $\Replica[q]$ uses $\nf$ of these messages to construct a \emph{commit certificate} $\CC{\rn}$ for rounds $\rn-(\n-1), \dots, \rn$. Finally, the aggregator $\Replica[q]$ sends the commit certificate $\CC{\rn}$ to all replicas via a \MName{CertifyCC} message, after which replicas can commit $m_{\rn-\n+1},\dots, m_{\rn}$. In this example, replica $\Replica_3$ is left in the dark by a faulty primary and can only participate \emph{after} receiving a recovery certificate.}\label{fig:lin_poe_cc_ill}
\end{figure}

\begin{remark}\label{rem:failure}
We note that the failure detection stage still exchanges a quadratic number of \MName{Failure} messages. As shown in Lemma~\ref{lem:sync}, \PoE{} depends on these all-to-all messages to \emph{synchronize} replicas after periods of unreliable communication. We do not believe we can fully eliminate such quadratic communication when operating in an asynchronous environment in which communication has unreliable periods in which messages are \emph{arbitrary delayed or lost} without making additional assumptions (e.g., by assuming reliable or synchronous communication, or by assuming external services to detect and deal with failures).
\end{remark}

We complete our treatment of Linear-\PoE{} with a caution that applies to \emph{all} consensus protocols that use threshold signatures. First, the usage of threshold signatures will increase the latency between receiving client requests and executing client request. In \PoE{}, this increase stems from the increase from two communication rounds to three communication rounds. Furthermore, even though the usage of threshold signatures does sharply reduces the \emph{overall} communication complexity of consensus (the total number of messages sent between all replicas), this overall reduction does not imply a significant reduction of the communication complexity at the level of individual replicas, however, as illustrated next:
\begin{example}\label{ex:ts}
Consider a deployment of \PoE{} with $\n = 31$ replicas. We assume that \MName{Propose} messages are larger than all other messages, which is typically the case due to the usage of digests (Section~\ref{ss:fine}). Consequently, the primary, whom broadcasts \MName{Propose} messages to all other replicas, has a much higher bandwidth usage than all other replicas. 

Now consider a deployment in which \MName{Propose} messages with client requests have size $\CSize = \SI{10}{\kibi\byte}$ and in which all other messages have a size of $\MSize = \SI{256}{\byte}$. Without using threshold signatures, the primary will send $\n-1$ \MName{Propose} messages, receive $\n-1$ \MName{Prepare} messages, and send and receive $\n-1$ \MName{CheckCommit} messages per consensus decision. Hence, in this case, a single consensus decision costs $(\n-1)(\CSize + 3\MSize) = \SI{322}{\kibi\byte}$ of bandwidth for the primary. If we switch to using threshold signatures, then the primary will still send $\n-1$ \MName{Propose} messages, will receive $\n-1$ \MName{Support} message, and send $\n-1$ \MName{Certify} messages. Furthermore, once every $\n$ rounds the primary will send and receive $3(\n-1)$ messages associated with the linear check-commit protocol. Hence, in this case, a single consensus decision still costs $(\n-1)(\CSize + 2\MSize + \frac{3\MSize}{\n}) \approx \SI{315}{\kibi\byte}$ of bandwidth for the primary, a reduction of only $2\%$.
\end{example}

Due to Example~\ref{ex:ts}, the usage of Linear-\PoE{} over \PoE{} in deployments in which the bandwidth at the primary is the bottleneck for performance will only yield a minor improvement in throughput, this at the cost of a roughly-$\frac{1}{3}$-th increase in latency.

\subsection{\MAC{}-based message authentication}\label{ss:mac}

Up till now, we have presented the design of \PoE{} using digital signatures and threshold signatures, two powerful forms of asymmetric cryptography that provide \emph{strong message authentication}: messages signed by some replica $\Replica$ using either digital signatures or threshold signatures can be safely \emph{forwarded} by replicas, while faulty replicas are unable to forge or tamper with such signed messages. This strong form of message authentication is used throughout the design of \PoE{}, especially within the view-change protocol that relies on message forwarding.

Unfortunately, the usage of asymmetric cryptography for message authentication comes at a high computational cost~\cite{icdcs,mc_2021}. As an alternative, one can consider using message authentication codes (\MAC{}s)~\cite{cryptobook}, which are based on symmetric cryptography, to provide message authentication with much lower costs. Unfortunately, \MAC{}s only provide a \emph{weaker form} of message authentication: \MAC{}s can only be used to verify the sender of messages (if the sender is non-faulty), as \MAC{}s prevent faulty replicas from impersonating non-faulty replicas. Hence, \MAC{}s do not protect messages against tampering when they are forwarded. 

Still, it is well-known that \PBFT{}-style consensus protocols can be built using \MAC{}s only~\cite{pbftj,icdcs,mc_2021}. This is also the case for \PoE{}: when using message digests (as outlined in Section~\ref{ss:fine}), neither the normal-case protocol, nor the check-commit protocol, nor the failure detection stage of the view-change protocol rely on message forwarding, as each of these protocols only rely on counting the number of senders of valid messages. Hence, to assure that \PoE{} can operate using only \MAC{}s, one only needs to redesign the new-view proposal stage and new-view accept stage of the view-change protocol.  Such a redesign is possible in an analogous way as the redesign of \PBFT{} to a \MAC{}-based version~\cite{pbftj}.

\subsection{Out-of-Order Processing}

All variants of \PoE{} support out-of-order processing of consensus decisions, which is crucial for providing high consensus throughput in environments with high message delays such as wide-area (Internet) deployments.

In out-of-order consensus protocols such as \PoE{}, the primary can freely propose requests for future rounds even if the current consensus decision is not yet finalized (as long as those future rounds lay in the current window determined by the window size). By doing so, the primary can fully utilize its outgoing bandwidth to replicate future requests, instead of waiting for other replicas to finish their consensus steps. The following example illustrates the impact of out-of-order processing:

\begin{example}\label{ex:stimate}
Consider a deployment with $\n = 31$ replicas, with a message delay of $\delta = \SI{15}{\milli\second}$ (e.g., replicas are distributed over datacenters in a country), and in which all replicas have an outgoing bandwidth of $\BW = \SI{1}{\giga\bit\per\second}$. In this environment, a single \PoE{} consensus round takes at-least 3 consecutive messages (\MName{Propose}, \MName{Prepare}, and \MName{CheckCommit}) and takes at-least $3\delta = \SI{45}{\milli\second}$. Hence, if consensus processing is not out-of-order (sequential), then this deployment of \PoE{} will only have a throughput of less-than $\frac{\SI{1}{\second}}{3\delta} \approx 22$ consensus decisions per second.

As in Example~\ref{ex:ts}, we consider a deployment in which \MName{Propose} messages with client requests have size $\CSize = \SI{10}{\kilo\byte}$ and in which all other messages have a size of $\MSize = \SI{256}{\byte}$. As noted in Example~\ref{ex:ts}, the primary has the highest bandwidth usage in \PoE{} and in this deployment, each consensus decision will cost $(\n-1)(\CSize + 3\MSize)$. Hence, if consensus processing is out-of-order, then this deployment of \PoE{} can have a throughput of $\frac{\BW}{(\n-1)(\CSize + 3\MSize)} \approx 3028$ consensus decisions per second, which is two orders of magnitude larger than the sequential approach.
\end{example}

\section{Analytical Evaluation}\label{sec:anal}

In Section~\ref{sec:complexity}, we analyzed in-depth the cost of consensus using the \PoE{} consensus protocol and its variants. Next, we compare these costs with the costs of existing and frequently-used consensus protocols. A summary of our comparison can be found in Figure~\ref{fig:anal}.

\subsection{The Baseline of Comparison: \PBFT{}}
The \emph{Practical Byzantine Fault Tolerance} consensus protocol~\cite{pbftj} was introduced two decades ago and to this day is a baseline for providing high-performance consensus in practical environments. \PBFT{} is highly resilient and can even deal with network failures: although network failures can temporarily disrupt \emph{new} consensus decisions in \PBFT{}, \PBFT{} is able to automatically recover its operations once the network becomes reliable, this without ever loosing any previously-made consensus decisions.  Furthermore, recent works have shown that highly-optimized and fine-tuned implementations of this protocol can achieve throughputs surpassing more modern protocols in moderately-sized deployments~\cite{mc_2021,icdcs}. 

We have already provided a high-level description of the working of \PBFT{} in Example~\ref{ex:pbft}. As mentioned in Section~\ref{sec:poe}, \PoE{} shares the primary-backup design of \PBFT{} and, as proven in Section~\ref{ss:poe_proofs}, \PoE{} shares the high resilience of \PBFT{}. The main differences between \PBFT{} and \PoE{} can be summarized as follows:
\begin{enumerate}
\item \PoE{} utilizes \emph{speculative execution}, due to which replicas can execute transactions and inform clients directly after the \emph{prepare phase}, whereas \PBFT{} only executes client requests after their \emph{commit phase}. Due to this, \PoE{} can inform clients within only \emph{two communication rounds}, while \PBFT{} requires three communication rounds before it can inform clients.

\item \PoE{} utilizes the \emph{check-commit protocol}, a decentralized single phase protocol that servers the same roles as the \emph{commit phase} of \PBFT{}, which requires a round of communication, and the checkpoint protocol of \PBFT{}, which requires another round of communication. Due to this, the normal case of \PoE{} only takes three rounds of communication (of which only two are all-to-all), whereas the normal case of \PBFT{} requires four rounds of communication (of which three are all-to-all).
\end{enumerate}
The view-change protocols of \PoE{} and \PBFT{} are comparable in communication costs: the main difference between these view-change protocols is that view-changes  in \PoE{} need to account for speculative execution (e.g., perform rollbacks), but this accounting is done while determining the state represented by new view proposals and does not impose additional communication costs.

Due to the above analysis, we can conclude that \PoE{} will \emph{outperform} \PBFT{} in all situations, as \PoE{} lowers communication costs \emph{in all cases} (by eliminating one round of all-to-all communication), while also potentially reducing client latencies due to speculative execution.\

\subsection{Optimistic Protocols: \ZZ{}}
Several consensus protocols have attempted to reduce the communication cost of \PBFT{} via \emph{optimistic consensus}~\cite{zyzzyvaj,fabj,bft700j}. In an optimistic consensus protocol, the normal-case consensus protocol consists of a \emph{fast path} that will succeed under optimal conditions (no faulty behavior or unreliable communication) and a \emph{slow path} to deal with non-optimal conditions.

An example of such optimistic consensus is \ZZ{}~\cite{zyzzyvaj}: in \ZZ{}, all replicas execute directly after they receive a proposal of the primary, directly inform the client, and proceed with the next round of consensus. Hence, in \ZZ{} the optimal-case cost of consensus are minimal: only one round of primary-to-backup communication. We note that this \emph{fast path} is not able to detect failures between replicas: \ZZ{} requires clients to inform replicas of any failures, after which replicas can enter a \emph{slow path} to deal with these failures.

Unfortunately, consensus protocols with \emph{optimistic} fast paths such as \ZZ{}~\cite{zyzzyvaj} and \FaB{}~\cite{fabj} have over time shown vulnerabilities to faulty behavior, this especially in the presence of unreliable communication~\cite{zfail,zfailfix}. 

We note that speculative execution, as used by \PoE{}, and optimistic execution, e.g., as used by \ZZ{}~\cite{zyzzyvaj} and \FaB{}~\cite{fabj}, are not the same: as part of the normal-case of \PoE{}, \PoE{} will internally detect and correct any replica failures, this without any assumptions on correct behavior by any replicas or clients.

\subsection{Consensus with Threshold Signatures: \SBFT{}}
Several recent \PBFT{}-style consensus protocols have explored the usage of \emph{threshold signatures} to transform the two phases of all-to-all communication in \PBFT{} (the \emph{prepare} and \emph{commit} phases) to two phases of all-to-one-to-all communication. These transformations are similar to how we can transform the prepare phase of \PoE{} from all-to-all communication to an all-to-one support sub-phase and an one-to-all certify sub-phase. We note that such a transformation, when applied on the two all-to-all communication phases of \PBFT{}, will result in a consensus protocol with a high latency: in such a protocol, it will take five communication rounds before replicas can execute client requests and inform clients.  Although such a transformation can successfully reduce the global communication cost of the normal-case operations of the consensus protocol, such threshold signature transformations do little to address the costs associated with any checkpoint and view-change protocols.

A good examples of a \PBFT{}-style consensus protocols that use threshold signatures is \SBFT{}~\cite{sbft}, which uses an optimistic fast path to reduce the number of rounds when all replicas are non-faulty, uses threshold signatures to eliminate all-to-all communication, and uses threshold signatures to reduce the number of messages send to the client. The design of \SBFT{} uses a fast path which starts execution of client requests after the prepare phase \emph{if} no replicas are faulty. Furthermore, \SBFT{} can aggregrate the a proof of the execution results and send such proof in a single message to the client (instead of $\f+1$ messages), this to reduce communication costs towards the client. With this optimization, \SBFT{} is able to inform clients in four rounds of communication (which is one more round than \PBFT{} and Linear-\PoE{} and two more rounds than \PoE{}). In case the fast path fails, \SBFT{} falls back to a \emph{slow path} via a Linear-\PBFT{} implementation. To reduce the load on the primary (see Example~\ref{ex:ts}), \SBFT{} uses non-primary replicas as the aggregator that constructs prepared and commit certificates during the prepare and commit phases. The usages of non-primary aggregators at these parts of the protocol does introduce additional failure cases, however, for which \SBFT{} introduces separate recovery mechanisms. Similar fine-tuning can also be applied to \PoE{}, but we have not explored such fine-tuning in this work (as separate aggregators only introduce minimal bandwidth savings for the primary). Besides the fast path and the slow path, \SBFT{} also requires a checkpoint protocol similar to the one utilized by \PBFT{} (which can be run periodically).

\subsection{Chained consensus: \HS{}} Another approach toward utilizing threshold signatures in primary-backup consensus protocols is provided by \HS{}. \HS{} provides a clean-slate consensus design that is tuned toward minimizing complexity and communication cost, both during normal-case operations and during view-changes. To achieve this, \HS{} relies on \emph{chaining consensus}: the $i$-th consensus proposal builds upon the preceding ($i-1$)-th consensus proposal. This allows \HS{} to represent the state of the ledger via a single value, namely the last-made consensus decision (that builds upon all preceding decisions).  Finally, \HS{} uses threshold signatures to produce constant-size certificates for each consensus decision. This combination of techniques allows \HS{} to implement cheap \emph{primary rotation}: each round starts with a switch of primary via a constant-sized single-message view-change.

The design of \HS{} requires 4 consecutive all-to-one-to-all phases of communication before consensus is reached on a single request, which leads to 7 rounds of communication between the initial proposal of a client request and replicas being able to execute these requests. Both the normal-case and the view-change of \HS{} are \emph{linear}. To deal with unresponsive replicas and replica failures, \HS{} uses a Pacemaker. Unfortunately, the standard Pacemaker of \HS{} assumes \emph{partial synchrony} and cannot recover from network failures. To improve the resilience against network failure in \HS{}, one can replace the standard Pacemaker that \HS{} uses with a Pacemaker suitable for an asynchronous environment. As stated in Remark~\ref{rem:failure}, we believe that any asynchronous Pacemaker that can sufficiently synchronize replicas after network failure will have to operate similarly to the \emph{failure detection stage} of \PoE{}, which would introduce a decentralized all-to-all communication phase in the recovery path of \HS{}.

Due to the chained design of \HS{}, \HS{} does not support out-of-order processing: consensus decisions are strictly made in sequence. \HS{} does support overlapping of rounds of consecutive consensus decisions, however. Hence, in practice, \HS{} is able to propose a request every 2 rounds of communication. Consequently, the performance of typical deployments of \HS{} are \emph{latency based} and non-local deployments can only reach tens-to-hundreds consensus decisions per second. E.g., with a message delay of $\SI{15}{\milli\second}$, \HS{} can perform at-most $33$ consensus decisions per second, whereas, as shown in Example~\ref{ex:stimate}, an out-of-order \PoE{} can easily process thousands of consensus decisions per second.

\subsection{Trusted Hardware}
There is a large body of work on consensus protocols that utilize \emph{trusted hardware} to simplify and optimize consensus. The usage of trusted hardware restricts the behavior allowed by faulty or malicious replicas, e.g., by assuring faulty replicas cannot share their private keys and cannot forge round numbers and other counters. A representative example of a consensus protocol that utilizes trusted hardware is \MinBFT{}~\cite{minbft}. Due to the usage of trusted hardware, \MinBFT{} can skip the prepare phase of \PBFT{} and can achieve consensus in only two rounds of communication. Furthermore, the usage of trusted hardware makes \MinBFT{} more resilient against failure, as it requires $\f$ fewer non-faulty replicas than normal consensus protocols (\MinBFT{} can operate even in deployments in which $3\f \geq \n > 2\f$ holds). Unfortunately, the presented version of \MinBFT{} does require \emph{reliable communication}, as the protocol itself does not deal with message loss.

To the best of our knowledge, \MinBFT{} and other consensus protocols that utilize trusted hardware are the only consensus protocols \emph{besides} \PoE{} that have a latency of two rounds of communication (without relying on a fault-prone optimistic path).

\subsection{Other Consensus Protocols}
There are many other recent \PBFT{}-style consensus protocols that we did not cover in the above. Most of these protocols have properties similar to the ones discussed, however. E.g., \FaB{}~\cite{fabj} is an optimistic protocol that has a fast path in the same vein as \ZZ{}~\cite{zyzzyvaj} and \LinBFT{}~\cite{linbft} is a version of \PBFT{} that uses threshold signatures and is based on techniques similar to those used in \SBFT{}~\cite{sbft} and \HS{}~\cite{hotstuff}.

\section{Conclusion}\label{sec:conclusion}
In this paper, we introduced \PoE{}, a consensus protocol designed for high-performance low-latency resilient data management systems and that can operate in practical asynchronous deployments. \PoE{} introduces the usage of speculative execution and \emph{proof-of-execution} to minimize the latency of transaction processing in such resilient systems and introduces a single-round check-commit protocol to further reduce communication costs of consensus. Furthermore, the flexible design of \PoE{} allows for optimizations that further balance communication costs, transaction latency, and recovery complexity. 

The flexible low-latency high-throughput design of \PoE{} is especially suited for resilient data management systems. To illustrate this, we performed an in-depth analytical and experimental comparison with other consensus protocols, that underlined the outstanding performance of \PoE{}. Hence, we believe that \PoE{} is a promising step towards flexible general-purpose resilient data management systems.

\bibliographystyle{plainurl}
\bibliography{resources}

\end{document}